\documentclass{amsart}

\usepackage{xspace}
\usepackage[latin1]{inputenc}
\usepackage{graphicx}
\usepackage{pstricks}
\usepackage{amssymb}
\usepackage{mathrsfs}
\usepackage[active]{srcltx}
\usepackage{url}
\usepackage{hyperref}
\usepackage{amsmath,amstext,amsxtra,amsgen,amsbsy,amsopn,amscd,amsthm,amsfonts}
\usepackage{latexsym}

\usepackage{bookmark}

\newtheorem{theorem}{Theorem}
\newtheorem{proposition}{Proposition}
\newtheorem{lemma}[proposition]{Lemma}
\newtheorem{corollary}[proposition]{Corollary}

\theoremstyle{definition}

\theoremstyle{remark}

\newtheorem{remark}[proposition]{Remark}

\numberwithin{equation}{section}

\newcommand\R{{\ensuremath {\mathbb R} }}
\newcommand\C{{\ensuremath {\mathbb C} }}
\newcommand\N{{\ensuremath {\mathbb N} }}

\renewcommand\phi{\varphi}
\renewcommand\iint{\int \hspace{-0.3cm} \int}
\renewcommand\le{\leqslant}
\renewcommand\ge{\geqslant}
\renewcommand\epsilon{\varepsilon}
\renewcommand\hat{\widehat}
\renewcommand\tilde{\widetilde}
\renewcommand\bar{\overline}
\newcommand{\Dref}{D_{\textnormal{ref}}}
\newcommand{\gH}{\mathfrak{H}}

\newcommand{\gS}{\mathfrak{S}}

\newcommand{\cP}{\mathcal{P}}
\newcommand{\cD}{\mathcal{D}}
\newcommand{\cH}{\mathcal{H}}
\newcommand\ii{{\ensuremath {\infty}}}

\newcommand\projneg[1]{{\ensuremath{\chi_{(-\ii,\mu]}\left(#1\right)}}}
\newcommand\projnego[1]{{\ensuremath{\chi_{(-\ii,0]}\left(#1\right)}}}
\newcommand\projnegplus[1]{{\ensuremath{\chi_{(-\ii,\mu_+]}\left(#1\right)}}}
\newcommand\projnegminus[1]{{\ensuremath{\chi_{(-\ii,\mu_-]}\left(#1\right)}}}
\newcommand\projnegpm[1]{{\ensuremath{\chi_{(-\ii,\mu_\pm]}\left(#1\right)}}}
\newcommand{\Ebdf}[2]{{\ensuremath{\cE^{\rm{BDF}}_{#1}\left(#2\right)}}}

\renewcommand\d[1]{{\ensuremath{\,\text{d}#1}}}

\newcommand{\cE}{\mathcal{E}}

\newcommand{\F}{\mathcal{F}}

\newcommand{\cK}{\mathcal{K}}
\newcommand{\cM}{\mathcal{M}}

\newcommand{\cO}{\mathcal{O}}
\newcommand{\cC}{\mathcal{C}}
\newcommand{\cU}{\mathcal{U}}

\newcommand{\cF}{\mathcal{F}}
\newcommand{\cQ}{\mathcal{Q}}
\newcommand{\cX}{\mathcal{X}}
\newcommand{\cV}{\mathcal{V}}

\newcommand{\CJ}{\mathscr{C}}
\newcommand{\alp}{\boldsymbol{\alpha}}

\newcommand{\ind}{\boldsymbol{1}}

\newcommand{\Pm}{P^0_{\Lambda,-}}
\newcommand{\Pp}{P^0_{\Lambda,+}}
\newcommand{\Ppm}{P^0_{\Lambda,\pm}}
\newcommand{\Ql}{Q_{\textnormal{lin}}}
\newcommand{\nur}{{\nu_{\textnormal{ren}}}}
\newcommand{\rhoref}{{\rho_{\textnormal{ref}}}}
\newcommand{\rhorefpm}{{\rho_{\textnormal{ref},\pm}}}
\newcommand{\rhorefp}{{\rho_{\textnormal{ref},+}}}
\newcommand{\rhorefm}{{\rho_{\textnormal{ref},-}}}
\newcommand{\Qref}{{Q_{\textnormal{ref}}}}
\newcommand{\Qrefpm}{{Q_{\textnormal{ref},\pm}}}
\newcommand{\Qrefp}{{Q_{\textnormal{ref},+}}}
\newcommand{\Qrefm}{{Q_{\textnormal{ref},-}}}
\newcommand{\Pref}{{P_{\textnormal{ref}}}}
\newcommand{\Qlinpm}{{Q_{\textnormal{lin},\pm}}}
\newcommand{\Qlinp}{{Q_{\textnormal{lin},+}}}
\newcommand{\Qlinpr}{{Q_{\textnormal{lin},+,r}}}
\newcommand{\Qlinm}{{Q_{\textnormal{lin},-}}}

\DeclareMathOperator{\tr}{Tr}

\begin{document}
 
\title[Charge renormalization and pair production]{Charge renormalization and static electron/positron pair production for a nonlinear Dirac model with weak interactions}

\author{Julien Sabin}

\address{D\'epartement de Math\'ematiques, CNRS UMR 8088, Universit\'e de Cergy-Pontoise, 95000 Cergy-Pontoise, France}

\email{julien.sabin@u-cergy.fr}

\date{\today}

\begin{abstract}
The Hartree-Fock approximation of Quantum Electrodynamics provides a rigorous
framework for the description of relativistic electrons in external fields.
This nonlinear model takes into account the infinitely many virtual electrons
of Dirac's vacuum as well as the Coulomb interactions between all the
particles. The state of the system is an infinite-rank projection satisfying a
nonlinear equation. In this paper, we construct solutions to this equation, in
the regime of weak interactions (that is, small coupling constant $\alpha$),
and strong external fields (that is, large atomic charge $Z$ such that $\alpha
Z:=\kappa$ stays fixed). In this regime, we are able to remove the ultraviolet
cut-off $\Lambda$ as soon as $\alpha\log\Lambda$ stays fixed. As an
application of this result, we compare the critical strength
$\kappa_c(\alpha)$ of the external potential needed to produce an additional
particle in the vacuum, when $\alpha=0$ or $\alpha>0$. We prove that
$\lim_{\alpha\to0}\kappa_c(\alpha)/\kappa_c(0)> 1$, and we identify the limit
exactly. Because of the dielectric behavior of Dirac's vacuum, static
electron/positron pair production occurs in the interacting case for a stronger field that in the non-interacting case, which is a mere consequence of charge renormalization.
\end{abstract}

\maketitle 

\section*{Introduction}
 
The relativistic kinetic energy of a quantum electron is described by the Dirac operator $D^0=-i\alp\cdot\nabla+\beta$ acting on $L^2(\R^3,\C^4)$, where $\alp=(\alpha_j)_{j=1,2,3}$ and 
$$\alpha_j=\left(\begin{array}{cc}
         0 & \sigma_j \\
         \sigma_j & 0
        \end{array}\right), \quad\forall j=1,2,3,\qquad\beta=\left(\begin{array}{cc}
                                                                     \text{Id}_{\C^2} & 0\\
                                                                     0 & -\text{Id}_{\C^2}
                                                                    \end{array}\right),
$$
$$\sigma_1=\left(\begin{array}{cc}
                  0 & 1\\
                  1 & 0
                 \end{array}\right),\qquad
     \sigma_2=\left(\begin{array}{cc}
                     0 & -i\\
                     i & 0
                    \end{array}\right),\qquad
     \sigma_3=\left(\begin{array}{cc}
		      1 & 0\\
		      0 & -1
                    \end{array}\right).
$$
Here, the units are taken so that the reduced Planck's constant $\hbar$, the speed of light $c$, and the mass of an electron $m_e$, are all set to 1. The spectrum of the Dirac operator is $\sigma(D^0)=(-\ii,-1]\cup[1,+\ii)$. To explain why real electrons have a positive kinetic energy, Dirac suggested to postulate that, in nature, all the negative kinetic energy states of the Dirac operator are already occupied by virtual electrons.  Thus, the relativistic vacuum is composed of infinitely many charged, virtual particles, usually called the \emph{Dirac sea}. 

In an external electromagnetic field, the distribution of charge of the Dirac sea changes. This phenomenon is called \emph{vacuum polarization}.  In the same spirit, a negative kinetic energy virtual electron can turn into a positive kinetic energy real electron when given enough energy. This transformation results in a ``hole'' in the Dirac sea, which itself has a positive energy with respect to the fully filled Dirac sea. This hole can thus be observed and is identified with the anti-particle of the electron, the positron. Hence, such a process is labeled \emph{electron/positron pair production}. In this paper, we are interested in the case where this energy is brought to the vacuum by an external, classical electric field. This situation has been originally considered by Sauter\cite{Sauter-31}, Heisenberg-Euler \cite{HeiEul-36}, and Schwinger \cite{Schwinger-51a}. Since it is non-negligible compared to the other 
sources of pair production only for huge electric fields, it has not yet been observed experimentally. However, the creation of such a strong field could be achieved in the near future, due to recent progress in laser physics \cite{Dunne-09}. 

There exists two approaches to the problem of electron/positron pair production under an external electric potential $V$. The first one consists in studying the vacuum in its ground state, in the presence of $V$. One tries to determine if the vacuum then contains real particles or not, if $V$ is strong enough. In other words, for large $V$ one expects that it is energetically more favorable for the vacuum to contain real electrons or positrons, than to contain no real particle at all.  If the vacuum contains real particles, we say that the external potential $V$ produces pairs, or more precisely produces particles. Indeed, for this formulation of electron/positron pair production, called \emph{static} pair production, it is not necessary that the number of electrons created is the same as the number of positrons created. We still refer to it as pair production though, because one can think of it as a localized particle together with an anti-particle at infinity. The second approach consists in studying the time evolution of the vacuum with no electron and no positron (the \emph{free} vacuum), when the external field is compactly supported in time. One wants to determine if, for large times, some pairs still exist. In this case, notice that it is necessarily pairs that are produced since the total charge of the system is conserved. The second approach is thus called \emph{dynamical} pair production. While dynamical pair production may be more relevant from a physical point of view, it is also much harder to study mathematically. 

Only a few mathematical results are known about electron/positron pair production. Indeed, one has to deal with infinitely many relativistic particles, interacting with an external field and possibly with each other. The correct theory to describe such systems is Quantum Electrodynamics, which has only been formulated in a perturbative fashion so far. Even the apparently simpler many-body theory of $N$ relativistic electrons is not well-defined mathematically. For the time being, one has to make approximations to obtain results about such systems. 

The case where the interactions between the particles are neglected is one of these approximations. In this setting, everything reduces to the spectral theory of the one-body operator $D^0+V$, for which the issue of dealing with an infinite number of particles does not exist anymore. The problem of static pair production has been studied by Klaus and Scharf \cite{KlaSch-77a}. The question of dynamical pair production has been tackled in the non-interacting case by Pickl and D\"urr \cite{PicDur-08}, who proved that pairs were created for a certain class of potentials adiabatically evolving in time. This relies on a subtle analysis of the dispersive properties of the essential spectrum of the Dirac operator. 

When the interactions between the particles are taken into account, one can perform a mean-field approximation to again reduce the problem to the one-body space. The additional difficulty compared to the case without interactions is that the state of the system is not described by a wavefunction anymore, but rather by a \emph{one-body density matrix}, an infinite-rank operator on the one-body space (usually an orthogonal projection). The idea of performing a mean-field approximation of Quantum Electrodynamics goes back to Chaix and Iracane \cite{ChaIra-89}. The first rigorous models formulating these ideas were developed by Hainzl and Siedentop \cite{HaiSie-03}, and by Bach, Barbaroux, Helffer, and Siedentop \cite{BacBarHelSie-99}. Later on, another related model, which took into account vacuum polarization, was introduced by Hainzl, Lewin, and S\'er\'e in 2005 \cite{HaiLewSer-05a}, and extended in a series of articles by Gravejat, Hainzl, Lewin, S\'er\'e and Solovej \cite{HaiLewSer-05b,HaiLewSol-07,HaiLewSer-08,GraLewSer-09,GraLewSer-11,GraHaiLewSer-12}. In the present article, we will use the latter, called the \emph{Bogoliubov-Dirac-Fock} model. In particular, the polarized vacuum is constructed in a rigorous fashion as a ground state of a certain nonlinear energy functional. Static electron/positron pair production has already been studied by the author in this framework in a previous work \cite{Sabin-11}.

In the case without interactions, the one-body density matrix of the vacuum polarized by a potential $V$ is given by
$$P=\projnego{D^0+V}.$$
If $V$ is smooth enough, the eigenvalues of $D^0+\kappa V$ are continuous with respect to the parameter $\kappa$. Klaus and Scharf \cite{KlaSch-77a} proved (in the case $V\le0$) that a particle is created in the vacuum if one of these eigenvalues \emph{crosses 0} when increasing $\kappa$. If it crosses from the right to the left, it corresponds to the creation of an electron, while if it crosses from the left to the right, it corresponds to the creation of a positron. By the min-max principle of \cite{DolEstSer-00}, in the case $V\le0$ eigenvalues can only decrease so only electrons are created. This picture was confirmed by Hainzl \cite{Hainzl-04}, from the point of view of the charge density of the system. When interactions are turned on, the one-body density matrix of the polarized vacuum satisfies a nonlinear equation of the form 
\begin{equation}\label{eq:BDF-equation-intro-0}
P=\projnego{D^0+V+\alpha X_P}, 
\end{equation}
where $X_P$ is an operator depending on $P$ (see Equation \eqref{eq:BDF-equation} below) and $\alpha>0$ is a coupling constant controlling the size of the non-linearity (or, equivalently, the size of the interactions). The number $\alpha$ can also be interpreted as the square of the charge of an electron, $\alpha=e^2$, in which case it should also appear in front of the electric potential $V$. Since we are interested in pair production in strong fields, it is also natural to consider a fixed potential $V$ (interpreted as the shape of a nucleus, for instance), which strength is controlled by a parameter $Z>0$ (interpreted as the charge of the nucleus). The equation satisfied by $P$ is then
\begin{equation}\label{eq:BDF-equation-intro}
P=\projnego{D^0+(\alpha Z)V+\alpha X_P}, 
\end{equation}
As a consequence, studying \eqref{eq:BDF-equation-intro-0} for $V$ fixed and $\alpha>0$ small is the same as studying \eqref{eq:BDF-equation-intro} for $\alpha Z$ fixed and $\alpha>0$ small, which is a well-known regime in the physics literature \cite{MohPluSof-98,ReiGreAre-71,GreRei-77,Shabaev-02}. In the sequel, we will thus write $\kappa=\alpha Z$. The existence of a solution to \eqref{eq:BDF-equation-intro} is non trivial. Actually, one can show that there is \emph{no} projection on $L^2(\R^3,\C^4)$ satisfying \eqref{eq:BDF-equation-intro}. One way to circumvent this problem is to impose an ultraviolet cut-off $\Lambda>0$, and to solve instead
\begin{equation}\label{eq:BDF-equation-intro-2}
 P=\projnego{\Pi_\Lambda(D^0+\kappa V+\alpha X_P)\Pi_\Lambda}, 
\end{equation}
where $\Pi_\Lambda$ is the multiplication operator in Fourier space by $\ind_{B(0,\Lambda)}$. We will show (Theorem \ref{thm:thm1}) that for a general class of $V$, Equation \eqref{eq:BDF-equation-intro-2} has a unique solution for $\kappa$ fixed if $\alpha$ is small enough, $\Lambda$ is large enough, and $\alpha\log\Lambda$ fixed. Furthermore, we will show that this projection $P$ satisfies
$$P-\projnego{D^0+\frac{\kappa}{1+\frac{2}{3\pi}\alpha\log\Lambda}V}\longrightarrow0$$
in the Hilbert-Schmidt topology, as $\alpha\to0$ with $\alpha\log\Lambda$ and $\kappa$ fixed. We interpret this result as a \emph{first order charge renormalization} of our model, meaning that we can compute the screening induced by the polarized vacuum explicitly in terms of $\alpha$ and $\Lambda$. The term ``first order'' refers here to the fact that this is valid only in the limit $\alpha\to0$. Roughly speaking, if we put an external potential $V$ in the vacuum, what we observe in return (in the limit $\alpha\to0$) is a screened potential $V_{\rm{scr}}$, with
\begin{equation}\label{eq:V-renorm}
V_{\rm{scr}}=\frac{1}{1+\frac{2}{3\pi}\alpha\log\Lambda}V. 
\end{equation}
A similar formula was already obtained in \cite{GraLewSer-09}, and it was valid \emph{for all $\alpha>0$}. However, in \cite{GraLewSer-09}, the non-linearity $X_P$ is simpler than the one we consider in this article (namely, it does not contain the so-called \emph{exchange term}). Here, we obtain \eqref{eq:V-renorm} only in the limit $\alpha\to0$, but we manage to include the exchange term. 

Finally, we study static pair production for this model with interactions. Because of the nonlinearity $X_P$, it is not clear if the picture provided by Klaus and Scharf is still valid, that is if an eigenvalue of $D(\alpha,\kappa V):=\Pi_\Lambda(D^0+\kappa V+\alpha X_P)\Pi_\Lambda$ crosses 0 when $\kappa$ is increased, then a particle is created. Instead, we will compare two operators
$$P_+=\projnegplus{\Pi_\Lambda(D^0+\kappa V+\alpha X_{P_+})\Pi_\Lambda},$$
and 
$$P_-=\projnegminus{\Pi_\Lambda(D^0+\kappa V+\alpha X_{P_-})\Pi_\Lambda},$$
where $-1<\mu_-<0<\mu_+<1$. If the operator $D^0+\kappa V$ has an eigenvalue $\lambda(\kappa V)\in(\mu_-,\mu_+)$ and $\alpha$ is small enough, $P_-$ is interpreted as the polarized vacuum with a certain charge $q$, while $P_+$ is the polarized vacuum with charge $q+1$. We want to determine which one is the most advantageous energetically. To do so, we introduce the function
$$F(\kappa,\alpha)=\cE(P_+)-\cE(P_-),$$
where $\cE$ is the energy functional which will be defined in the next section. Our second result (Theorem \ref{thm:thm2}) provides an explicit formula valid for $\alpha$ small enough and $\alpha\log\Lambda$ fixed
$$F(\kappa,\alpha)=\lambda(\kappa V_{\rm{scr}})+\cO_{\alpha\to0}(\alpha),$$
where $V_{\rm{scr}}$ is given by \eqref{eq:V-renorm}. In particular, if $\lambda(\kappa V)$ is positive for $\kappa<\kappa_c$ and negative for $\kappa>\kappa_c$, then the polarized vacuum has charge $q$ for $\kappa<\kappa_c(\alpha)$ and charge $q+1$ for $\kappa>\kappa_c(\alpha)$: an electron is created; if $\lambda(\kappa V)$ is negative for $\kappa<\kappa_c$ and positive for $\kappa>\kappa_c$, then a positron is created. Furthermore, we have 
$$\kappa_c(\alpha)\longrightarrow\left(1+\frac{2}{3\pi}\alpha\log\Lambda\right)\kappa_c>\kappa_c$$
as $\alpha\to0$ with $\alpha\log\Lambda$ fixed. Hence, due to the screening of the polarized vacuum, static pair production in the interacting case requires a stronger potential than in the non-interacting case.

In principle, we can expand explicitly the function $F$ to any order. For instance, we can compute the next order of $F$ in $\alpha$, that is $\lim_{\alpha\to0}(F(\kappa,\alpha)-\lambda(\kappa V_{\rm{scr}}))/\alpha$. We are interested in the sign of this quantity, which would determine the growth of $F$ to the next order. Unfortunately, we were not able to determine its sign, due to the complicated form of this term. It remains an interesting open problem to compute this sign. Finally, let us mention \cite{Sok-12}, where the same regime $\alpha\to0$, $\alpha\log\Lambda$ fixed is studied without external field and for the vacuum of charge 1, interpreted as an electron together with the Dirac sea. 

The paper is organized as follows. In Section 1, we introduce the model and present our main results. The proofs are given in Section 2. 

\bigskip

\noindent\textbf{Acknowledgments.} I am indebted to Mathieu Lewin for his advices and his help. I also acknowledge support from the ERC MNIQS-258023 and from the ANR ``NoNAP'' (ANR-10-BLAN 0101) of the French ministry of research. 

\section{Main Results}

In this section, we first explain the setting of the Bogoliubov-Dirac-Fock (BDF) model, and we state our results later.  

\subsection{The Bogoliubov-Dirac-Fock model} The content of this section can be found in \cite{HaiLewSer-05a,HaiLewSer-05b, HaiLewSol-07,HaiLewSer-08,GraLewSer-09}. Let $\gH=L^2(\R^3,\C^4)$. For any $\Lambda>0$, we define
$$\gH_\Lambda:=\{f\in\gH,\quad\text{supp}\hat{f}\subset B(0,\Lambda)\},$$
where the hat denotes the Fourier transform of $f$, defined for $f$ integrable by 
$$\widehat{f}(\xi):=\frac{1}{(2\pi)^{3/2}}\int_{\R^3}f(x)e^{-ix\cdot\xi}\d{x},$$
and $B(0,\Lambda)$ denotes the ball of center 0 and radius $\Lambda$ in $\C^4$. We call such a $\Lambda$ a \emph{cut-off} in Fourier space. Since the Dirac operator $D^0$ is a multiplication operator by $p\mapsto\alp\cdot p+\beta$ in Fourier space, it stabilizes $\gH_\Lambda$ and its restriction to $\gH_\Lambda$ as an operator is denoted by $D^0_\Lambda$. Notice that $D^0_\Lambda$ is bounded and that $\|D^0_\Lambda\|_{\gH_\Lambda\to\gH_\Lambda}=\sqrt{1+\Lambda^2}$. We use the notation $\Pm=\chi_{(-\ii,0]}(D^0_\Lambda)$ (resp. $\Pp=\chi_{[0,+\ii)}(D^0_\Lambda)$) for the negative (resp. positive) spectral projection of the Dirac operator $D^0_\Lambda$. For any bounded operator $A$ over $\gH_\Lambda$ and any $\epsilon,\epsilon'=+,-$, we define the block matrices $A_{\epsilon\epsilon'}:=P^0_{\Lambda,\epsilon}AP^0_{\Lambda,\epsilon'}$. We also define the associated spectral subspaces $\gH_{\Lambda,\pm}:=P^0_{\Lambda,\pm}\gH_\Lambda$. For any separable Hilbert space $\cH$ and any $p>0$, we 
denote by $\gS_p(\cH)$ the Schatten class of all bounded operators $A$ over $\cH$ such that $\|A\|_{\gS_p}^p:=\tr|A|^p<+\ii$, where $|A|=A^*A$.  We then introduce the vector space
$$\gS_{1,\Pm}(\gH_\Lambda):=\left\{Q\in\gS_2(\gH_\Lambda),\quad Q_{++}, Q_{--}\in\gS_1(\gH_\Lambda)\right\}.$$
It is a Banach space endowed with the norm 
$$\|Q\|_{1,\Pm}:=\|Q_{++}\|_{\gS_1}+\|Q_{--}\|_{\gS_1}+\|Q_{+-}\|_{\gS_2}+\|Q_{-+}\|_{\gS_2}.$$
For any $Q\in\gS_{1,\Pm}$, we define its generalized trace by 
$$\tr_0(Q):=\tr(Q_{++}+Q_{--}).$$
Any $Q\in\gS_{1,\Pm}$ has an integral kernel $Q(\cdot,\cdot)\in L^2(\R^3\times\R^3,\cM_4(\C))$, where $\cM_4(\C)$ denotes the set of all $4\times 4$ complex matrices. Thanks to the cut-off in Fourier space, it has a unique smooth representative allowing to define the \emph{density of charge} $\rho_Q$ associated to $Q$, defined by $\rho_Q(x)=\tr_{\C^4}Q(x,x)$ for all $x\in\R^3$. It satisfies $\rho_Q\in L^2(\R^3)\cap\cC$, where $\cC$ denotes the \emph{Coulomb space}, 
$$\cC:=\{f,\quad\int_{\R^3}\frac{|\hat{f}(p)|^2}{|p|^2}\d{p}<+\ii\}=\dot{H}^{-1}(\R^3).$$
For any $f,g\in\cC$, the Coulomb inner product between $f$ and $g$ is defined as
$$D(f,g):=4\pi\int_{\R^3}\frac{\hat{f}(p)\overline{\hat{g}(p)}}{|p|^2}\d{p}.$$
If $f,g$ are smooth enough (e.g. $L^{6/5}\cap\cC$), then
$$D(f,g)=\int_{\R^3}\int_{\R^3}\frac{f(x)\bar{g(y)}}{|x-y|}\d{x}\d{y}.$$
In the BDF theory, a quantum state is described by a self-adjoint operator $P$ on $\gH_\Lambda$ such that $0\le P\le 1_{\gH_\Lambda}$ in the sense of operators. Such a $P$ is called a \emph{one-body density matrix}. For instance, the operator $\Pm$ represents the quantum state of the (free) Dirac sea. For any one-body density matrix $P$, we define its \emph{renormalized one-body density matrix} by $Q=P-\Pm$. The set of all admissible quantum states is then
$$\cK_\Lambda:=\{Q\in\gS_{1,\Pm},\quad Q=Q^*,\quad -\Pm\le Q\le1-\Pm\}.$$
For any $Q$ such that $-\Pm\le Q\le 1-\Pm$, we have $Q_{++}\ge0$, $-Q_{--}\ge0$ and the inequality
$$Q^2\le Q_{++}-Q_{--},$$
so that $Q\in\gS_2$ if $Q_{++},Q_{--}\in\gS_1$. The quantity $\tr(Q_{++})$ (resp. $\tr(Q_{--})$) represents the average number of electrons (resp. positrons) of the system. Hence, the convex set $\cK_\Lambda$ represents the set of all quantum state with a finite average number of electrons and positrons. The BDF energy of any quantum state $Q\in\cK_\Lambda$ is
\begin{equation}\label{eq:BDF-energy}
 \Ebdf{V}{Q}=\tr_0(D^0_\Lambda Q)+\alpha\int_{\R^3}\rho_QV+\frac{\alpha}{2}D(\rho_Q,\rho_Q)-\frac{\alpha}{2}\int_{\R^3}\int_{\R^3}\frac{|Q(x,y)|^2}{|x-y|}\d{x}\d{y}.
\end{equation}
Here, $V$ is the external electric potential and $\alpha>0$ is the coupling constant. The first term of this energy represents the kinetic energy of the system. The second term is the interaction energy between the external potential and the density of charge $\rho_Q$ of the system. The third term, usually called \emph{direct term}, is the Coulomb electrostatic energy of the density of charge $\rho_Q$. The last term, well-known in Hartree-Fock theory, is called the \emph{exchange term}. It is well-defined for $Q\in\cK_\Lambda$ by the Hardy-Kato inequality
$$\int_{\R^3}\int_{\R^3}\frac{|Q(x,y)|^2}{|x-y|}\d{x}\d{y}\le\frac{\pi}{2}\tr(|D^0_\Lambda|Q^2),$$
the right term being finite since $D^0_\Lambda$ is bounded. Since $\Ppm D^0_\Lambda=\pm|D^0_\Lambda|\Ppm$, the kinetic energy can be rewritten as
$$\tr_0D^0_\Lambda Q=\tr|D^0_\Lambda|(Q_{++}-Q_{--}),$$
which is finite for $Q\in\cK_\Lambda$. The direct term is also well-defined since $\rho_Q\in\cC$. Finally, we assume that $V$ is given by the Coulomb potential generated by a fixed external density of charge $\nu\in\cC$ of strength $Z>0$:
 $$V=-Z\nu\star\frac{1}{|\cdot|},$$
so that we can rewrite
$$\int_{\R^3}\rho_QV=-D(\rho_Q,Z\nu).$$
With this choice of $V$ and by the Hardy-Kato inequality, we find that $\Ebdf{\nu}{Q}$ is well-defined for any $Q\in\cK_\Lambda$, and that it is bounded-below for $0\le\alpha\le 4/\pi$ since it satisfies
$$\Ebdf{\nu}{Q}\ge-\frac{\alpha Z^2}{2}D(\nu,\nu),\qquad\forall Q\in\cK_\Lambda.$$
It was proved in \cite{HaiLewSer-05b} that for any $\nu\in\cC$, for any $Z>0$, and for any $0\le\alpha<4/\pi$, the BDF energy has at least one minimizer on $\cK_\Lambda$. If $\nu\equiv0$, then $Q=0$ is the unique global minimizer of the BDF energy. If $\nu\neq0$, uniqueness does not hold in general. However, there always exists one minimizer $\bar{Q}$ such that $\bar{Q}+\Pm$ is an orthogonal projection. More precisely, this minimizer satisfies the following self-consistent equation
\begin{equation}\label{eq:BDF-equation}
\left\{\begin{array}{lll}
          \bar{Q} & = & \chi_{(-\ii,0]}(D_{\bar{Q}})-\Pm,\\
          D_{\bar{Q}} & = & \Pi_\Lambda(D^0-\alpha Z V_\nu+\alpha(V_{\rho_{\bar{Q}}}-R_{\bar{Q}}))\Pi_\Lambda,
         \end{array}\right. 
\end{equation}
where $\Pi_\Lambda$ is the multiplication by $\ind_{B(0,\Lambda)}$ in Fourier space, $V_\rho:=\rho\star|\cdot|^{-1}$ for any $\rho$, and $R_Q:=\frac{Q(x,y)}{|x-y|}$ for any $Q$. This minimizer is interpreted as the polarized vacuum in the presence of $V_\nu$. 

\subsection{Static Pair Production}\label{sec:setting} We now present the setting we use to study pair production. As we said in the introduction, we consider the vacuum polarized by a density of charge $Z\nu$. Moreover, we consider the limit of weak interactions, ie $\alpha\to0$ and strong external field $Z\to+\ii$, with $\alpha Z$ fixed. Hence, we define a new parameter,
$$\kappa:=\alpha Z.$$

We make the following assumption on $\nu\in L^2(\R^3)\cap\cC$. We assume that there exists an eigenvalue $\lambda(\kappa\nu)$ of multiplicity 1 of the operator $D^{\kappa\nu}:=D^0-\kappa V_\nu$ on $\gH$ that \emph{crosses 0} at some value $\kappa=\kappa_c\in\R$. In mathematical terms, it means that there exists $\kappa_c\in\R$ and $\epsilon>0$ such that 
\begin{itemize}
 \item For all $\kappa\in[\kappa_c-\epsilon,\kappa_c+\epsilon]$, the operator $D^{\kappa\nu}$ on $\gH$ has only one eigenvalue $\lambda(\kappa\nu)$ in some fixed neighborhood of 0, of multiplicity 1;
 \item For $\kappa\in[\kappa_c-\epsilon,\kappa_c)$, $\lambda(\kappa\nu)>0$, and for $\kappa\in(\kappa_c,\kappa_c+\epsilon]$, $\lambda(\kappa\nu)<0$. 
\end{itemize}
The assumption of multiplicity 1 is made for simplification and it can be easily removed. For the sake of clarity, we will also absorb $\kappa_c$ in $\nu$ by assuming in the following that 
$$\boxed{
\kappa_c=1.
}
$$
We could also assume that $\lambda(\kappa\nu)$ changes sign the other way around, which only changes the physical interpretation: a positron is created instead of an electron. Of course, since $\nu\in L^2(\R^3)\cap\cC$, the eigenvalue $\lambda(\kappa\nu)$ is a continuous function of $\kappa$. Hence, there exist $-1<\mu_-<0<\mu_+<1$ such that for all $\kappa\in[1-\epsilon,1+\epsilon]$, $\lambda(\kappa)\in(\mu_-,\mu_+)$ and $\mu_\pm\notin\sigma(D^{\kappa\nu})$.

The interacting vacuum polarized by $Z\nu$ is a solution to the equation
\begin{equation}\label{eq:BDF-equation-2}
\left\{\begin{array}{lll}
          Q & = & \chi_{(-\ii,0]}(D_Q)-\Pm,\\
          D_Q & = & \Pi_\Lambda(D^{\kappa\nu}+\alpha(V_{\rho_Q}-R_Q))\Pi_\Lambda.
         \end{array}\right. 
\end{equation}
Here, the operator $D_Q$ is seen an operator on $\gH_\Lambda$. The non-interacting vacuum in presence of the density of charge $Z\nu$ is given by
$$\Ql(\kappa)=\chi_{(-\ii,0]}(D^{\kappa\nu})-P^0_-,$$
where $P^0_-=\chi_{(-\ii,0]}(D^0)$ is the negative spectral projection of the Dirac operator in $\gH$. Notice that, with our notations, the non-interacting vacuum satisfies
$$\left\{\begin{array}{lllllll}
          \Ql(\kappa) & = & \Qlinm(\kappa) & = & \chi_{(-\ii,\mu_-]}(D^{\kappa\nu})-P^0_- & \text{if} & \kappa < 1;\\
          \Ql(\kappa) & = & \Qlinp(\kappa) & = & \chi_{(-\ii,\mu_+]}(D^{\kappa\nu})-P^0_- & \text{if} & \kappa > 1.
         \end{array}\right.
$$
Consider the non-linear analogue of these equations, ie
\begin{equation}\label{eq:BDF-equation-mupm}
Q_\pm(\kappa,\alpha)=\chi_{(-\ii,\mu_\pm]}(\Pi_\Lambda(D^{\kappa\nu}+\alpha(V_{\rho_Q}-R_Q))\Pi_\Lambda)-\Pm. 
\end{equation}
Our first result shows the existence of solutions to this equation. By an easy adaptation of \cite[Theorem 2]{HaiLewSer-05a}, it can be shown that, for $\alpha$ small enough, $Q_-(\kappa,\alpha)$ (resp. $Q_+(\kappa,\alpha)$) is a minimizer of the BDF energy in the domain of charge $\tr_0(\Qlinm(\kappa))=q$ (resp. $\tr_0(\Qlinp(\kappa))=q+1$). Hence, $Q_+(\kappa,\alpha)$ has one electron more than $Q_-(\kappa,\alpha)$, as in the linear case. We want to know which one of these two has the lowest BDF energy. To answer this question, we introduce the function
\begin{equation}\label{eq:def-F}
F(\kappa,\alpha)=\Ebdf{Z\nu}{Q_+(\kappa,\alpha)}-\Ebdf{Z\nu}{Q_-(\kappa,\alpha)}. 
\end{equation}
If $F$ changes sign while increasing $\kappa$, it means that the vacuum gains one particle, as explained in the introduction. Our second result provides an expansion of $F(\kappa,\alpha)$ as $\alpha\to0$ with $\alpha\log\Lambda$ fixed. 

\subsection{Statement of the Results} Our first result concerns the existence of solutions to the equation \eqref{eq:BDF-equation-mupm}. In \cite{HaiLewSer-05a}, it has been solved by a fixed-point argument when $\mu=0$, for $\alpha$ small enough, for $\Lambda$ large enough, and for $\nu\in L^2(\R^3)\cap\cC$ such that $\|\kappa\nu\|_{L^2\cap\cC}$ is small enough. Here, we need to solve it for $\mu\notin\sigma(D^{\kappa\nu})$ and under no smallness assumption on the parameter $\kappa$ (indeed, if $\kappa$ is too small, $D^{\kappa\nu}$ cannot have an eigenvalue near 0, which we assume in our application to pair production). Hence, we have to go beyond the method of \cite{HaiLewSer-05a}. A solution to \eqref{eq:BDF-equation-2} can be found by a minimization argument, as in \cite{HaiLewSer-05b}. With this technique, no assumption on the smallness of $\|\kappa\nu\|_{L^2\cap\cC}$ and of $\alpha$ (except $\alpha<4/\pi$) is needed. By the same arguments, one can solve \eqref{eq:BDF-equation-mupm} with general $\mu$ by minimizing the BDF energy in sectors of charge $N$, given by $\cK_\Lambda(N):=\{Q\in\cK_\Lambda, \tr_0(Q)=N\}$. The existence of minimizers in these charge 
sectors has been proved in \cite{HaiLewSer-08}, under the assumption of a binding condition, well-known in many-body theories. However, we do not know if these binding conditions are satisfied in the regime of parameters we consider. Even if we could prove this binding condition, the range of $\mu$s for which we can solve \eqref{eq:BDF-equation-mupm}, ie the range of the map $N\mapsto\mu$, is not known. This range has been studied carefully in \cite{GraLewSer-09}, as well as the verification of the binding condition, in the case where the exchange term is neglected in the BDF energy. In the setting of our study, nothing is known about this issue. For all these reasons, we solve \eqref{eq:BDF-equation-mupm} by a fixed-point argument, following the method of \cite{HaiLewSer-05a}. This approach has also the advantage of furnishing a priori estimates on the solution of \eqref{eq:BDF-equation-mupm} in adequate norms, which are useful in proving the renormalization formula \eqref{eq:V-renorm} and also in the proof of Theorem \ref{thm:thm2}. 

Before stating our result, we have to introduce the Banach space where we will find solutions to \eqref{eq:BDF-equation-mupm}, which was introduced in \cite{HaiLewSer-05a}. For any $x\in\R^3$, we use the notation $E(x):=\sqrt{1+x^2}$. We define the space
$$\cX:=\cQ\times\left(L^2(\R^3)\cap\cC\right),$$
where
$$\cQ:=\{Q\in\gS_2(\gH_\Lambda),\:\|Q\|_\cQ^2:=\int_{\R^3}\int_{\R^3}E(p-q)^2E(p+q)|\hat{Q}(p,q)|^2\d{p}\d{q}<+\ii\}.$$
The vector space $\cX$ is a Banach space, endowed with the canonical norm
$$\|(Q,\rho)\|_\cX:=\|Q\|_\cQ+\|\rho\|_{L^2\cap\cC}.$$
We find solutions to \eqref{eq:BDF-equation-mupm} by applying a fixed-point argument to the couple $(Q,\rho_Q)$. That is why the Banach space $\cX$ contains both the one-body density matrices $Q$ as well as the densities of charge $\rho$. 

\begin{theorem}\label{thm:thm1}
Let $\nu\in L^2(\R^3)\cap\cC$, $\mu\in(-1,1)$, and $\epsilon\ge0$ such that 
\begin{equation}\label{eq:condition-gap}
\forall \kappa\in[1-\epsilon,1+\epsilon],\qquad\mu\notin\sigma(D^{\kappa\nu}).
\end{equation}
 Assume furthermore that
 $$\int_{\R^3}\log(2+|k|)^2|\hat{\nu}(k)|^2\d{k}<+\ii.$$
 Then, there exists $L_0=L_0(\mu,\nu,\epsilon)>0$, $\alpha_0=\alpha_0(\mu,\nu,\epsilon)>0$, $\Lambda_0=\Lambda_0(\mu,\nu,\epsilon)>0$, and $R_0=R_0(\mu,\nu,\epsilon)>0$ such that for all $0\le\alpha\le\alpha_0$ and $\Lambda\ge\Lambda_0$ satisfying $\alpha\log\Lambda\le L_0(\mu,\nu,\epsilon)$, and for all $\kappa\in[1-\epsilon,1+\epsilon]$, the equation
 \begin{equation}\label{eq:thm}
Q=\chi_{(-\ii,\mu]}(\Pi_\Lambda(D^0-\kappa V_\nu+\alpha(V_{\rho_Q}-R_Q))\Pi_\Lambda)-\Pm  
 \end{equation}
 has a unique solution $(Q,\rho_Q)\in\cX$ such that 
 \begin{equation}\label{eq:apriori}
\left\|(Q+G_{1,0}(\kappa\nur),\rho_Q-B_\Lambda\kappa\nur)\right\|_\cX\le R_0.   
 \end{equation}
 Here, we used the notations
 \begin{equation}\label{eq:def-nur}
\nur:=\frac{1}{1+\alpha B_\Lambda}\nu,
 \end{equation}
 and for any $\rho\in L^2(\R^3)\cap\cC$, $G_{1,0}(\rho)\in\cQ$ is defined by
 \begin{equation}\label{eq:def-G10}
  G_{1,0}(\rho)=\frac{1}{2\pi}\int_\R\frac{1}{D^0_\Lambda+i\eta}\Pi_\Lambda V_\rho\Pi_\Lambda\frac{1}{D^0_\Lambda+i\eta}\d{\eta}.
 \end{equation}
 Finally, the function $B_\Lambda$ is given by
 \begin{equation}\label{eq:dev-Blambda}
  B_\Lambda=\frac{1}{\pi}\int_0^{\frac{\Lambda}{\sqrt{1+\Lambda^2}}}\frac{z^2-z^4/3}{1-z^2}\d{z}=\frac{2}{3\pi}\log\Lambda-\frac{5}{9\pi}+\frac{2\log2}{3\pi}+\cO_{\Lambda\to\ii}(\Lambda^{-2}).
 \end{equation}
\end{theorem}

\begin{remark}
 The function $B_\Lambda$ already appeared in the previous works \cite{HaiLewSer-05a,HaiLewSer-05b,GraLewSer-09}. It is logarithmically divergent in the cut-off $\Lambda$. This kind of divergence is well-known is Quantum Electrodynamics, and has been noticed in the early days of this theory by Dirac. The reason why we study the regime where $\alpha\log\Lambda$ is bounded is to control this divergence. 
\end{remark}

\begin{remark}
 One can apply this theorem to $\epsilon=0$, in which case the condition on $\mu,\nu$ is simply $\mu\notin\sigma(D^\nu)$. Here, we allow $\epsilon>0$ since we need estimates uniform in $\kappa$ for the proof of Theorem \ref{thm:thm2}.
\end{remark}

\begin{corollary}[First order charge renormalization]\label{coro}
 Under the assumptions of Theorem \ref{thm:thm1}, and if we denote by $Q(\kappa\nu,\alpha)$ the solution of \eqref{eq:thm} given by Theorem \ref{thm:thm1}, we have as $\alpha\to0$ with $\alpha\log\Lambda\le L_0$ fixed:
 $$
 \boxed{
 Q(\kappa\nu,\alpha)-\Ql(\kappa\nur)\longrightarrow0
 }$$
 in $\cQ$, where
 $$
 \boxed{
 \Ql(\kappa\nur)=\chi_{(-\ii,\mu]}(\Pi_\Lambda(D^0-\kappa V_\nur)\Pi_\Lambda)-\Pm.
 }
 $$
\end{corollary}

\begin{remark}
 This corollary does not imply the convergence of $Q(\kappa\nu,\alpha)$ towards $\chi_{(-\ii,\mu]}(D^{\kappa\nur})-P^0_-$ in $\cQ$ as $\alpha\to0$, merely because this last operator does \emph{not} belong to $\cQ$. However, it belongs to $\gS_2$ by \cite{KlaSch-77b}. Since convergence in $\cQ$ implies convergence in $\gS_2$, we deduce that
 $$Q(\kappa\nu,\alpha)\longrightarrow\chi_{(-\ii,\mu]}(D^{\kappa\nur})-P^0_-$$
 in $\gS_2$, as $\alpha\to0$ with $\alpha\log\Lambda\le L_0$ fixed. 
\end{remark}

\begin{remark}
 By \eqref{eq:dev-Blambda}, in the limit $\alpha\to0$ with $\alpha\log\Lambda$ fixed, we have
 $$\nur=\frac{1}{1+\frac{2}{3\pi}\alpha\log\Lambda}\nu,$$
 which was the form used in the introduction. 
\end{remark}

Thanks to Theorem \ref{thm:thm1}, we can build a solution to \eqref{eq:thm} with $\nu\in L^2(\R^3)\cap\cC$ satisfying the assumptions of Section \ref{sec:setting}, for $\mu=\mu_\pm$, for any $\kappa\in[1-\epsilon,1+\epsilon]$, and for any $\alpha>0$ small enough, $\Lambda$ large enough such that $\alpha\log\Lambda$ is small enough. Since the corresponding $Q_\pm(\kappa\nu,\alpha)$ belongs to $\cQ\subset\gS_2(\gH_\Lambda)$ and is a difference of orthogonal projections, we can use the fact that
$$P=P^2=P^*\Rightarrow (P-\Pm)^2=(P-\Pm)_{++}-(P-\Pm)_{--},$$
to infer that $Q_\pm(\kappa,\alpha)\in\cK_\Lambda$, hence the function $F$ given by \eqref{eq:def-F} is well-defined. Our second result is the following:

\begin{theorem}\label{thm:thm2}
 Let $\nu\in L^2(\R^3)\cap\cC$ satisfying the assumptions of Section \ref{sec:setting}, with furthermore
 $$\int_{\R^3}\log(2+|k|)^2|\hat{\nu}(k)|^2\d{k}<+\ii.$$ 
 Then, for $\alpha$ small enough, for $\Lambda$ large enough such that $\alpha\log\Lambda$ is small enough, we have the following expansion for all $\kappa\in[1-\epsilon,1+\epsilon]$:
 \begin{equation}\label{eq:F-expansion}
  \boxed{
  F(\kappa,\alpha)=\lambda\left(\kappa\nur\right)+\cO_{\alpha\to0}(\alpha),
  }
  \end{equation}
where $\cO_{\alpha\to0}(\alpha)$ denotes a function which absolute value is $\le C\alpha$ for $\alpha$ small, where $C$ only depends on $\mu$, $\nu$, and $\epsilon$.  
\end{theorem}

\begin{remark}
 Theorem \ref{thm:thm2} shows that for $\alpha$ small enough, the interacting vacuum for $\kappa>\kappa_c(\alpha)$ has one electron more than the interacting vacuum for $\kappa<\kappa_c(\alpha)$, with $\kappa_c(\alpha)\sim_{\alpha\to0}1+\alpha B_\Lambda>1=\kappa_c$. In words, charge renormalization hinders pair production. 
\end{remark}

\begin{remark} One could look for the next order in the expansion \eqref{eq:F-expansion}, i.e. look for a number $c$ such that
$$F(\kappa,\alpha)=\lambda(\kappa\nur)+c\alpha+o_{\alpha\to0}(\alpha).$$
 Of particular importance would be the determination of the sign of $c$, since it would give the next order influence of the interactions on pair production. Indeed, the first term in the expansion of $F$, $\lambda(\kappa\nur)$ is purely an effect of charge renormalization, and not a direct consequence of the interactions between the particles (though charge renormalization is itself a consequence of the interactions). It can be shown that $c=\langle\phi,H(\kappa\nur)\phi\rangle$, where $\phi\in\gH$ is the eigenvector associated to $\lambda(\kappa\nur)$, and $H(\kappa\nur)$ is a self-adjoint operator on $\gH$. However, we could not determine if $H(\kappa\nur)$ had a sign, because of its very complicated form. For this reason, we do not provide the next order expansion of $F$ in this paper. The details are provided in \cite{Sab-phd}. In principle, our method enables to compute the expansion of $F$ to any order.
\end{remark}

\begin{remark}
 When the exchange term is included, the operators $P^0_-$ and $D^0$ should in principle be replaced by effective operators $\cP^0_-$ and $\cD^0$ on $\gH_\Lambda$ satisfying a non-linear equation and which depend on both $\alpha$ and $\Lambda$, see \cite{HaiLewSol-07}. Our strategy works the same with $\cP^0_-$ and $\cD^0$, but the estimates are even more technical. Some details on how to adapt our method will be provided in \cite{Sab-phd}. In particular, in Corollary \ref{coro}, the operator $D^0-\kappa V_\nur$ becomes 
 \begin{equation}\label{eq:def-Deff}
D_{\text{eff}}=g_1'(0)\alp\cdot(-i\nabla)+g_0(0)\beta-\frac{\kappa}{1+\alpha B_\Lambda}V_\nu,  
 \end{equation}
where
 $$g_1'(0)=1+\frac{2}{3\pi}\alpha\log\Lambda,\quad g_0(0)=\frac{g_1'(0)}{1-\frac{1}{3\pi}\alpha\log\Lambda},\quad B_\Lambda=\frac{2}{3\pi}\frac{\log\Lambda}{g_1'(0)}.$$
 In terms of mass and charge renormalization, the interpretation of this operator is not completely clear to us (see \cite{LieSie-00} and \cite[Sect. 2.5]{HaiLewSol-07} for a discussion).
\end{remark}

The rest of the article is devoted to the proofs of Theorem \ref{thm:thm1} and Theorem \ref{thm:thm2}. 

\section{Proofs}

\subsection{Proof of Theorem \ref{thm:thm1}}\label{sec:proof-thm1}

The strategy for solving \eqref{eq:thm} in \cite{HaiLewSer-05a} was to see the operator $$D_Q:=D^0_\Lambda+\alpha\Pi_\Lambda(V_{\rho_Q-Z\nu}-R_Q)\Pi_\Lambda$$
as a small perturbation of $D^0_\Lambda$ if $\alpha$ is small enough, and then use the resolvent formula to expand 
$$\chi_{(-\ii,\mu]}(D_Q)-\chi_{(-\ii,\mu]}(D^0_\Lambda)$$
in powers of $\alpha$, to find a functional to which they apply the fixed-point theorem of Banach-Picard. Here, we obviously cannot follow this strategy since $\alpha Z=\kappa$ is fixed and not small. Instead, we have to find a new reference operator $\Dref$, independent of $Q$, satisfying $D_Q=\Dref+\cO(\alpha)$, in an adequate sense. We choose
\begin{equation}\label{eq:def-dref}
\boxed{
 \Dref:=D^0_\Lambda-\Pi_\Lambda V_{\kappa\nur}\Pi_\Lambda+\alpha \Pi_\Lambda R_{1,0}(\kappa\nur)\Pi_\Lambda,
 }
\end{equation}
where\footnote{Here and in the sequel, for any operator $Q$ with integral kernel $Q(x,y)$, $R[Q]$ (or $R_Q$) is the operator which kernel is given by $\frac{Q(x,y)}{|x-y|}$, and for any density of charge $\rho$, $V[\rho]$ (or $V_\rho$) is the Coulomb potential generated by $\rho$, i.e. $V[\rho]:=\rho\star|\cdot|^{-1}$.} $R_{1,0}(\kappa\nur):=R[G_{1,0}(\kappa\nur)]$ is the exchange term associated to $G_{1,0}(\kappa\nur)$, which was defined in \eqref{eq:def-G10}, and $\nur$ was defined in \eqref{eq:def-nur}. 
\begin{remark}
 A more obvious choice for $\Dref$ would have been $\Dref=D^0_\Lambda-\Pi_\Lambda V_{\kappa\nu}\Pi_\Lambda$. Unfortunately, for this choice we do not have $D_Q-\Dref=\cO(\alpha)$ since there remains diverging terms in the expression $\Pi_\Lambda(V_{\rho_Q}-R_Q)\Pi_\Lambda$. These terms diverge precisely as $\log\Lambda$, and extracting them leads to the expression \eqref{eq:def-dref}. 
\end{remark}
With this new reference operator, Equation \eqref{eq:thm} is equivalent to
\begin{equation}\label{eq:thm-newref}
 Q=\chi_{(-\ii,\mu]}\left(\Dref+\alpha\Pi_\Lambda\left[V_{\rho_Q+\rhoref}-R_{Q+\Qref}\right]\Pi_\Lambda\right)-\Pref,
\end{equation}
where 
$$\Pref:=\chi_{(-\ii,\mu]}(\Dref),\quad \Qref:=\Pref-\Pm+G_{1,0}(\kappa\nur),\quad \rhoref:=\rho_{\Pref-\Pm}-B_\Lambda\kappa\nur.$$
Under adequate conditions on $\nu$, we will show that $\|(\Qref,\rhoref)\|_\cX=\cO_{\alpha\to0}(1)$, so that we are able to solve \eqref{eq:thm-newref} by a fixed-point argument. First, we need to derive some properties of the operator $\Dref$.

\subsubsection{Spectral properties of the operator $\Dref$} We first recall the following general result, which was partly contained in \cite{HaiLewSer-05a}.

\begin{lemma}\label{lemma:basic-estimates}
 Let $\rho\in L^2(\R^3)\cap\cC$ and $Q\in\cQ$. Then, we have
 \begin{equation}\label{eq:ineq-VLii}
  \|V_\rho\|_{L^\ii}\le2\sqrt{\pi}\|\rho\|_{L^2\cap\cC},
 \end{equation}
 and there exists $C>0$ such that for all $\eta\in\R$, 
 \begin{equation}\label{eq:ineq-VD0-1S6}
  \|V_{\rho}|D^0+i\eta|^{-1}\|_{\gS_6}\le \frac{C}{E(\eta)^{1/2}}\|\rho\|_\cC.
 \end{equation}
 For any $\eta\in\R$, we have
\begin{equation}\label{eq:ineq-RQD0-1S2}
 \|R_Q|D^0+i\eta|^{-1}\|_{\gS_2}\le \frac{C_0}{E(\eta)^{1/2}}\|Q\|_\cQ,
\end{equation}
where 
\begin{equation}\label{eq:def-C0}
 C_0:=\frac{\sqrt{2}}{\pi^2}\inf_{\tau\in(0,2)}\sup_{x\in\R^3}\left(E(2x)^\tau\int_{\R^3}\frac{\d{u}}{E(2u)^{1+\tau}|u-x|^2}\right).
\end{equation}
Finally, we have the estimates,
\begin{equation}\label{eq:ineq-RQD0-1L2L2}
 \|R_Q|D^0|^{-1}\|_{L^2\to L^2}=\|R_Q\|_{H^1\to L^2}\le 2\|Q\|_{\gS_2},
\end{equation}
\begin{equation}\label{eq:ineq-RQD0-12L2L2}
 \|R_Q|D^0|^{-1/2}\|_{L^2\to L^2}=\|R_Q\|_{H^{1/2}\to L^2}\le \frac{\pi}{2}\|Q\|_\cQ.
 \end{equation}
\end{lemma}
\begin{proof}
 Inequality \eqref{eq:ineq-VLii} was proved in \cite[Lemma 3]{HaiLewSer-05a}. Inequality \eqref{eq:ineq-VD0-1S6} is a consequence of the Kato-Seiler-Simon inequality
 \begin{equation}
  \forall p\ge 2,\:\forall f,g\in L^p(\R^3),\quad \|f(x)g(-i\nabla)\|_{\gS_p}\le (2\pi)^{-3/p}\|f\|_{L^p}\|g\|_{L^p},
 \end{equation}
 of the Sobolev inequality $\|V_\rho\|_{L^6}\le C\|\nabla V_\rho\|_{L^2}=C\|\rho\|_\cC$, and of 
  $$\left\|\frac{1}{|\alp\cdot p+\beta+i\eta|}\right\|_{L^6}^6=\int_{\R^3}\frac{\d{p}}{(E(p)^2+\eta^2)^3}=\frac{C}{E(\eta)^3}.$$
 Finally, inequality \eqref{eq:ineq-RQD0-1S2} is proved using the fact that for all $p,q\in\R^3$,
 $$\frac{E(p+q)E(\eta)}{(E(p)^2+\eta^2)E(p-q)}\le 2,$$
 hence
 \begin{multline*}
\|R_Q|D^0+i\eta|^{-1}\|^2_{\gS_2}=\iint\frac{|\hat{R_Q}(p,q)|^2}{E(p)^2+\eta^2}\d{p}\d{q}\\
\le\frac{2}{E(\eta)}\iint\frac{E(p-q)^2}{E(p+q)}|\hat{R_Q}(p,q)|^2\d{p}\d{q}\le\frac{C_0^2}{E(\eta)}\|Q\|_\cQ^2,  
 \end{multline*}
 where the second inequality was proved in \cite[Lemma 8]{HaiLewSer-05a}. To prove \eqref{eq:ineq-RQD0-1L2L2}, let $\phi\in H^1$. Then, by the H\"older inequality we have for all $x\in\R^3$,
 $$|R_Q\phi(x)|^2\le\int_{\R^3}|Q(x,y)|^2\d{y}\int_{\R^3}\frac{|\phi(y)|^2}{|x-y|^2}\d{y}.$$
 By the Hardy inequality, one deduces
 $$\|R_Q\phi\|_{L^2}^2\le4\|Q\|_{\gS_2}^2\|\nabla\phi\|_{L^2}^2\le4\|Q\|_{\gS_2}^2\||D^0|\phi\|_{L^2}^2.$$
 The proof of \eqref{eq:ineq-RQD0-12L2L2} relies on the same computation for $\phi\in H^{1/2}$:
 $$|R_Q\phi(x)|^2\le\int_{\R^3}\frac{|Q(x,y)|^2}{|x-y|}\d{y}\int_{\R^3}\frac{|\phi(y)|^2}{|x-y|}\d{y},$$
 so that by the Hardy-Kato inequality applied to both terms,
 $$\|R_Q\phi\|_{L^2}^2\le\frac{\pi^2}{4}\tr(|D^0|Q^2)\langle|D^0|\phi,\phi\rangle_{L^2}\le\frac{\pi^2}{4}\|Q\|_\cQ^2\||D^0|^{1/2}\phi\|_{L^2}^2.$$
 
\end{proof}

\begin{lemma}[Properties of $G_{1,0}(\rho)$]\label{lemma:prop-G10} Let $\rho\in L^2(\R^3)\cap\cC$. Then, for any $\Lambda>1$, the operator $G_{1,0}(\rho)$ defined by \eqref{eq:def-G10} satisfies for some $C>0$:
\begin{equation}\label{eq:ineq-G10Q}
\|G_{1,0}(\rho)\|_\cQ\le C\sqrt{\log\Lambda}\|\rho\|_{L^2\cap\cC}, 
\end{equation}
\begin{equation}\label{eq:ineq-G10S2}
\|G_{1,0}(\rho)\|_{\gS_2}\le 4\|\rho\|_\cC.
\end{equation}
\end{lemma}
\begin{proof}
 Inequality \eqref{eq:ineq-G10Q} was proved in \cite[Lemma 11]{HaiLewSer-05a}. It relies of the explicit form of the integral kernel of $G_{1,0}(\rho)$ in Fourier space, obtained by the residuum formula:
 \begin{equation}
  \forall p,q\in\R^3,\quad \hat{G_{1,0}(\rho)}(p,q)=\frac{1}{2(2\pi)^{3/2}}\hat{V_\rho}(p-q)M(p,q),
 \end{equation}
 where the function $M$ is defined by
 \begin{equation}\label{eq:def-M}
  M(p,q)=\frac{1}{E(p)+E(q)}\left(\frac{\alp\cdot p+\beta}{E(p)}\frac{\alp\cdot q+\beta}{E(q)}-1\right).
 \end{equation}
  It was proved in \cite[Lemma 12]{HaiLewSer-05a} that 
 \begin{equation}
  |M(p,q)|^2\le16\frac{|p-q|^2}{E(p+q)^4},
 \end{equation}
 hence 
 $$\|G_{1,0}(\rho)\|_{\gS_2}^2\le \frac{64}{\pi}\int_{\R^3}|k|^2|\hat{V_\rho}(k)|^2\d{k}\int_{\R^3}\frac{\d{\ell}}{E(\ell)^4}=16\|\rho\|_\cC^2,$$
 where we used that $\int_{\R^3}E(\ell)^{-4}\d{\ell}=\pi^2$. 
\end{proof}

Combining Lemmas \ref{lemma:basic-estimates} and \ref{lemma:prop-G10}, we infer that for any $\alpha>0$ and $\Lambda>1$, the operator $\Dref$ is bounded and self-adjoint on $\gH_\Lambda$. The following result ensures that $\Dref$ has a gap around $\mu$ for $\alpha$ small enough and $\Lambda$ large enough. For any $\omega\in L^2(\R^3)\cap\cC$, we introduce the auxiliary operator on $\gH_\Lambda$,
\begin{equation}
 \Dref(\omega):=D^0_\Lambda-\Pi_\Lambda V_\omega\Pi_\Lambda+\alpha \Pi_\Lambda R_{1,0}(\omega)\Pi_\Lambda,
\end{equation}
so that $\Dref=\Dref(\kappa\nur)$. 

\begin{lemma}\label{lemma:gap-Dref}
 Let $\nu\in L^2(\R^3)\cap\cC$ and $\mu\in(-1,1)$. Let $\epsilon\ge0$ such that for all $s\in[1-\epsilon,1+\epsilon]$, $\mu\notin\sigma_\gH(D^{s\nu})$. Then for all $L>0$, there exists $\alpha_0=\alpha_0(\mu,\nu,\epsilon,L)>0$, $\Lambda_0=\Lambda_0(\mu,\nu,\epsilon,L)>1$ and $\xi_0=\xi_0(\mu,\nu,\epsilon,L)>0$ such that for all $0\le\alpha\le\alpha_0$, for all $\Lambda\ge\Lambda_0$ satisfying $\alpha\sqrt{\log\Lambda}\le L$, and for all $s\in[1-\epsilon,1+\epsilon]$, we have $d(\mu,\sigma_{\gH_\Lambda}(\Dref(s\nu)))\ge\xi_0$.
\end{lemma}

\begin{proof}
 The proof is done by contradiction. Assume there exists $\alpha_n\to0$, $\Lambda_n\to+\ii$ satisfying $\alpha_n\sqrt{\log\Lambda_n}\le L$, $s_n\in[1-\epsilon,1+\epsilon]$ such that $d(\mu,\sigma_{\gH_{\Lambda_n}}(\Dref(s_n\nu))\le1/n$. Up to extracting a subsequence, we can assume $s_n\to s\in[1-\epsilon,1+\epsilon]$. Since $d(\mu,\sigma_{\gH_{\Lambda_n}}(\Dref(s_n\nu)))\le1/n$, there exists $\phi_n\in\gH_n:=\gH_{\Lambda_n}$ with $\|\phi_n\|_{L^2}=1$ and $\|(D_n-\mu)\phi_n\|_{L^2}\le2/n$, with $D_n:=\Dref(s_n\nu)$. For any $\eta\in\R$ and $n\in\N$ we have
 $$D_n\phi_n=(1+\Pi_n(\alpha_n R_n(s_n\nu)-V_{s_n\nu})(D^0+i\eta)^{-1})(D^0+i\eta)\phi_n-i\eta\phi_n,$$
 where we used the notations $\Pi_n:=\Pi_{\Lambda_n}$ and $R_n(s_n\nu):=R_{1,0}(s_n\nu)$. Now for $\eta$ large enough,  we have $\|V_{s_n\nu}(D^0+i\eta)^{-1}\|<1$ by \eqref{eq:ineq-VD0-1S6}, and $\|\alpha_n R_{1,0}^{(n)}(\omega)(D^0+i\eta)^{-1}\|<1$ by \eqref{eq:ineq-RQD0-1S2} and the fact that $\alpha_n\sqrt{\log\Lambda_n}\le L$. Then, we have for such a $\eta$,
 $$(D^0+i\eta)\phi_n=\frac{1}{(1+\Pi_n(\alpha_n R_n(s_n\nu)-V_{s_n\nu})(D^0+i\eta)^{-1})}(D_n\phi_n+i\eta\phi_n)	,$$
 showing that $(\phi_n)_n$ is bounded in $H^1(\R^3)$. Hence, up to a subsequence we can assume $\phi_n\rightharpoonup \phi$ in $H^1(\R^3)$. This implies that $D^0 \phi_n\rightharpoonup D^0 \phi$ and $\mu \phi_n\rightharpoonup\mu \phi$ weakly in $L^2(\R^3)$. Since $V_{s\nu} (D^0)^{-1}$ is a compact operator we also have $V_{s_n\nu}\phi_n\to V_{s\nu}\phi$ strongly in $L^2(\R^3)$ and hence $\Pi_nV_{s_n\nu}\phi_n\to V_{s\nu}\phi$ strongly in $L^2(\R^3)$. Using \eqref{eq:ineq-RQD0-1L2L2} and \eqref{eq:ineq-G10S2}, we infer that 
   $$\|\alpha_n R_n(s_n\nu)|D^0|^{-1}\|_{L^2\to L^2}\le  2\alpha_n\|G_{1,0}(s_n\nu)\|_{\gS_2}\le 8\alpha_n(1+\epsilon)\|\nu\|_\cC\to0.$$
   We thus have $\Pi_n\alpha_n R_n(s_n\nu)\phi_n$ converging strongly to 0 in $L^2$. Hence, $D_n\phi_n\rightharpoonup(D^0-V_{s\nu})\phi$ weakly in $L^2$ and since $\|(D_n-\mu)\phi_n\|_{L^2}\le2/n$, we have $(D^0-V_{s\nu})\phi=\mu\phi$. But $\mu\notin\sigma_\gH(D^{s\nu})$, hence we must have $\phi=0$. This implies that $\Pi_nV_{s_n\nu}\phi_n\to 0$ strongly in $L^2(\R^3)$. We deduce $(D^0-\mu)\phi_n\to0$ strongly in $L^2$, which is a contradiction by the Weyl criterion for the essential spectrum since $\|\phi_n\|_{L^2}=1$, $\phi_n\rightharpoonup0$ weakly in $L^2$ and $\mu\in(-1,1)$. 
\end{proof}

 \begin{lemma}\label{lemma:D0Dref-1L2L2}
  Let $\omega\in L^2(\R^3)\cap\cC$ and $\mu\in(-1,1)$ such that $\mu\notin\Dref(\omega)$. Then for all $\alpha\ge0$ such that $8\alpha\|\omega\|_\cC<1-|\mu|$, the operator $|D^0|(\Dref(\omega)-\mu)^{-1}$ (extended by 0 on $\gH_\Lambda^\perp$) is bounded on $L^2(\R^3,\C^4)$, and we have
  $$\left\||D^0|\frac{1}{\Dref(\omega)-\mu}\right\|_{L^2\to L^2}\le\theta(\mu,\omega),$$
  where
  $$\theta(\mu,\omega)=\left(1-\frac{8\alpha\|\omega\|_\cC}{1-|\mu|}\right)^{-1}\frac{1}{1-|\mu|}\left(1+\frac{2\sqrt{\pi}\|\omega\|_\cC}{d(\mu,\sigma(\Dref(\omega)))}\right).$$
 \end{lemma}
 \begin{proof}
 By the resolvent identity, we have
 $$\frac{1}{\Dref(\omega)-\mu}=\frac{1}{D^0_\Lambda-\mu}+\frac{1}{D^0_\Lambda-\mu}\Pi_\Lambda(V_\omega-\alpha R_{1,0}(\omega))\Pi_\Lambda\frac{1}{\Dref(\omega)-\mu}.$$
 Using $\||D^0|(D^0_\Lambda-\mu)^{-1}\|_{L^2\to L^2}\le(1-|\mu|)^{-1}$ and $\|R_{1,0}(\omega)|D^0|^{-1}\|_{L^2\to L^2}\le 8\|\omega\|_{\cC}$ by \eqref{eq:ineq-RQD0-1L2L2} and \eqref{eq:ineq-G10S2}, we obtain 
 $$\left(1-\frac{8\alpha\|\omega\|_{\cC}}{1-|\mu|}\right)\left\||D^0|\frac{1}{\Dref(\omega)-\mu}\right\|_{L^2\to L^2}\le\frac{1}{1-|\mu|}\left(1+\frac{2\sqrt{\pi}\|\omega\|_{L^2\cap\cC}}{d(\mu,\sigma(\Dref(\omega)))}\right).$$ 
 \end{proof}

 Lemma \ref{lemma:gap-Dref} ensures that we can use Cauchy's formula to $\Pref$ since $\Dref$ has a gap around $\mu$. The next result is the analogue of \cite[Lemma 9]{HaiLewSer-05a}. It ensures that the operator $D_Q$ also has a gap around $\mu$. For any $Q\in\cQ$ and $\omega,\rho'\in L^2(\R^3)\cap\cC$ we define the following operator on $\gH_\Lambda$:
\begin{equation}\label{eq:DQmu}
 D_{\omega,\rho',Q}=\Dref(\omega)+\alpha\Pi_\Lambda\left(V_{\rho'}-R_Q\right)\Pi_\Lambda.
\end{equation}
 With this notation, we have $D_Q=D_{\kappa\nur,\rho_Q+\rhoref,Q+\Qref}$. 
\begin{lemma}\label{lemma:gap-DQ}
 Let $\omega\in L^2(\R^3)\cap\cC$, $\mu\in(-1,1)$, $\alpha\ge0$, $\Lambda>0$,  $(Q,\rho')\in\cX$ be such that the following conditions are satisfied:
 \begin{enumerate}
  \item $\mu\notin\sigma(\Dref(\omega))$;
  \item $8\alpha\|\omega\|_\cC<1-|\mu|$;
  \item $\alpha b\|(Q,\rho')\|_\cX<1$,
 \end{enumerate}
where, using the notations of Lemma \ref{lemma:D0Dref-1L2L2},
 $$b:=\frac{1}{d(\mu,\sigma(\Dref(\omega)))}+\theta(\mu,\omega).$$
 Then, we have the operator inequality
 $$|D_{\omega,\rho',Q}-\mu|\ge\left(1-\alpha b\|(Q,\rho')\|_\cX\right)|\Dref(\omega)-\mu|.$$
 In particular, $\mu\notin\sigma(D_{\omega,\rho',Q})$. 
\end{lemma}
\begin{remark}
 The conditions listed in the above lemma give a non trivial relation between $\alpha$ and $\mu$ (recall that $\Dref(\omega)$ depends on $\alpha$). Fortunately, these conditions are satisfied if $\mu\in(-1,1)$, $\mu\notin\sigma(D^0-V_\omega)$ and $\alpha$ is small enough by Lemma \ref{lemma:gap-Dref}.
\end{remark}
\begin{proof}
 We adapt the proof of \cite[Lemma 9]{HaiLewSer-05a} to our case. First of all, since $\sigma_{\text{ess}}(\Dref(\omega))=[-E(\Lambda),-1]\cup[1,E(\Lambda)]$ and $\mu\in(-1,1)$, $\mu\notin\sigma(\Dref(\omega))$, we have $d(\mu,\sigma(\Dref(\omega)))>0$. Let $u\in\gH_\Lambda$. By Lemma \ref{lemma:basic-estimates}, we have
 $$\|\Pi_\Lambda V_{\rho'} u\|_{L^2}\le\|V_{\rho'}\|_{L^\ii}\|u\|_{L^2}\le\frac{2\sqrt{\pi}\|\rho'\|_{L^2\cap\cC}}{d(\mu,\sigma(\Dref(\omega)))}\||\Dref(\omega)-\mu|u\|_{L^2},$$
 where we used the operator inequality $|\Dref(\omega)-\mu|\ge d(\mu,\sigma(\Dref(\omega)))$. For the exchange term, we write
 $$\|\Pi_\Lambda R_Q u\|_{L^2}\le\|R_Q|D^0|^{-1}\|_{\gS_2}\left\||D^0|\frac{1}{|\Dref(\omega)-\mu|}\right\|_{L^2\to L^2}\||\Dref(\omega)-\mu|u\|_{L^2}.$$
 By Lemma \ref{lemma:D0Dref-1L2L2}, this implies
$$\|\Pi_\Lambda R_Q u\|_{L^2}\le C_0\|Q\|_\cQ\theta(\mu,\omega)\||\Dref(\omega)-\mu|u\|_{L^2},$$
showing that 
$$|\Pi_\Lambda(V_{\rho'}-R_Q)\Pi_\Lambda|\le\left(\frac{1}{d(\mu,\sigma(\Dref(\omega)))}+\theta(\mu,\omega)\right)\|(Q,\rho')\|_\cX|\Dref(\omega)-\mu|.$$
\end{proof}

\subsubsection{Resolvent expansion} Now that we know sufficient conditions on our parameters to prove that both $\Dref$ and $D_Q$ have a gap around $\mu$, we can use Cauchy formula, and then resolvent expansion (at least formally for the moment), to see that \eqref{eq:thm-newref} is equivalent to 
\begin{equation*}
 Q  =  \frac{1}{2\pi}\int_\R\left(\frac{1}{\Dref-\mu+i\eta}-\frac{1}{D_Q-\mu+i\eta}\right)\d{\eta},
\end{equation*}
which leads to
\begin{multline*}
   Q=\sum_{n\ge1}\frac{(-1)^{n+1}\alpha^n}{2\pi}\int_\R\frac{1}{\Dref-\mu+i\eta}\times\\
   \times\left(\Pi_\Lambda(V_{\rho_Q+\rhoref}-R_{Q+\Qref})\Pi_\Lambda\frac{1}{\Dref-\mu+i\eta}\right)^n\d{\eta}.
\end{multline*}
We write it as
\begin{equation}\label{eq:serie}
 Q=\sum_{n\ge1}\alpha^n G_n^{\kappa\nur}(Q+\Qref,\rho_Q+\rhoref),
\end{equation}
where for any $Q\in\cQ$, $\omega,\rho'\in L^2(\R^3)\cap\cC$,
\begin{equation}
 G_n^\omega(Q,\rho')=\frac{(-1)^{n+1}}{2\pi}\int_\R\frac{1}{\Dref(\omega)-\mu+i\eta}\left(\Pi_\Lambda(V_{\rho'}-R_Q)\Pi_\Lambda\frac{1}{\Dref(\omega)-\mu+i\eta}\right)^n\d{\eta}.
\end{equation}
We will prove in particular that \eqref{eq:serie} makes sense by showing that the series converges in $\cQ$, for $\alpha$ small enough. Equation \eqref{eq:serie} implies that 
\begin{equation}\label{eq:serie-rho}
\rho_Q=\sum_{n\ge1}\alpha^n\rho[G_n^{\kappa\nur}(Q+\Qref,\rho_Q+\rhoref)]. 
\end{equation}
We use the same trick as in \cite{HaiLewSer-05a}, by writing
$$\rho[G_1^\omega(Q,\rho')]=\rho[G_{1,0}(\rho')]+J_1^\omega(Q,\rho'),$$
where we recall that $G_{1,0}(\rho')$ was defined in \eqref{eq:def-G10}. It was proved in \cite{HaiLewSer-05a} that 
$$\forall k\in\R^3,\quad \cF\left[\rho[G_{1,0}(\rho')]\right](k)=-B_\Lambda(k)\hat{\rho'}(k),$$
where $\cF[f]:=\hat{f}$. The function $B_\Lambda(k)$ is computed in \cite[Eq.(88)]{GraLewSer-09}. It satisfies $B_\Lambda(k)=0$ for $|k|>2\Lambda$ while for $|k|\le2\Lambda$ we have the formula 
\begin{multline*}
 B_\Lambda(k)=\frac{1}{\pi}\int_0^{Z_\Lambda(|k|)}\frac{z^2-z^4/3}{(1-z^2)(1+|k|^2(1-z^2)/4)}\d{z}\\
 +\frac{|k|}{2\pi}\int_0^{Z_\Lambda(|k|)}\frac{z-z^3/3}{\sqrt{1+\Lambda^2}-|k|z/2}\d{z},
\end{multline*}
with
$$Z_\Lambda(r)=\frac{\sqrt{1+\Lambda^2}-\sqrt{1+(\Lambda-r)^2}}{r},\quad\forall r\ge0.$$
Notice that $B_\Lambda=B_\Lambda(0)$. Defining $J_n^\omega(Q,\rho'):=\rho[G_n^\omega(Q,\rho')]$ for $n\ge2$, we see that \eqref{eq:serie-rho} is equivalent to
\begin{multline*}
\forall k\in\R^3,\quad\F[\rho_Q+\rhoref](k)=\frac{1}{1+\alpha B_\Lambda(k)}\hat{\rhoref}(k)\\
+\sum_{n\ge1}\frac{\alpha^n}{1+\alpha B_\Lambda(k)}\cF[J_n^\nur(Q+\Qref,\rho_Q+\rhoref)](k). 
\end{multline*}
Hence, the couple $(Q,\rho_Q+\rhoref)$ satisfies
$$(Q,\rho_Q+\rhoref)=\Phi(Q,\rho_Q+\rhoref)=\left(\Phi_1^{\kappa\nur}(Q+\Qref,\rho_Q+\rhoref),\Phi_2^{\kappa\nur}(Q+\Qref,\rho_Q+\rhoref)\right),$$
where 
$$\Phi_1^\omega(Q,\rho')=\sum_{n\ge1}\alpha^n G_n^\omega(Q,\rho'),$$
$$\forall k\in\R^3,\quad\cF[\Phi_2^\omega(Q,\rho')](k)=\frac{1}{1+\alpha B_\Lambda(k)}\hat{\rhoref}(k)+\sum_{n\ge1}\frac{\alpha^n}{1+\alpha B_\Lambda(k)}\cF[J_n^\omega(Q,\rho')](k).$$
We prove that $\Phi$ stabilizes a certain ball in $\cX$ and that it is contractant on this ball. By the Banach-Picard theorem, it proves Theorem \ref{thm:thm1}. To do so, we need estimates on the function $\Phi$. 

\subsubsection{Estimates on $\Phi$} We begin by a useful result about commutators. 

\begin{lemma}\label{lemm:comm}
 Let $\rho\in L^2(\R^3)\cap\cC$ and $Q\in\cQ$. Then we have the following estimates:
 \begin{equation}\label{eq:Vcomm}
 \forall 0<\zeta<1,\quad \forall t>\zeta+\frac{1}{2},\quad \left\|\left[|D^0|^\zeta,V_\rho\right]|D^0|^{-t}\right\|_{\gS_2}\le C\|\rho\|_{\cC},
\end{equation}
\begin{equation}\label{eq:Rcomm}
 \forall 0<\zeta<1,\quad\forall t>\zeta,\quad \left\|\left[|D^0|^\zeta,R_Q\right]|D^0|^{-t}\right\|_{\gS_2}\le C\|Q\|_\cQ.
\end{equation}
\end{lemma}
\begin{proof}
 We begin by noticing that, if $T$ is a self-adjoint operator, we have
 $$\|[|D^0|^\zeta,T]|D^0|^{-t}\|_{\gS_2}\le \sqrt{2}\|[|D^0|^\zeta,T^>]|D^0|^{-t}\|_{\gS_2},$$
 where $\hat{T^>}(p,q)=T(p,q)\chi_{|p|>|q|}$ for all $p,q\in\R^3$.
 Indeed, 
 \begin{eqnarray*}
\|[|D^0|^\zeta,T]|D^0|^{-t}\|_{\gS_2}^2 & = & \iint_{\R^3\times\R^3}\frac{|E(p)^\zeta-E(q)^\zeta|^2}{E(q)^{2t}}|\hat{T}(p,q)|^2\d{p}\d{q}\\
 & = & \iint_{|p|>|q|}\cdots+\iint_{|p|<|q|}\cdots,  
 \end{eqnarray*}
 and using $|\hat{T}(p,q)|=|\hat{T}(q,p)|$ we find
\begin{multline*}
\iint_{|p|<|q|}\frac{|E(p)^\zeta-E(q)^\zeta|^2}{E(q)^{2t}}|\hat{T}(p,q)|^2\d{p}\d{q} \le \iint_{|p|<|q|} \frac{|E(p)^\zeta-E(q)^\zeta|^2}{E(p)^{2t}}|\hat{T}(p,q)|^2\d{p}\d{q}\\
 =  \iint_{|p|>|q|} \frac{|E(p)^\zeta-E(q)^\zeta|^2}{E(q)^{2t}}|\hat{T}(p,q)|^2\d{p}\d{q}. 
\end{multline*}
Hence, to estimate $\|[|D^0|^\zeta,T]|D^0|^{-t}\|_{\gS_2}$, it is sufficient to estimate $\|[|D^0|^\zeta,T^>]|D^0|^{-t}\|_{\gS_2}$ or $\|[|D^0|^\zeta,T^<]|D^0|^{-t}\|_{\gS_2}$. Now let $A=[|D^0|^\zeta,V_\rho^>]|D^0|^{-t}$. For any $|p|>|q|$, we have by the mean value theorem 
   $$|E(p)^\zeta-E(q)^\zeta|\le \zeta E(q)^{\zeta-1}|E(p)-E(q)|\le\frac{\zeta}{2}E(q)^{\zeta-1}|p-q|.$$
   Hence,
   $$\|A\|_{\gS_2}^2\le\zeta\iint_{\R^3\times\R^3}|p-q|^2|\hat{V_\rho}(p-q)|^2\frac{\d{p}\d{q}}{E(q)^{2(t-\zeta+1)}}\le C\|\rho\|_\cC^2,$$
   since $2(t-\zeta+1)>3$. 
   To prove \eqref{eq:Rcomm}, the same strategy only works in the case $t\ge1/2\ge\zeta$ (notice here that the equality $t=\zeta=1/2$ is authorized). In the general case, first assume $0<\zeta<1/2$. Then, using the formula
   $$\forall 0<\zeta<1,\,\forall x>0,\quad x^\zeta=\frac{\sin(\pi\zeta)}{\pi}\int_0^\ii\left(1-\frac{s}{x+s}\right)\frac{\text{d}s}{s^{1-\zeta}},$$
   we find 
$$[|D^0|^\zeta,R_Q] = \frac{\sin(\pi\zeta)}{\pi}\int_0^\ii\frac{1}{|D^0|+s}[|D^0|,R_Q]\frac{1}{|D^0|+s}s^{\zeta}\d{s}.$$
   We then estimate
 \begin{multline}\label{eq:ineqcommR}
  \|[|D^0|^\zeta,R_Q]|D^0|^{-t}\|_{\gS_2}\le\frac{\sin(\pi\zeta)}{\pi}\int_0^\ii\left\|\frac{|D^0|^{\frac{1}{2}}}{|D^0|+s}\right\|_{L^2\to L^2}\times\\
  \times\| |D^0|^{-\frac{1}{2}}[|D^0|,R_Q]|D^0|^{-\frac{1}{2}}\|_{\gS_2}\left\|\frac{|D^0|^{\frac{1}{2}-t}}{|D^0|+s}\right\|_{L^2\to L^2}s^\zeta\d{s}.
 \end{multline}
 Setting $B= |D^0|^{-\frac{1}{2}}[|D^0|,R_Q]|D^0|^{-\frac{1}{2}}$, one has
 \begin{eqnarray*}
  \|B\|_{\gS_2}^2 & = & \iint_{\R^3\times\R^3} |\hat{R_Q}(p-q)|^2\frac{(E(p)-E(q))^2}{E(p)E(q)}\d{p}\d{q}\\
& \le & C\iint_{\R^3\times\R^3} \frac{E(p-q)^2}{E(p+q)}|\hat{R_Q}(p-q)|^2\d{p}\d{q}\\
& \le & C\|Q\|_\cQ^2,
 \end{eqnarray*}
 where the last inequality was proved in \cite[Lemma 8]{HaiLewSer-05a}. Now assuming $t>\zeta>0$ we have $\tau:=\frac{1}{2}-t\le1$, hence using
\begin{equation}\label{eq:xtau}
\forall \tau\le1,\,\forall s\ge0\quad\sup_{x\ge1}\frac{x^\tau}{x+s}\le\frac{1}{(1+s)^{1-\max(\tau,0)}},
\end{equation}
 we deduce
$$
  \left\|\left[|D^0|^\zeta,R_Q\right]|D^0|^{-t}\right\|_{\gS_2}  \le  C\int_0^\ii\frac{s^\zeta\d{s}}{(1+s)^{\min(t+1,3/2)}}\|Q\|_\cQ,
$$
the right side being finite since $t>\zeta$ and $\zeta<1/2$. If $1/2\le\zeta<1$, we write $|D^0|^\zeta=|D^0|^{\zeta'}|D^0|^{1/2}$ with $0\le\zeta'<1/2$ to infer
$$|D^0|^{\zeta}=\frac{\sin(\pi\zeta')}{\pi}\int_0^{\ii}\frac{|D^0|^{3/2}}{|D^0|+s}\frac{\d{s}}{s^{1-\zeta'}}.$$
Now for all $s\ge0$ we have
$$[\frac{|D^0|^{3/2}}{|D^0|+s},R_Q^<]=\frac{1}{|D^0|+s}(s[|D^0|^{3/2},R_Q^<]+[|D^0|^{1/2},|D^0|R_Q^<|D^0|])\frac{1}{|D^0|+s},$$
If $B_1=|D^0|^{-1/2}[|D^0|^{3/2},R_Q^<]|D^0|^{-1}$, then
\begin{eqnarray*}
 \|B_1\|_{\gS_2}^2 & = & \iint_{|p|<|q|} \frac{|E(p)^{3/2}-E(q)^{3/2}|^2}{E(p)E(q)^2}|\hat{R_Q}(p,q)|^2\d{p}\d{q}\\
  & \le & C\iint_{|p|<|q|}\frac{E(p-q)^2}{E(p+q)}|\hat{R_Q}(p,q)|^2\d{p}\d{q}\le C\|Q\|_{\cQ}^2.
\end{eqnarray*}
If $B_2=|D^0|^{-1}[|D^0|^{1/2},|D^0|R_Q^<|D^0|]|D^0|^{-3/2}$, then
\begin{eqnarray*}
 \|B_2\|_{\gS_2}^2 & = & \iint_{|p|<|q|}\frac{|E(p)^{1/2}-E(q)^{1/2}|^2E(p)^2E(q)^2}{E(p)^2E(q)^3}|\hat{R_Q}(p,q)|^2\d{p}\d{q}\\
 & \le & C\iint_{|p|>|q|}\frac{E(p-q)^2}{E(p+q)}|\hat{R_Q}(p,q)|^2\d{p}\d{q}\le C\|Q\|_{\cQ}^2.
\end{eqnarray*}
We deduce that if $B_3=[|D^0|^\zeta,R_Q^<]|D^0|^{-t}$, we have, using \eqref{eq:xtau} again
\begin{eqnarray*} 
\|B_3\|_{\gS_2} & \le & C\|Q\|_\cQ\int_0^\ii\left(s\|\frac{|D^0|^{1/2}}{|D^0|+s}\|\|\frac{|D^0|^{1-t}}{|D^0|+s}\|+\|\frac{|D^0|}{|D^0|+s}\|\|\frac{|D^0|^{3/2-t}}{|D^0|+s}\|\right)\frac{\d{s}}{s^{1-\zeta'}} \\
 & \le & C\|Q\|_\cQ\int_0^\ii\left(\frac{s^{\zeta'}}{(1+s)^{\min(t+1/2,3/2)}}+\frac{s^{\zeta'-1}}{(1+s)^{\min(t-1/2,1)}}\right)\d{s},
\end{eqnarray*}
the right side being finite since $t>\zeta'+1/2=\zeta$ and $3/2-\zeta'>1$. 
\end{proof}

The next proposition is the analogue of \cite[Proposition 10]{HaiLewSer-05a}. It is the main tool to prove the contractivity of the map $\Phi$. 

\begin{proposition}\label{prop:main}
Let $\omega\in L^2(\R^3)\cap\cC$, $\mu\in(-1,1)$, $\alpha\ge0$, $\Lambda>1$, $(Q,\rho')\in\cX$ be such that the following conditions are satisfied:

\begin{enumerate}
 \item $\mu\notin\sigma(\Dref(\omega))$;
 \item $\mu\notin\sigma(D_{\omega,\rho',Q})$;
 \item $8\alpha\|\omega\|_\cC<1-|\mu|$.
\end{enumerate}

Then, we have the following estimates:
\begin{multline*}
  \left\|(\Phi_1^\omega(Q,\rho'),\Phi_2^\omega(Q,\rho'))\right\|_\cX\le2\sqrt{\pi}\left\|\rhoref\right\|_{L^2\cap\cC}+\alpha\sqrt{\log\Lambda}\Xi\left\|(Q,\rho')\right\|_\cX\\
  +\sum_{n\ge2}\Big(\alpha\Xi\left\|(Q,\rho')\right\|_\cX\Big)^n,
\end{multline*}
 \begin{equation}\label{eq:estimate_derivative}
   \left\|((\Phi_1^\omega)'(Q,\rho'),(\Phi_2^\omega)'(Q,\rho'))\right\|\le \alpha\sqrt{\log\Lambda}\Xi+\alpha\sum_{n\ge2}n\Big(\alpha\Xi\left\|(Q,\rho')\right\|_\cX\Big)^{n-1},
 \end{equation}
 with
 \begin{equation}
  \Xi:=C(\mu)(1+\alpha\sqrt{\log\Lambda})^{14}\left(1+\|\omega\|_{L^2\cap\cC}^{14}\right)\theta(\mu,\omega)^6,
 \end{equation}
  where $C(\mu)>0$ is some constant only depending on $\mu$, and $\theta(\mu,\omega)$ is defined in Lemma \ref{lemma:D0Dref-1L2L2}. 
\end{proposition}

\begin{proof}
  We start by estimating $\Phi_1^\omega=\sum_{n\ge1}\alpha^nG_n^\omega$ in norm $\|\cdot\|_\cQ$. The difficulty relies in the fact that the definition of the norm $\|\cdot\|_\cQ$ uses the Fourier representation of the kernel of the operator considered. In \cite{HaiLewSer-05a}, it was not an issue since all the studied operators, for instance the resolvent $(D^0_\Lambda+i\eta)^{-1}$, have an explicit kernel in Fourier space. This is not the case in our setting: the kernel in Fourier space of the resolvent $(\Dref(\omega)-\mu+i\eta)^{-1}$ is not explicit. It is thus harder to estimate the $\cQ$-norms of these operators. To circumvent this difficulty, we will expand the resolvent $(\Dref(\omega)-\mu+i\eta)^{-1}$ in terms of the resolvent $(D^0_\Lambda-\mu+i\eta)^{-1}$ and of $\omega$, in \emph{each} $G_n^\omega$ separately. The terms only containing resolvents of the form $(D^0_\Lambda-\mu+i\eta)^{-1}$ can be estimated via the results of \cite{HaiLewSer-05a}, while the terms still containing 
resolvents of the form $(\Dref(\omega)-\mu+i\eta)^{-1}$ decay faster and can be estimated using inequalities of the form $\|Q\|_\cQ\ \lesssim\| |D^0|^{\frac{3}{2}}Q\|_{\gS_2}$. Indeed, the norm $\| |D^0|^{\frac{3}{2}}\cdot\|_{\gS_2}$ is a rough upper bound for the $\cQ$-norm, but has the advantage of being computable without using the kernel in Fourier space of the operator. 
  We thus expand the resolvent as
  $$\forall\eta\in\R,\quad\frac{1}{\Dref(\omega)-\mu+i\eta}=A_0(\eta)+A_1(\eta),$$
  with 
  $$A_0(\eta)=\frac{1}{D^0_\Lambda-\mu+i\eta}\sum_{k=0}^4\left(Y(\omega,\alpha)\frac{1}{D^0_\Lambda-\mu+i\eta}\right)^k,$$
  \begin{equation}\label{eq:def-A1eta}
 A_1(\eta)=\left(\frac{1}{D^0_\Lambda-\mu+i\eta}Y(\omega,\alpha)\right)^5\frac{1}{\Dref(\omega)-\mu+i\eta},  
  \end{equation}
  $$Y(\omega,\alpha):=\Pi_\Lambda(V_\omega-\alpha R_{1,0}(\omega))\Pi_\Lambda.$$
  Let $n\ge1$ and write
 \begin{eqnarray*}
G_n^\omega(Q,\rho') & = & \frac{(-1)^{n+1}}{2\pi}\sum_{x\in\{0,1\}^{n+1}}\int_\R A_{x_0}(\eta)\prod_{j=1}^n\Big(XA_{x_j}(\eta)\Big)\d{\eta}\\
 & = & \frac{(-1)^{n+1}}{2\pi}\sum_{x\in\{0,1\}^{n+1}}G_{n,x},  
 \end{eqnarray*}
 where we used the notation $X=V_{\rho'}-R_{Q}$. The term $G_{n,(0,\ldots,0)}$ was already estimated in \cite[Proposition 10]{HaiLewSer-05a},
 \begin{equation}\label{eq:estCMP}
\|G_{n,(0,\ldots,0)}\|_\cQ\le \kappa_n(1+(\alpha\sqrt{\log\Lambda})^4)^{n+1}(1+\|\omega\|_{L^2\cap\cC}^4)^{n+1}\|(Q,\rho')\|^n, 
 \end{equation}
 where $\kappa_n=C(\mu)^n$ for $n\ge2$ and $\kappa_1=C(\mu)\sqrt{\log\Lambda}$, for some number $C(\mu)$ only depending on $\mu$. Actually, the estimate \eqref{eq:estCMP} was obtained for $\mu=0$ in \cite{HaiLewSer-05a}. The way to obtain it for any $\mu$ is to replace Equation (58) in \cite{HaiLewSer-05a} by
 \begin{equation}
  \frac{1}{\sqrt{(E(p)-\mu)^2+\eta^2}}\frac{1}{\sqrt{(E(q)-\mu)^2+\eta^2}}\le\frac{2C'(\mu)}{E(p+q)E(\eta)},\quad\forall p,q\in\R^3,\:\forall\eta\in\R,
 \end{equation}
 for some $C'(\mu)$. It remains to treat $G_{n,x}$ with $x\neq(0,\ldots,0)$. Using that for any $p,q\in\R^3$ one has 
 $$E(p-q)^2E(p+q)\le E(p)^3+2E(p)^2E(q)+2E(p)E(q)^2+E(q)^3,$$
 we deduce that 
 \begin{equation}\label{eq:ineqnormQ}
\|Q\|_\cQ^2\le \||D^0|^{\frac{3}{2}}Q\|_{\gS_2}^2+2\||D^0|Q|D^0|^{\frac{1}{2}}\|_{\gS_2}^2+2\||D^0|^{\frac{1}{2}}Q|D^0|\|_{\gS_2}^2+\|Q|D^0|^{\frac{3}{2}}\|_{\gS_2}^2,
 \end{equation}
 for any $Q\in\cQ$ such that the right side is well defined. To estimate $\|G_{n,x}\|_\cQ$, we will estimate the right side of \eqref{eq:ineqnormQ} with $Q=G_{n,x}$. 
\begin{lemma}[Estimates on $A_0(\eta)$, $A_1(\eta)$]\label{lemma:estA0A1}
 Let $\eta\in\R$. We have the following estimates.
  \begin{equation}\label{eq:maxmax}
  \left\{\begin{array}{l}
   \|A_0(\eta)\|+\||D^0|A_0(\eta)\|\le C(\mu)(1+(\alpha\sqrt{\log\Lambda})^4)(1+\|\omega\|_{L^2\cap\cC}^4),\\
   \|A_1(\eta)\|+\||D^0|A_1(\eta)\|\le C(\mu)(1+(\alpha\sqrt{\log\Lambda})^5)\|\omega\|_{L^2\cap\cC}^5\theta(\mu,\omega).
   \end{array}\right.
 \end{equation}
For all $0\le\sigma<\frac{1}{2}$, we have
 \begin{equation}\label{eq:estD032+sigmaA1}
  \left\||D^0|^{\frac{3}{2}+\sigma}A_1(\eta)\right\|_{\gS_2}\le\frac{C(\mu)(1+\alpha\sqrt{\log\Lambda})^5\|\omega\|^5_{L^2\cap\cC}\theta(\mu,\omega)}{E(\eta)^{3/2}},
 \end{equation}
  and for any $B\in\gS_2$ such that $|D^0|^{\frac{3}{2}+\sigma}B\in\gS_2$, we have
   \begin{equation}\label{eq:D032B}
  \left\||D^0|^{\frac{3}{2}+\sigma}A_0(\eta)XB\right\|_{\gS_2}\le C(\mu)(1+\alpha\sqrt{\log\Lambda})^4\left(1+\|\omega\|_{L^2\cap\cC}^4\right)\|(Q,\rho')\|_\cX\left\||D^0|^{\frac{3}{2}+\sigma}B\right\|_{\gS_2}.
 \end{equation}
We also have
 \begin{equation}\label{eq:estA1D032}
  \left\|A_1(\eta)|D^0|^{\frac{3}{2}}\right\|_{\gS_2}\le\frac{C(\mu)(1+\alpha\sqrt{\log\Lambda})^7\|\omega\|_{L^2\cap\cC}^7\theta(\mu,\omega)^3}{E(\eta)^{3/2}},
 \end{equation}
 and for any $B\in\gS_2$ such that $B|D^0|^{\frac{3}{2}}\in\gS_2$, we have
 \begin{equation}\label{eq:BD032}
  \left\|BXA_0(\eta)|D^0|^{\frac{3}{2}}\right\|_{\gS_2}\le C(\mu)(1+\alpha\sqrt{\log\Lambda})^4\left(1+\|\omega\|_{L^2\cap\cC}^4\right)\|(Q,\rho')\|_\cX\left\|B|D^0|^{\frac{3}{2}}\right\|_{\gS_2}.
 \end{equation}
Finally, we have
 \begin{equation}\label{eq:estD0A165}
  \left\||D^0|A_1(\eta)\right\|_{\gS_{6/5}}\le\frac{C(\mu)(1+\alpha\sqrt{\log\Lambda})^5\|\omega\|_{L^2\cap\cC}^5\theta(\mu,\omega)}{E(\eta)^{5/2}},
 \end{equation}
and for any $B$ such that $|D^0|B\in\gS_{6/5}$, we have
 \begin{equation}\label{eq:estD0B}
 \left\||D^0|A_0(\eta)XB\right\|_{\gS_{6/5}}\le C(\mu)(1+\alpha\sqrt{\log\Lambda})^4(1+\|\omega\|_{L^2\cap\cC}^4)\|(Q,\rho')\|_\cX\Big\||D^0|B\Big\|_{\gS_{6/5}}, 
 \end{equation}
\end{lemma}

\begin{proof}[Proof of Lemma \ref{lemma:estA0A1}]
 To prove \eqref{eq:maxmax}, just use $\||D^0|(D^0_\Lambda-\mu+i\eta)^{-1}\|_{L^2\to L^2}\le C(\mu)$, $\|Y(\omega,\alpha)|D^0|^{-1/2}\|_{L^2\to L^2}\le C(1+\alpha\sqrt{\log\Lambda})\|\omega\|_{L^2\cap\cC}$ and $\||D^0|(\Dref(\omega)-\mu+i\eta)^{-1}\|_{L^2\to L^2}\le\theta(\mu,\omega)$. To prove \eqref{eq:estD032+sigmaA1}, we write $|D^0|^{\frac{3}{2}+\sigma}A_1(\eta)=B_1(\eta)+B_2(\eta)$ with 
 \begin{multline*}
B_1(\eta):=\frac{|D^0|}{D^0_\Lambda-\mu+i\eta}Y(\omega,\alpha)\frac{|D^0|^{\frac{1}{2}+\sigma}}{D^0_\Lambda-\mu+i\eta}Y(\omega,\alpha)\times\\
\times\left(\frac{1}{D^0_\Lambda-\mu+i\eta}Y(\omega,\alpha)\right)^3\frac{1}{\Dref(\omega)-\mu+i\eta},  
 \end{multline*}
 $$B_2(\eta):=\frac{|D^0|}{D^0_\Lambda-\mu+i\eta}\left[|D^0|^{\frac{1}{2}+\sigma},Y(\omega,\alpha)\right]\left(\frac{1}{D^0_\Lambda-\mu+i\eta}Y(\omega,\alpha)\right)^4\frac{1}{\Dref(\omega)-\mu+i\eta}.$$
Now $B_1(\eta)=B_{11}(\eta)+B_{12}(\eta)$ with 
\begin{multline*}
 B_{11}(\eta)=\frac{|D^0|}{D^0_\Lambda-\mu+i\eta}Y(\omega,\alpha)|D^0|^{-\frac{1}{2}}\frac{|D^0|}{D^0_\Lambda-\mu+i\eta}Y(\omega,\alpha)|D^0|^{-\frac{1}{2}}\frac{|D^0|^{\frac{1}{2}+\sigma}}{D^0_\Lambda-\mu+i\eta}\times\\
\times\left(Y(\omega,\alpha)\frac{1}{D^0_\Lambda-\mu+i\eta}\right)^3(D^0_\Lambda-\mu+i\eta)\frac{1}{\Dref(\omega)-\mu+i\eta},
\end{multline*}
\begin{multline*}
 B_{12}(\eta)=\frac{|D^0|}{D^0_\Lambda-\mu+i\eta}Y(\omega,\alpha)|D^0|^{-\frac{1}{2}}\frac{|D^0|}{D^0_\Lambda-\mu+i\eta}\left[|D^0|^\sigma,Y(\omega,\alpha)\right]|D^0|^{-t}\times\\
\times\frac{|D^0|^t}{D^0_\Lambda-\mu+i\eta}\left(Y(\omega,\alpha)\frac{1}{D^0_\Lambda-\mu+i\eta}\right)^3(D^0_\Lambda-\mu+i\eta)\frac{1}{\Dref(\omega)-\mu+i\eta},
\end{multline*}
for some $\frac{1}{2}+\sigma<t<1$. Using $|D^0||D^0-\mu+i\eta|^{-1}\le C(\mu)$, $\|Y(\omega,\alpha)|D^0|^{-\frac{1}{2}}\|_{L^2\to L^2}\le C(\mu)(1+\alpha\sqrt{\log\Lambda})\|\omega\|_{L^2\cap\cC}$ by \eqref{eq:ineq-VLii} and \eqref{eq:ineq-RQD0-12L2L2}, $\|(D^0_\Lambda-\mu+i\eta)(\Dref(\omega)-\mu+i\eta)^{-1}\|_{L^2\to L^2}\le C(\mu)\theta(\mu,\omega)$, $\frac{1}{2}+\sigma<1$, and the two inequalities
 \begin{equation}\label{eq:estVnuD0-1S6}
\left\|V_\omega\frac{1}{D^0-\mu+i\eta}\right\|_{\gS_6}\le \frac{C(\mu)\|\omega\|_\cC}{E(\eta)^{1/2}},  
 \end{equation}
obtained by \eqref{eq:ineq-VD0-1S6},
\begin{equation}\label{eq:RD0+ieta}
 \left\|R_{1,0}(\omega)\frac{1}{D^0-\mu+i\eta}\right\|_{\gS_2}\le\frac{C(\mu)\sqrt{\log\Lambda}\|\omega\|_{L^2\cap\cC}}{E(\eta)^{1/2}},
\end{equation}
obtained by \eqref{eq:ineq-RQD0-1S2} and \eqref{eq:ineq-G10Q}, we get
\begin{equation}
\left\|Y(\omega,\alpha)\frac{1}{D^0-\mu+i\eta}\right\|_{\gS_6}\le \frac{C(\mu)(1+\alpha\sqrt{\log\Lambda})\|\omega\|_{L^2\cap\cC}}{E(\eta)^{1/2}},   
\end{equation}
hence
 $$\|B_{11}(\eta)\|_{\gS_2}\le \frac{C(\mu)(1+\alpha\sqrt{\log\Lambda})^5\|\omega\|_{L^2\cap\cC}^5\theta(\mu,\omega)}{E(\eta)^{3/2}}.$$
The term $B_{12}(\eta)$ is treated the same way, except that we use the commutator estimates of Lemma \ref{lemm:comm}, to get
$$\|\left[|D^0|^\sigma,Y(\omega,\alpha)\right]|D^0|^{-t}\|_{\gS_2}\le C(\mu)(1+\alpha\sqrt{\log\Lambda})\|\omega\|_{L^2\cap\cC}.$$
 We thus find
$$\|B_{12}(\eta)\|_{\gS_2}\le \frac{C(\mu)(1+\alpha\sqrt{\log\Lambda})^5\|\omega\|_{L^2\cap\cC}^5\theta(\mu,\omega)}{E(\eta)^{3/2}}.$$
 We next set $B_2(\eta)=B_{21}(\eta)+B_{22}(\eta)$ with
 \begin{multline*}
 B_{21}(\eta)=\frac{|D^0|}{D^0_\Lambda-\mu+i\eta}\left[|D^0|^{\frac{1}{2}+\sigma},Y(\omega,\alpha)\right]|D^0|^{-s}\frac{|D^0|}{D^0_\Lambda-\mu+i\eta}Y(\omega,\alpha)\frac{|D^0|^{s-1}}{D^0_\Lambda-\mu+i\eta}\times\\
 \times\left(Y(\omega,\alpha)\frac{1}{D^0_\Lambda-\mu+i\eta}\right)^3(D^0_\Lambda-\mu+i\eta)\frac{1}{\Dref(\omega)-\mu+i\eta}
\end{multline*}
and
\begin{multline*}
 B_{22}(\eta) =  \frac{|D^0|}{D^0_\Lambda-\mu+i\eta}\left[|D^0|^{\frac{1}{2}+\sigma},Y(\omega,\alpha)\right]|D^0|^{-s}\frac{|D^0|}{D^0_\Lambda-\mu+i\eta}\left[|D^0|^{s-1},Y(\omega,\alpha)\right]|D^0|^{-u}\times\\
 \times\frac{|D^0|^u}{D^0_\Lambda-\mu+i\eta}\left(Y(\omega,\alpha)\frac{1}{D^0_\Lambda-\mu+i\eta}\right)^3(D^0_\Lambda-\mu+i\eta)\frac{1}{\Dref(\omega)-\mu+i\eta},
\end{multline*}
for some $1+\sigma<s<\frac{3}{2}$ and $s-\frac{1}{2}<u<1$. Such a couple $(s,u)$ exists since $\sigma<\frac{1}{2}$. Using the estimates as before, we obtain 
$$ \|B_2(\eta)\|_{\gS_2}\le\frac{C(\mu)(1+\alpha\sqrt{\log\Lambda})^5\|\omega\|^5_{L^2\cap\cC}\theta(\mu,\omega)}{E(\eta)^{3/2}}.
$$
We thus have proved \eqref{eq:estD032+sigmaA1}. Let us now prove \eqref{eq:D032B}. Hence, let $B\in\gS_2$ such that $|D^0|^{\frac{3}{2}+\sigma}B\in\gS_2$. Recall that $A_0(\eta)$ is the sum of five terms. We will treat the first two independently. We thus write 
 $$|D^0|^{\frac{3}{2}+\sigma}A_0(\eta)XB=H_1+H_2+H_3.$$
 The first term to estimate is 
 $$H_1:=\frac{|D^0|^{\frac{3}{2}+\sigma}}{D^0_\Lambda-\mu+i\eta}XB=\frac{|D^0|}{D^0_\Lambda-\mu+i\eta}X|D^0|^{\frac{1}{2}+\sigma}B+\frac{|D^0|}{D^0_\Lambda-\mu+i\eta}\left[|D^0|^{\frac{1}{2}+\sigma},X\right]B.$$
 By Lemma \ref{lemm:comm}, we know that 
 \begin{equation}\label{eq:XD0-1}
\|X|D^0|^{-1}\|_{L^2\to L^2}\le C\|(Q,\rho')\|_\cX,
 \end{equation}
 $$\left\|\left[|D^0|^{\frac{1}{2}+\sigma},X\right]|D^0|^{-\frac{3}{2}-\sigma}\right\|_{L^2\to L^2}\le C\|(Q,\rho')\|_\cX.$$
 Hence,
 $$\|H_1\|_{\gS_2}\le C(\mu)\|(Q,\rho')\|_\cX\left\||D^0|^{\frac{3}{2}+\sigma}B\right\|_{\gS_2}.$$
 Next, we treat
 \begin{eqnarray*}
H_2 & := & \frac{|D^0|^{\frac{3}{2}+\sigma}}{D^0_\Lambda-\mu+i\eta}Y(\omega,\alpha)\frac{1}{D^0_\Lambda-\mu+i\eta}XB\\
 & = & \frac{|D^0|}{D^0_\Lambda-\mu+i\eta}Y(\omega,\alpha)\frac{|D^0|^{\frac{1}{2}+\sigma}}{D^0_\Lambda-\mu+i\eta}XB\\
 & & +\frac{|D^0|}{D^0_\Lambda-\mu+i\eta}\left[|D^0|^{\frac{1}{2}+\sigma},Y(\omega,\alpha)\right]\frac{1}{D^0_\Lambda-\mu+i\eta}XB\\
 & = & H_{21}+H_{22}.
 \end{eqnarray*}
The first term is easily estimated by
 \begin{multline*}
\|H_{21}\|_{\gS_2}\le C(\mu)\|Y(\omega,\alpha)|D^0|^{-1}\|_{L^2\to L^2}\|H_1\|_{\gS_2}\\
\le C(\mu)(1+\alpha\sqrt{\log\Lambda})\|\omega\|_{L^2\cap\cC}\|(Q,\rho')\|_\cX\left\||D^0|^{\frac{3}{2}+\sigma}B\right\|_{\gS_2},  
 \end{multline*}
while we notice that the second term can be estimated by
\begin{eqnarray*}
\|H_{22}\|_{\gS_2} & \le & C(\mu)\left\|\left[|D^0|^{\frac{1}{2}+\sigma},Y(\omega,\alpha)\right]|D^0|^{-\frac{3}{2}+\sigma}\right\|_{L^2\to L^2} \|H_1\|_{\gS_2}\\
 & \le &  C(\mu)(1+\alpha\sqrt{\log\Lambda})\|\omega\|_{L^2\cap\cC}\|(Q,\rho')\|_\cX\left\||D^0|^{\frac{3}{2}+\sigma}B\right\|_{\gS_2}. 
\end{eqnarray*}
We conclude that 
$$\|H_2\|_{\gS_2}\le C(\mu)(1+\alpha\sqrt{\log\Lambda})\|\omega\|_{L^2\cap\cC}\|(Q,\rho')\|_\cX\left\||D^0|^{\frac{3}{2}+\sigma}B\right\|_{\gS_2}.$$
The last term to control is 
$$H_3:=\frac{|D^0|^{\frac{3}{2}+\sigma}}{D^0_\Lambda-\mu+i\eta}\left(Y(\omega,\alpha)\frac{1}{D^0_\Lambda-\mu+i\eta}\right)^2\sum_{k=0}^2\left(Y(\omega,\alpha)\frac{1}{D^0_\Lambda-\mu+i\eta}\right)^k XB.$$
We use that $\|H_3\|_{\gS_2}\le\|H_{31}\|_{L^2\to L^2}\|H_{32}\|_{\gS_2}$, with
$$H_{31}:=\frac{|D^0|^{\frac{3}{2}+\sigma}}{D^0_\Lambda-\mu+i\eta}\left(Y(\omega,\alpha)\frac{1}{D^0_\Lambda-\mu+i\eta}\right)^2,$$
$$H_{32}:=\sum_{k=0}^2\left(Y(\omega,\alpha)\frac{1}{D^0_\Lambda-\mu+i\eta}\right)^k XB.$$
On the one hand we have 
$$\|H_{32}\|_{\gS_2}\le C(\mu)(1+\alpha\sqrt{\log\Lambda})^2\left(1+\|\omega\|_{L^2\cap\cC}^2\right)\|(Q,\rho')\|_\cX\left\||D^0|^{\frac{3}{2}+\sigma}B\right\|_{\gS_2},$$
while on the other hand 
\begin{multline*}
 H_{31}=\frac{|D^0|}{D^0_\Lambda-\mu+i\eta}Y(\omega,\alpha)\frac{|D^0|^{\frac{1}{2}+\sigma}}{D^0_\Lambda-\mu+i\eta}Y(\omega,\alpha)\frac{1}{D^0_\Lambda-\mu+i\eta}\\
 +\frac{|D^0|}{D^0_\Lambda-\mu+i\eta}\left[|D^0|^{\frac{1}{2}+\sigma},Y(\omega,\alpha)\right]|D^0|^{-t}\frac{|D^0|}{D^0_\Lambda-\mu+i\eta}Y(\omega,\alpha)|D^0|^{-\frac{1}{2}}\frac{|D^0|^{t-\frac{1}{2}}}{D^0_\Lambda-\mu+i\eta}\\
 +\frac{|D^0|}{D^0_\Lambda-\mu+i\eta}\left[|D^0|^{\frac{1}{2}+\sigma},Y(\omega,\alpha)\right]|D^0|^{-t}\frac{|D^0|}{D^0_\Lambda-\mu+i\eta}\left[|D^0|^{t-1},Y(\omega,\alpha)\right]|D^0|^{-s}\frac{|D^0|^{s}}{D^0_\Lambda-\mu+i\eta},
\end{multline*}
for some $1+\sigma<t<\frac{3}{2}$ and $t-\frac{1}{2}<s<1$, so that writing
\begin{multline*}
\frac{|D^0|}{D^0_\Lambda-\mu+i\eta}Y(\omega,\alpha)\frac{|D^0|^{\frac{1}{2}+\sigma}}{D^0_\Lambda-\mu+i\eta}Y(\omega,\alpha)\frac{1}{D^0_\Lambda-\mu+i\eta}\\
=\frac{|D^0|}{D^0_\Lambda-\mu+i\eta}Y(\omega,\alpha)|D^0|^{-\frac{1}{2}}\frac{|D^0|}{D^0_\Lambda-\mu+i\eta}Y(\omega,\alpha)|D^0|^{-\frac{1}{2}}\frac{|D^0|^{\frac{1}{2}+\sigma}}{D^0_\Lambda-\mu+i\eta}\\
+\frac{|D^0|}{D^0_\Lambda-\mu+i\eta}Y(\omega,\alpha)|D^0|^{-\frac{1}{2}}\frac{|D^0|}{D^0_\Lambda-\mu+i\eta}\left[|D^0|^\sigma,Y(\omega,\alpha)\right]|D^0|^{-u}\frac{|D^0|^u}{D^0_\Lambda-\mu+i\eta},
\end{multline*}
for some $\frac{1}{2}+\sigma<u<1$, one has
$$\|H_{31}\|_{L^2\to L^2}\le C(\mu)(1+\alpha\sqrt{\log\Lambda})^2\|\omega\|_{L^2\cap\cC}^2.$$
Hence,
$$\|H_3\|_{\gS_2}\le C(\mu)(1+\alpha\sqrt{\log\Lambda})^4\left(1+\|\omega\|_{L^2\cap\cC}^4\right)\|(Q,\rho')\|_\cX\left\||D^0|^{\frac{3}{2}+\sigma}B\right\|_{\gS_2},$$
which finishes the proof of \eqref{eq:D032B}. Let us thus pass to \eqref{eq:estA1D032}. The proof is the same as \eqref{eq:estD032+sigmaA1}, with a additional commutator between $|D^0|^{\frac{1}{2}}$ and $(\Dref(\omega)-\mu+i\eta)^{-1}$. Indeed, we have
\begin{multline*}
A_1(\eta)|D^0|^\frac{3}{2}=\left(\frac{1}{D^0_\Lambda-\mu+i\eta}Y(\omega,\alpha)\right)^5|D^0|^{\frac{1}{2}}\frac{1}{\Dref(\omega)-\mu+i\eta}|D^0|\\
+\left(\frac{1}{D^0_\Lambda-\mu+i\eta}Y(\omega,\alpha)\right)^5\left[\frac{1}{\Dref(\omega)-\mu+i\eta},|D^0|^{\frac{1}{2}}\right]|D^0|=T_1+T_2. 
\end{multline*}
We have by the same method as before
$$\|T_1\|_{\gS_2}\le\frac{C(\mu)(1+\alpha\sqrt{\log\Lambda})^5\|\omega\|_{L^2\cap\cC}^5\theta(\mu,\omega)}{E(\eta)^{3/2}}.$$
Now notice that for any $\zeta\in\R$,
$$\left[\frac{1}{\Dref(\omega)-\mu+i\eta},|D^0|^\zeta\right]=\frac{1}{\Dref(\omega)-\mu+i\eta}\left[|D^0|^\zeta,Y(\omega,\alpha)\right]\frac{1}{\Dref(\omega)-\mu+i\eta},$$
so that $T_2=T_{21}+T_{22}$ with
\begin{multline*}
 T_{21}=\left(\frac{1}{D^0_\Lambda-\mu+i\eta}Y(\omega,\alpha)\right)^5|D^0|^{t-1}\frac{1}{\Dref(\omega)-\mu+i\eta}|D^0|\times\\
\times|D^0|^{-t}\left[|D^0|^\frac{1}{2},Y(\omega,\alpha)\right]\frac{1}{\Dref(\omega)-\mu+i\eta}|D^0|,
\end{multline*}
\begin{multline*}
T_{22}=\left(\frac{1}{D^0_\Lambda-\mu+i\eta}Y(\omega,\alpha)\right)^5\frac{1}{\Dref(\omega)-\mu+i\eta}\left[|D^0|^{t-1},Y(\omega,\alpha)\right]\times\\
\times\frac{1}{\Dref(\omega)-\mu+i\eta}|D^0||D^0|^{-t}\left[|D^0|^\frac{1}{2},Y(\omega,\alpha)\right]\frac{1}{\Dref(\omega)-\mu+i\eta}|D^0|,
\end{multline*}
for some $1<t<\frac{3}{2}$. Hence by the methods introduced before, we deduce that
$$\|T_2\|_{\gS_2}\le\frac{C(\mu)(1+\alpha\sqrt{\log\Lambda})^7\|\omega\|_{L^2\cap\cC}^7\theta(\mu,\omega)^3}{E(\eta)^{3/2}},$$
leading to \eqref{eq:estA1D032}. The proof of \eqref{eq:BD032} is the same as the proof of \eqref{eq:D032B}. Next, we prove \eqref{eq:estD0A165}. It is because of this term that we have to expand $(\Dref(\omega)-\mu+i\eta)^{-1}$ to the 5th order, as we can see from the following  estimate:
\begin{eqnarray*}
 \||D^0|A_1(\eta)\|_{\gS_{6/5}} & \le & \left\|\frac{|D^0|}{D^0_\Lambda-\mu+i\eta}\right\|_{L^2\to L^2}\left\|Y(\omega,\alpha)\frac{1}{D^0_\Lambda-\mu+i\eta}\right\|_{\gS_6}^5\times\\
 & & \times\left\|(D^0_\Lambda-\mu+i\eta)\frac{1}{\Dref(\omega)-\mu+i\eta}\right\|_{L^2\to L^2}\\
  & \le & \frac{C(\mu)(1+\alpha\sqrt{\log\Lambda})^5\|\omega\|_{L^2\cap\cC}^5\theta(\mu,\omega)}{E(\eta)^{5/2}},
\end{eqnarray*}
which is exactly \eqref{eq:estD0A165}. Finally, to prove \eqref{eq:estD0B}, we simply notice that
\begin{eqnarray*}
\left\||D^0|A_0(\eta)XB\right\|_{\gS_{6/5}} & \le & \||D^0|A_0(\eta)\|\|X|D^0|^{-1}\|\||D^0|B\|_{\gS_{6/5}}\\
 & \le & C(\mu)(1+\alpha\sqrt{\log\Lambda})^4(1+\|\omega\|_{L^2\cap\cC}^4)\|(Q,\rho')\|_\cX\||D^0|B\|_{\gS_{6/5}}, 
\end{eqnarray*}
by inequalities \eqref{eq:maxmax} and \eqref{eq:XD0-1}. This ends the proof of Lemma \ref{lemma:estA0A1}.
\end{proof}

Let us now estimate $\||D^0|^{\frac{3}{2}}G_{n,x}\|_{\gS_2}$. First, assume that $x_0=1$. We have 
 \begin{equation}\label{eq:D032Gnx}
  \||D^0|^{\frac{3}{2}}G_{n,x}\|_{\gS_2}\le\int_\R\||D^0|^{\frac{3}{2}}A_1(\eta)\|_{\gS_2}\prod_{j=1}^n\|X A_{x_j}(\eta)\|_{L^2\to L^2}\d{\eta}.
 \end{equation}
 We deduce from \eqref{eq:maxmax} that for all $1\le j\le n$,
 \begin{eqnarray*}
  \|XA_{x_j}(\eta)\|_{L^2\to L^2} & \le & \|V_{\rho'}A_{x_j}(\eta)\|_{L^2\to L^2}+\|R_{Q}A_{x_j}(\eta)\|_{L^2\to L^2} \\
 & \le & \|V_{\rho'}\|_{L^2\to L^2} \|A_{x_j}(\eta)\|_{L^2\to L^2}+\|R_{Q}|D^0|^{-1}\|_{\gS_2}\||D^0|A_{x_j}(\eta)\|_{L^2\to L^2}\\
 & \le & C(\mu)(1+\alpha\sqrt{\log\Lambda})^5(1+\|\omega\|_{L^2\cap\cC}^5)\theta(\mu,\omega)\|(Q,\rho')\|_\cX.
 \end{eqnarray*}	
 Hence, for all $\eta\in\R$, 
 \begin{equation}\label{eq:prodXAj}
  \prod_{j=1}^n\|X A_{x_j}(\eta)\|_{L^2\to L^2}\le\left(C(\mu)(1+\alpha\sqrt{\log\Lambda})^5(1+\|\omega\|_{L^2\cap\cC}^5)\theta(\mu,\omega)\right)^n\|(Q,\rho')\|_\cX^n.
 \end{equation}
  Now by \eqref{eq:estD032+sigmaA1} applied with $\sigma=0$, we have
$$\int_\R \left\||D^0|^{\frac{3}{2}}A_1(\eta)\right\|_{\gS_2}\d{\eta} \le C(\mu)(1+\alpha\sqrt{\log\Lambda})^5\|\omega\|^5_{L^2\cap\cC}\theta(\mu,\omega).$$
 Hence, for all $n\ge1$ and for all $x\in\{0,1\}^{n+1}$ such that $x_0=1$, we have 
\begin{equation}\label{eq:x0=1}
\||D^0|^{\frac{3}{2}}G_{n,x}\|_{\gS_2}\le \left(C(\mu)(1+\alpha\sqrt{\log\Lambda})^5(1+\|\omega\|_{L^2\cap\cC}^5)\theta(\mu,\omega)\right)^{n+1}\|(Q,\rho')\|_\cX^n.
\end{equation}
 It remains to treat the case $x\neq(0,\ldots,0)$ with $x_0\neq1$. Let $n\ge1$, $x\in\{0,1\}^{n+1}$ such that $x\neq(0,\ldots,0)$ and $x_0=0$. Then there exists a smallest index $j_0\ge1$ such that $x_{j_0}=1$. We can thus write
$$G_{n,x}=\int_\R \left(A_0(\eta)X\right)^{j_0}A_1(\eta)\prod_{j=j_0+1}^n\left(XA_{x_j}(\eta)\right)\d{\eta}.$$
Using \eqref{eq:D032B} $j_0$ times with $\sigma=0$ we get
\begin{multline*}
\left\||D^0|^{\frac{3}{2}}G_{n,x}\right\|_{\gS_2}\le\left(C(\mu)(1+\alpha\sqrt{\log\Lambda})^4(1+\|\omega\|_{L^2\cap\cC}^4)\right)^{j_0}\|(Q,\rho')\|_\cX^{j_0}\times\\
\times\int_\R\||D^0|^{\frac{3}{2}}A_1(\eta)\prod_{j=j_0+1}^n\left(X A_{x_j}(\eta)\right)\|_{\gS_2}\d{\eta}. 
\end{multline*}
We estimate the right side using the estimate \eqref{eq:x0=1} in the case $x_0=1$, and we obtain 
 \begin{multline*}
\left\||D^0|^{\frac{3}{2}}G_{n,x}\right\|_{\gS_2}  \le  \left(C(\mu)(1+\alpha\sqrt{\log\Lambda})^4(1+\|\omega\|_{L^2\cap\cC}^4)\right)^{j_0}\|(Q,\rho')\|_\cX^{j_0}\times\\
\times\left(C(\mu)(1+\alpha\sqrt{\log\Lambda})^5(1+\|\omega\|_{L^2\cap\cC}^5)\theta(\mu,\omega)\right)^{n+1-j_0}\|(Q,\rho')\|_\cX^{n-j_0}.
 \end{multline*}
 Summing up, we have for all $n\ge1$ and all $x$,
 \begin{equation}
  \left\||D^0|^{\frac{3}{2}}G_{n,x}\right\|_{\gS_2}\le\left(C(\mu)(1+\alpha\sqrt{\log\Lambda})^5(1+\|\omega\|_{L^2\cap\cC}^5)\theta(\mu,\omega)\right)^{n+1}\|(Q,\rho')\|_\cX^n.
 \end{equation}
 We now estimate $\left\||D^0|G_{n,x}|D^0|^\frac{1}{2}\right\|_{\gS_2}$ and $\left\||D^0|^\frac{1}{2}G_{n,x}|D^0|\right\|_{\gS_2}$ for all $n\ge1$ and $x\neq(0,\ldots,0)$. We treat both cases at the same time by estimating
 $$\left\||D^0|^sG_{n,x}|D^0|^t\right\|_{\gS_2},$$
 for any $s,t\in[0,1]$. For $x\neq(0,\ldots,0)$, there exists a smallest index $0\le j_0\le n$ such that $x_{j_0}=1$. We distinguish three cases: $j_0=0$, $j_0=n$, and $1\le j_0\le n-1$. First, if $j_0=0$, then
 $$\left\||D^0|^sG_{n,x}|D^0|^t\right\|_{\gS_2}\le\int_\R \left\||D^0|^sA_1(\eta)X\right\|_{\gS_2}\prod_{j=1}^{n-1}\left\|A_{x_j}(\eta)X\right\|_{L^2\to L^2}\|A_{x_n}(\eta)|D^0|^t\|_{L^2\to L^2}\d{\eta}.$$  
Writing
\begin{multline}\label{eq:decomp-D0sA1X}
 |D^0|^sA_1(\eta)X=\frac{|D^0|^s}{D^0_\Lambda-\mu+i\eta}Y(\omega,\alpha)|D^0|^{-\frac{1}{2}}\frac{|D^0|}{D^0_\Lambda-\mu+i\eta}|D^0|^{-\frac{1}{2}}Y(\omega,\alpha)\times\\
\times\left(\frac{1}{D^0_\Lambda-\mu+i\eta}Y(\omega,\alpha)\right)^3\frac{1}{\Dref(\omega)-\mu+i\eta}|D^0||D^0|^{-1}X,
\end{multline}
 We obtain from the estimates above that
 $$\left\||D^0|^sA_1(\eta)X\right\|_{\gS_2}\le\frac{C(\mu)(1+\alpha\sqrt{\log\Lambda})^5\|\omega\|_{L^2\cap\cC}^5\theta(\mu,\omega)}{E(\eta)^{3/2}}\|(Q,\rho')\|_\cX,$$
 and one concludes that
 $$\left\||D^0|^sG_{n,x}|D^0|^t\right\|_{\gS_2}\le\left(C(\mu)(1+\alpha\sqrt{\log\Lambda})^5(1+\|\omega\|_{L^2\cap\cC}^5)\theta(\mu,\omega)\right)^{n+1}\|(Q,\rho')\|_\cX^n.$$
Then, if $j_0=n$, one has 
 $$\left\||D^0|^sG_{n,x}|D^0|^t\right\|_{\gS_2}\le\int_\R\||D^0|^sA_0(\eta)\|_{L^2\to L^2}\|XA_0(\eta)\|_{L^2\to L^2}^{n-1}\|XA_1(\eta)|D^0|^t\|_{\gS_2}\d{\eta}.$$
As before we have
$$\|XA_1(\eta)|D^0|^t\|_{\gS_2}\le\frac{C(\mu)(1+\alpha\sqrt{\log\Lambda})^5\|\omega\|_{L^2\cap\cC}^5\theta(\mu,\omega)}{E(\eta)^{3/2}}\|(Q,\rho')\|_\cX,$$
so that 
$$\left\||D^0|^sG_{n,x}|D^0|^t\right\|_{\gS_2}\le\left(C(\mu)(1+\alpha\sqrt{\log\Lambda})^5(1+\|\omega\|_{L^2\cap\cC}^5)\theta(\mu,\omega)\right)^{n+1}\|(Q,\rho')\|_\cX^n.$$
Finally, if $1\le j_0\le n-1$, then we write
\begin{multline*}
\left\||D^0|^sG_{n,x}|D^0|^t\right\|_{\gS_2}\le\int_\R\||D^0|^sA_0(\eta)\|_{L^2\to L^2}\|XA_0(\eta)\|_{L^2\to L^2}^{j_0-1}\|XA_1(\eta)X\|_{\gS_2}\times\\
\times\prod_{j=j_0+1}^{n-1}\|A_{x_j}(\eta)X\|_{L^2\to L^2}\|A_{x_n}(\eta)|D^0|^t\|_{L^2\to L^2}\d{\eta}, 
\end{multline*}
and using that
$$\|XA_1(\eta)X\|_{\gS_2}\le\frac{C(\mu)(1+\alpha\sqrt{\log\Lambda})^5\|\omega\|_{L^2\cap\cC}^5\theta(\mu,\omega)}{E(\eta)^{3/2}}\|(Q,\rho')\|_\cX^2,$$
obtained by the same decomposition as in \eqref{eq:decomp-D0sA1X}, we find
$$\left\||D^0|^sG_{n,x}|D^0|^t\right\|_{\gS_2}\le\left(C(\mu)(1+\alpha\sqrt{\log\Lambda})^5(1+\|\omega\|_{L^2\cap\cC}^5)\theta(\mu,\omega)\right)^{n+1}\|(Q,\rho')\|_\cX^n.$$
Hence, for all $s,t\in[0,1]$, for all $x\neq(0,\ldots,0)$ and all $n\ge1$ we have
$$\left\||D^0|^sG_{n,x}|D^0|^t\right\|_{\gS_2}\le\left(C(\mu)(1+\alpha\sqrt{\log\Lambda})^5(1+\|\omega\|_{L^2\cap\cC}^5)\theta(\mu,\omega)\right)^{n+1}\|(Q,\rho')\|_\cX^n.$$
It remains to estimate $\left\|G_{n,x}|D^0|^\frac{3}{2}\right\|_{\gS_2}$. The strategy is the same as the one for the estimation of $\left\||D^0|^\frac{3}{2}G_{n,x}\right\|_{\gS_2}$, except that we first consider the case $x_n=1$. For any $x\in\{0,1\}^{n+1}$ such that $x_n=1$, we have
\begin{equation}\label{eq:GnxD032}
 \left\|G_{n,x}|D^0|^\frac{3}{2}\right\|_{\gS_2}\le\int_\R\left(\prod_{j=0}^{n-1}\|A_{x_j}(\eta)X\|_{L^2\to L^2}\right)\left\|A_1(\eta)|D^0|^\frac{3}{2}\right\|_{\gS_2}\d{\eta}.
\end{equation}
In a similar fashion as the inequality \eqref{eq:prodXAj}, we have
$$\prod_{j=0}^{n-1}\|A_{x_j}(\eta)X\|_{L^2\to L^2}\le\left(C(\mu)(1+\alpha\sqrt{\log\Lambda})^5(1+\|\omega\|_{L^2\cap\cC}^5)\theta(\mu,\omega)^2\right)^n\|(Q,\rho')\|_\cX^n,$$
with now the factor $\theta(\mu,\omega)^2$ instead of $\theta(\mu,\omega)$ as in \eqref{eq:prodXAj} because of the term $\|(\Dref(\omega)-\mu+i\eta)^{-1}|D^0|\|_{L^2\to L^2}$. Using additionally \eqref{eq:estA1D032}, we get
$$\left\|G_{n,x}|D^0|^\frac{3}{2}\right\|_{\gS_2}\le\left(C(\mu)(1+\alpha\sqrt{\log\Lambda})^7\left(1+\|\omega\|_{L^2\cap\cC}^7\right)\theta(\mu,\omega)^3\right)^{n+1}\|(Q,\rho')\|_\cX^n,$$
for all $x\neq(0,\ldots,0)$ such that $x_n=1$. We treat the general case $x\neq(0,\ldots,0)$ in the same fashion as the estimate of $\||D^0|^{\frac{3}{2}}G_{n,x}\|_{\gS_2}$, using \eqref{eq:BD032}, and we get the same estimate as the $x_n=1$ case. Finally, by \eqref{eq:ineqnormQ} we conclude that for any $n\ge1$ and any $x\neq(0,\ldots,0)$,
\begin{equation}\label{eq:estGnx}
 \left\|G_{n,x}\right\|_\cQ\le\left(C(\mu)(1+\alpha\sqrt{\log\Lambda})^7\left(1+\|\omega\|_{L^2\cap\cC}^7\right)\theta(\mu,\omega)^3\right)^{n+1}\|(Q,\rho')\|_\cX^n.
\end{equation}
Combining this with \eqref{eq:estCMP}, we can write 
\begin{equation}\label{eq:est-Gn}
 \forall n\ge2,\quad \left\|G_n^\omega(Q,\rho')\right\|_\cQ\le\left(C(\mu)(1+\alpha\sqrt{\log\Lambda})^7\left(1+\|\omega\|_{L^2\cap\cC}^7\right)\theta(\mu,\omega)^3\right)^{n+1}\|(Q,\rho')\|_\cX^n,
\end{equation}
while 
\begin{equation}\label{eq:est-G1}
 \left\|G_1^\omega(Q,\rho')\right\|_\cQ\le\left(C(\mu)(1+\alpha\sqrt{\log\Lambda})^7\left(1+\|\omega\|_{L^2\cap\cC}^7\right)\theta(\mu,\omega)^3\right)^2\sqrt{\log\Lambda}\|(Q,\rho')\|_\cX.
\end{equation}
 We sum it up as
 \begin{equation}
  \|\Phi_1^\omega(Q,\rho')\|\le\alpha\sqrt{\log\Lambda}\Xi\|(Q,\rho')\|_\cX+\sum_{n\ge 2}\Big(\alpha\Xi\|(Q,\rho')\|_\cX\Big)^n,
 \end{equation}
 with 
 $$\Xi:=C(\mu)\left((1+\alpha\sqrt{\log\Lambda})^7\left(1+\|\omega\|_{L^2\cap\cC}^7\right)\theta(\mu,\omega)^3\right)^2.$$
 This is of course a rough estimate, but it has a simple form. We now turn to the estimates for the density terms, that is $\Phi_2^\omega$. We start by estimating $J_n^\omega=\rho_{G_n^\omega}$ for $n\ge 2$. As for $G_n^\omega$, we write
 $$J_n^\omega=\frac{(-1)^{n+1}}{2\pi}\sum_{x\in\{0,1\}^{n+1}}J_{n,x},$$
 with $J_{n,x}:=\rho_{G_{n,x}}$. It was proved in \cite[Proposition 10]{HaiLewSer-05a} that 
 \begin{equation}\label{eq:estrhoCMP}
  \forall n\ge2,\quad \|J_{n,(0,\ldots,0)}\|_{L^2\cap\cC}\le \left(C(\mu)(1+\alpha\sqrt{\log\Lambda})^4(1+\|\omega\|_{L^2\cap\cC}^4)\right)^{n+1}\|(Q,\rho')\|_\cX^n,
 \end{equation}
  To estimate $\|J_{n,x}\|_{L^2\cap\cC}$ for $x\neq(0,\ldots,0)$, we use the following lemma.
  \begin{lemma}
  For any operator $Q$ on $\gH_\Lambda$ such that $|D^0|^{\frac{3}{2}+\sigma}Q\in\gS_2$ for some $\sigma>0$ and such that $|D^0|Q\in\gS_{6/5}$, then $\rho_Q\in L^2\cap\cC$ and there exists a universal constant $C$ such that
   \begin{equation}\label{eq:ineqnormrho}
    \|\rho_Q\|_{L^2\cap\cC}\le C\left(\||D^0|^{\frac{3}{2}+\sigma}Q\|_{\gS_2}+\||D^0|Q\|_{\gS_{6/5}}\right).
   \end{equation}
 \end{lemma}
 \begin{proof}
  We prove it by duality. Let $V\in\gH_\Lambda$ be a smooth function. Then, for any $Q$ as in the statement of the lemma, we have $QV\in\gS_1$ with two different estimates on $\|QV\|_{\gS_1}$. First, we have
  $$\|QV\|_{\gS_1}\le\||D^0|^{\frac{3}{2}+\sigma}Q\|_{\gS_2}\||D^0|^{-\frac{3}{2}-\sigma}V\|_{\gS_2}\le C\||D^0|^{\frac{3}{2}+\sigma}Q\|_{\gS_2}\|V\|_{L^2},$$
  by the Kato-Seiler-Simon inequality. This proves that $\rho_Q\in L^2$ and that 
  $$\|\rho_Q\|_{L^2}\le C\||D^0|^{\frac{3}{2}+\sigma}Q\|_{\gS_2}.$$
  Secondly, we have 
  \begin{multline*}
\|QV\|_{\gS_1}\le\||D^0|Q\|_{\gS_{6/5}}\||D^0|^{-1}V\|_{\gS_6}\\
\le C\||D^0|Q\|_{\gS_{6/5}}\|V\|_{L^6}\le C\||D^0|Q\|_{\gS_{6/5}}\|\nabla V\|_{L^2},   
  \end{multline*}
  again by the Kato-Seiler-Simon inequality and the Sobolev inequality. This proves that $\rho_Q\in\cC$ and that
  $$\|\rho_Q\|_\cC\le C\||D^0|Q\|_{\gS_{6/5}}.$$
 \end{proof}

 Let us first treat the case $x_0=1$. Let $n\ge1$ and $x\neq(0,\ldots,0)$ with $x_0=1$. From inequalities \eqref{eq:D032Gnx}, \eqref{eq:estD032+sigmaA1} with any $0<\sigma<\frac{1}{2}$, and $\eqref{eq:estD0A165}$, we find
 $$\left\||D^0|^{\frac{3}{2}+\sigma}G_{n,x}\right\|_{\gS_2}\le\left(C(\mu)(1+\alpha\sqrt{\log\Lambda})^5(1+\|\omega\|_{L^2\cap\cC}^5)\theta(\mu,\omega)\right)^{n+1}\|(Q,\rho')\|_\cX^n,$$
 $$\left\||D^0|G_{n,x}\right\|_{\gS_{6/5}}\le\left(C(\mu)(1+\alpha\sqrt{\log\Lambda})^5(1+\|\omega\|_{L^2\cap\cC}^5)\theta(\mu,\omega)\right)^{n+1}\|(Q,\rho')\|_\cX^n.$$
 To deduce the general $x$ case from the $x_0=1$ case we use as usual \eqref{eq:D032B} and \eqref{eq:estD0B}. We thus have for all $n\ge2$,
 $$\|J_n^\omega(Q,\rho')\|_{L^2\cap\cC}\le \left(C(\mu)(1+\alpha\sqrt{\log\Lambda})^5(1+\|\omega\|_{L^2\cap\cC}^5)\theta(\mu,\omega)\right)^{n+1}\|(Q,\rho')\|_\cX^n.$$
 It remains to estimate $\|J_1^\omega(Q,\rho')\|_{L^2\cap\cC}$. We have proved that for $x\neq(0,0)$,
 $$\|J_{1,x}(Q,\rho')\|_{L^2\cap\cC}\le C(\mu)\left((1+\alpha\sqrt{\log\Lambda})^5(1+\|\omega\|_{L^2\cap\cC}^5)\theta(\mu,\omega)\right)^2\|(Q,\rho')\|_\cX.$$
 Let us write 
 $$A_0(\eta)=\frac{1}{D^0_\Lambda-\mu+i\eta}+A_{01}(\eta),$$
 so that 
 \begin{multline*}
 J^\omega_1(Q,\rho')=-\rho\left[\frac{1}{2\pi}\int_\R\frac{1}{D^0_\Lambda-\mu+i\eta}R_{Q}\frac{1}{D^0_\Lambda-\mu+i\eta}\d{\eta}\right]\\
 +\rho\left[\frac{1}{2\pi}\int_\R\frac{1}{D^0_\Lambda-\mu+i\eta}XA_{01}(\eta)\d{\eta}\right]+\rho\left[\frac{1}{2\pi}\int_\R A_{01}(\eta)X\frac{1}{D^0_\Lambda-\mu+i\eta}\d{\eta}\right]\\
 +\rho\left[\frac{1}{2\pi}\int_\R A_{01}(\eta)XA_{01}(\eta)\d{\eta}\right]+\frac{1}{2\pi}\sum_{\substack{x\in\{0,1\}^2 \\ x\neq(0,0)}}J_{1,x}.
 \end{multline*}
 By \cite[Proposition 10]{HaiLewSer-05a}, we deduce that 
 $$\|J_1^\omega(Q,\rho')\|_{L^2\cap\cC}\le\left(C(\mu)(1+\alpha\sqrt{\log\Lambda})^5(1+\|\omega\|_{L^2\cap\cC}^5)\theta(\mu,\omega)\right)^2\sqrt{\log\Lambda}\|(Q,\rho')\|_\cX.$$
 This ends the proof of Proposition \ref{prop:main}.
 \end{proof}

\begin{proposition}\label{prop:est-Qref-rhoref}
 Let $\nu\in L^2(\R^3)\cap\cC$ and $\mu\in(-1,1)$. Let $\epsilon\ge0$ such that for all $s\in[1-\epsilon,1+\epsilon]$, $\mu\notin\sigma_\gH(D^{s\nu})$. Assume furthermore that 
 \begin{equation}
  \int_{\R^3}\log(2+|k|)^2\left|\hat{\nu}(k)\right|^2\d{k}<+\ii.
 \end{equation}
There exists $\alpha_1=\alpha_1(\mu,\nu,\epsilon)>0$, $\Lambda_1=\Lambda_1(\mu,\nu,\epsilon)\ge1$, and $L_1=L_1(\mu,\nu,\epsilon)>0$ such that for any $0\le\alpha\le\alpha_1$, for any $\Lambda\ge\Lambda_1$ satisfying $\alpha\log\Lambda\le L_1$, and for any $\kappa\in[1-\epsilon,1+\epsilon]$ we have
 \begin{equation}\label{eq:est-Qtilde-rhotilde}
\|(\Qref,\rhoref)\|_\cX\le C(\mu,\nu,\epsilon).
 \end{equation}
\end{proposition}

\begin{proof}
 We split $\Qref$ into two terms, $\Qref=I+II$, where
 \begin{eqnarray*}
  I & = & \projneg{D^{\kappa\nur}_\Lambda}-P^0_{\Lambda,-}+G_{1,0}(\kappa\nur),\\
 II & = & \projneg{\Dref}-\projneg{D^{\kappa\nur}_\Lambda}.
\end{eqnarray*}
 Since $\nu\in L^2(\R^3)\cap\cC$, we have $V_\nu\in L^\ii$ and the eigenvalues of $D^{s\nu}$ are continuous with respect to $s$. Thus, there exists $\epsilon'=\epsilon'(\mu,\nu,\epsilon)>\epsilon$ such that $\mu\notin\sigma_\gH(D^{s\nu})$ for all $s\in[1-\epsilon',1+\epsilon']$. Hence, there exists $L_0=L_0(\mu,\nu,\epsilon)>0$ such that for all $\kappa\in[1-\epsilon,1+\epsilon]$, for all $\alpha\ge0$ and for all $\Lambda\ge1$ satisfying $\alpha\log\Lambda\le L_0$ we have
$$\frac{\kappa}{1+\alpha B_\Lambda}\in[1-\epsilon',1+\epsilon'].$$
Using Lemma \ref{lemma:gap-Dref}, this implies $d(\mu,\sigma_{\gH_\Lambda}(D^{\kappa\nur}))\ge\xi_0(\mu,\nu,\epsilon,L_0(\mu,\nu,\epsilon))>0$, for all $\kappa\in[1-\epsilon,1+\epsilon]$, for all $\Lambda\ge\Lambda_0(\mu,\nu,\epsilon,L_0(\mu,\nu,\epsilon))$, and for all $\alpha\ge0$ such that $\alpha\log\Lambda\le L_0$. In particular, $\mu\notin\sigma_{\gH_\Lambda}(D^{\kappa\nur})$ and we can write
$$I=\frac{1}{2\pi}\int_\R\frac{1}{D^0_\Lambda-\mu+i\eta}\Pi_\Lambda V_{\kappa\nur}\Pi_\Lambda\frac{1}{D^0_\Lambda-\mu+i\eta}\Pi_\Lambda V_{\kappa\nur}\Pi_\Lambda\frac{1}{D^{\kappa\nur}_\Lambda-\mu+i\eta}.$$
As in the proof of Proposition \ref{prop:main}, we estimate this term by decomposing it into $I=I_1+I_2$, with
\begin{multline*}
I_1=\frac{1}{2\pi}\int_\R\left(\frac{1}{D^0_\Lambda-\mu+i\eta}\Pi_\Lambda V_{\kappa\nur}\Pi_\Lambda\right)^2\times\\
\times\sum_{j=0}^2\frac{1}{D^0_\Lambda-\mu+i\eta}\left(\Pi_\Lambda V_{\kappa\nur}\Pi_\Lambda\frac{1}{D^0_\Lambda-\mu+i\eta}\right)^j\d{\eta}, 
\end{multline*}
$$I_2=\frac{1}{2\pi}\int_\R\left(\frac{1}{D^0_\Lambda-\mu+i\eta}\Pi_\Lambda V_{\kappa\nur}\Pi_\Lambda\right)^5\frac{1}{D^{\kappa\nur}_\Lambda-\mu+i\eta}\d{\eta}.$$
By \cite[Lemma 13, Lemma 15]{HaiLewSer-05a}, we have
$$\|I_1\|_\cQ\le C(1+\|\kappa\nur\|_{L^2\cap\cC}^4).$$
Notice that
$$I_2=\frac{1}{2\pi}\int_\R A_1(\eta)_{|\alpha=0, \omega=\kappa\nur}\d{\eta},$$
where $A_1(\eta)$ was defined in \eqref{eq:def-A1eta}. By \eqref{eq:ineqnormQ}, \eqref{eq:estD032+sigmaA1}, \eqref{eq:estA1D032}, and the trivial bound
\begin{multline*}
 2\||D^0|A_1(\eta)_{|\alpha=0, \omega=\kappa\nur}|D^0|^{1/2}\|_{\gS_2}+2\||D^0|^{1/2}A_1(\eta)_{|\alpha=0, \omega=\kappa\nur}|D^0|\|_{\gS_2} \\
 \le \frac{C(\mu)\|\kappa\nur\|_{L^2\cap\cC}^5\theta(\mu,\kappa\nur)_{|\alpha=0}}{E(\eta)^{3/2}},
\end{multline*}
we infer that
$$\|I_2\|_\cQ\le C(\mu)\|\kappa\nur\|_{L^2\cap\cC}^7\theta(\mu,\nu)_{|\alpha=0}^3.$$
Hence, for any $\kappa\in[1-\epsilon,1+\epsilon]$,
$$\|I\|_\cQ\le C(\mu)(1+(1+\epsilon)^7\|\nu\|_{L^2\cap\cC}^7)\theta(\mu,\kappa\nur)_{|\alpha=0}^3.$$
Since $d(\mu,\sigma_{\gH_\Lambda}(D^{\kappa\nur}))\ge\xi_0$, we have $\theta(\mu,\kappa\nur)_{|\alpha=0}\le C(\mu,\nu,\epsilon)$, uniformly in $\alpha\ge0$ (the dependence on $\alpha$ only being in $\nur$, which is not obvious from the notation), $\Lambda\ge\Lambda_0$ satisfying $\alpha\log\Lambda\le L_0$. This leads to
$$\|I\|_\cQ\le C(\mu,\nu,\epsilon).$$
By \eqref{eq:ineq-G10Q}, this implies in particular that for $\Lambda\ge3$ so that $\log\Lambda\ge1$, 
\begin{equation}\label{eq:est-projlibre}
 \|\projneg{D^{\kappa\nur}_\Lambda}-P^0_{\Lambda,-}\|_\cQ\le C(\mu,\nu,\epsilon)\sqrt{\log\Lambda}.
\end{equation}
Using again Lemma \ref{lemma:gap-Dref}, we know that for all $0\le\alpha\le\alpha_0$ and for all $\Lambda\ge\Lambda_0$ satisfying $\alpha\log\Lambda\le L_0$, we have $d(\mu,\sigma_{\gH_\Lambda}(\Dref))\ge\xi_0$ for all $\kappa\in[1-\epsilon,1+\epsilon]$. Hence, we can expand $II$ by Kato's formula, as for $I$. To estimate this expansion by the same method as in the proof of Proposition \ref{prop:main}, we need the additional condition
$$8\alpha\|\kappa\nur\|_\cC<1-|\mu|,$$
which ensures the finitess of $\theta(\mu,\kappa\nur)$, as seen by Lemma \ref{lemma:D0Dref-1L2L2}. This condition is satisfied for all $0\le\alpha\le\alpha_1$, for all $\Lambda\ge\Lambda_0$ satisfying $\alpha\log\Lambda\le L_0$, and for all $\kappa\in[1-\epsilon,1+\epsilon]$, for a certain $0\le\alpha_1=\alpha_1(\mu,\nu,\epsilon)\le\alpha_0(\mu,\nu,\epsilon)$. Using the same method used to estimate $\|I\|_\cQ$ and the proof of Proposition \ref{prop:main}, we find
\begin{equation}\label{eq:est-Qref-Qlin}
 \|II\|_\cQ\le C(\mu)\alpha\sqrt{\log\Lambda}(1+\alpha\sqrt{\log\Lambda})^7(1+(1+\epsilon)^7\|\nu\|_{L^2\cap\cC}^7)\theta(\mu,\kappa\nur)^3.
\end{equation}
We still have $\theta(\mu,\kappa\nur)\le C(\mu,\nu,\epsilon)$ for all $\kappa\in[1-\epsilon,1+\epsilon]$, and we also have $\alpha\sqrt{\log\Lambda}\le\alpha\log\Lambda\le L_0$ if $\Lambda\ge3$. Hence,
$$\|II\|_\cQ\le C(\mu,\nu,\epsilon)$$
for all $0\le\alpha\le\alpha_1$, for all $\Lambda\ge\max(3,\Lambda_0)$ with $\alpha\log\Lambda\le L_0$ and for all $\kappa\in[1-\epsilon,1+\epsilon]$. The estimate on $\rhoref$ is more subtle. We decompose it as
 $$\rhoref=-\cU_\Lambda\kappa\nur+\rho\left[I\right]+\rho\left[II\right].$$
  The operator $\cU_\Lambda$ is the multiplication operator in Fourier space by the function $U_\Lambda(|k|):=B_\Lambda-B_\Lambda(k)$. It has been defined and studied in \cite{GraLewSer-11}. In fact, combining \cite[Equation (3.3)]{GraLewSer-11} and \cite[Equation (A.3)]{GraLewSer-11}, one gets the estimate 
 \begin{equation}\label{eq:est-ULambda}
\forall\Lambda\ge1,\,\forall r\ge0,\qquad 0\le U_\Lambda(r)\le\kappa_1\left(1+\frac{1}{3\pi}\log(1+r^2)\right),  
 \end{equation}
with $\kappa_1=258/\pi$. We deduce
$$\left\|\cU_\Lambda \kappa\nur\right\|_\cC\le C(1+\epsilon)\|\nu\|_{L^2\cap\cC},$$
$$\left\|\cU_\Lambda \kappa\nur\right\|_{L^2}^2\le C(1+\epsilon)^2\int_{\R^3}\log(2+|k|)^2|\hat{\nu}(k)|^2\d{k}.$$
Next, using the same ideas as before, one gets 
$$\|\rho[I]\|_{L^2\cap\cC}\le C(\mu)(1+(1+\epsilon)^7\|\nu\|_{L^2\cap\cC}^7)\theta(\mu,\kappa\nur)_{|\alpha=0}^3\le C(\mu,\nu,\epsilon),$$
$$\|\rho[II]\|_{L^2\cap\cC}\le  C(\mu)\alpha\log\Lambda(1+\alpha\sqrt{\log\Lambda})^7(1+(1+\epsilon)^7\|\nu\|_{L^2\cap\cC}^7)\theta(\mu,\kappa\nur)^3\le C(\mu,\nu,\epsilon),$$
which ends the proof.
\end{proof}

\subsubsection{End of the proof of Theorem 1}

We apply the fixed point theorem of Banach-Picard to the function $\Phi$ defined earlier by 
$$\Phi(Q,\rho')=(\Phi_1^{\kappa\nur}(Q+\Qref,\rho'),\Phi_2^{\kappa\nur}(Q+\Qref,\rho')),$$
for any $(Q,\rho')\in B_\cX(R)\subset\cX$, where
$$B_\cX(R)=\{(Q,\rho')\in\cX,\quad\|(Q+\Qref,\rho')\|_\cX\le R\},$$
for a $R>0$ to be determined later on. For $\Phi(Q,\rho')$ to be correctly defined, we have to check the assumptions of Proposition \ref{prop:main}. First, as in the proof of Proposition \ref{prop:est-Qref-rhoref}, using Lemma \ref{lemma:gap-Dref}, there exists $L_0=L_0(\mu,\nu,\epsilon)>0$ such that for all $0\le\alpha\le\alpha_0(\mu,\nu,\epsilon,L_0)$, for all $\Lambda\ge\Lambda_0(\mu,\nu,\epsilon,L_0)$ with $\alpha\log\Lambda\le L_0$, we have $\mu\notin\sigma_{\gH_\Lambda}(\Dref)$ for all $\kappa\in[1-\epsilon,1+\epsilon]$. Next, we assume additionally that
$$8\alpha\|\kappa\nur\|_\cC<1-|\mu|,\qquad \alpha b \|(Q+\Qref,\rho')\|_\cX<1.$$
This condition is verified if we assume furthermore that $0\le\alpha\le\alpha_2$ for a certain $\alpha_2=\alpha_2(\mu,\nu,\epsilon,R)>0$. Using Lemma \ref{lemma:gap-DQ}, it implies that $\mu\notin\sigma_{\gH_\Lambda}(D_Q)$, where we recall that 
$$D_Q=\projneg{\Dref+\alpha\Pi_\Lambda(V_{\rho'}-R_{Q+\Qref})\Pi_\Lambda}.$$
We can now use Proposition \ref{prop:main} to infer that for all $(Q,\rho')\in B_\cX(R)$,
\begin{equation}\label{eq:est-phi'}
\|\Phi'(Q,\rho')\|_{\cX\to\cX}\le \alpha\sqrt{\log\Lambda}\Xi+\alpha \sum_{n\ge0}n(\alpha\Xi R)^{n-1}. 
\end{equation}
Assuming $\Lambda\ge3$, we have $\alpha\sqrt{\log\Lambda}\le\alpha\log\Lambda\le L_0$. Moreover, as in the proof of Proposition \ref{prop:est-Qref-rhoref}, we can show, using Lemma \ref{lemma:gap-Dref}, that there exists $C=C(\mu,\nu,\epsilon)>0$ such that for all $0\le\alpha\le\alpha_0$, for all $\Lambda\ge\Lambda_0$ with $\alpha\log\Lambda\le L_0$, and for all $\kappa\in[1-\epsilon,1+\epsilon]$ we have $\Xi\le C$. As a consequence, the right side of \eqref{eq:est-phi'} is finite as soon as $\alpha\le(CR)^{-1}/2=\alpha_3=\alpha_3(\mu,\nu,\epsilon,R)$. It thus satisfies
$$\forall (Q,\rho')\in B_\cX(R),\qquad \|\Phi'(Q,\rho')\|_\cX\le C_1\sqrt{\alpha},$$
for some $C_1=C_1(\mu,\nu,\epsilon,R)>0$. Hence, $\Phi$ is a contraction on $B_\cX(R)$ if $\alpha\le(4C_1)^{-2}=\alpha_4=\alpha_4(\mu,\nu,\epsilon,R)>0$. It remains to prove that $B_\cX(R)$ is invariant under $\Phi$. For any $(Q,\rho')\in B_\cX(R)$, we have
\begin{eqnarray*}
 \|\Phi(Q,\rho')+(\Qref,0)\|_\cX & \le & \|\Phi(Q,\rho')-\Phi(-\Qref,0)\|_\cX+\|\Phi(-\Qref,0)+(\Qref,0)\|_\cX \\
  & \le & C_1\sqrt{\alpha}R+\|(\Qref,\rhoref)\|_\cX.
\end{eqnarray*}
By Proposition \ref{prop:est-Qref-rhoref}, if $0\le\alpha\le\alpha_1$, $\Lambda\ge\Lambda_1$, and $\alpha\log\Lambda\le L_1$, we have
$$\|(\Qref,\rhoref)\|_\cX\le C=C(\mu,\nu,\epsilon).$$
Since we have chosen $\alpha$ such that $C_1\sqrt{\alpha}\le 1/2$, $\Phi$ stabilizes $B_\cX(R)$ if
$$R=2C.$$
To recap, if $0\le\alpha\le\min(\alpha_0,\alpha_1,\alpha_2,\alpha_3,\alpha_4)$, $\Lambda\ge\max(3,\Lambda_0,\Lambda_1)$ with $\alpha\log\Lambda\le\min(L_0,L_1)$, then for any $\kappa\in[1-\epsilon,1+\epsilon]$, $\Phi$ is a contraction on $B_\cX(R)$, with $R=R(\mu,\nu,\epsilon)>0$. Hence, there exists a unique $Q\in\cX$ satisfying
$$Q=\chi_{(-\ii,\mu]}\left(\Dref+\alpha\Pi_\Lambda\left[V_{\rho_Q+\rhoref}-R_{Q+\Qref}\right]\Pi_\Lambda\right)-\Pref,$$
$$\|(Q+\Qref,\rho_Q+\rhoref)\|_\cX\le R.$$
Defining $\tilde{Q}:=Q+\Qref-G_{1,0}(\kappa\nur)$, we have $\rho_Q+\rhoref=\rho_{\tilde{Q}}-B_\Lambda\kappa\nur$ and we have proved that there is a unique solution to
$$\tilde{Q}=\chi_{(-\ii,\mu]}\left(D^{\kappa\nu}_\Lambda+\alpha\Pi_\Lambda\left[V_{\rho_{\tilde{Q}}}-R_{\tilde{Q}}\right]\Pi_\Lambda\right)-P^0_{\Lambda,-},$$
$$\|(\tilde{Q}+G_{1,0}(\kappa\nur),\rho_{\tilde{Q}}-B_\Lambda\kappa\nur)\|_\cX\le R,$$
which ends the proof of Theorem \ref{thm:thm1}.

\subsection{Proof of Theorem \ref{thm:thm2}}

Let $\nu\in L^2(\R^3)\cap\cC$, $\mu_\pm\in(-1,1)$, and $\epsilon>0$ satisfying the assumptions given in Section \ref{sec:setting}. Since the eigenvalues of $D^{s\nu}$ are continuous with respect to $s$, there exists $\epsilon'>\epsilon$ such that $(\mu_\pm,\nu,\epsilon')$ also satisfy the assumptions given in Section \ref{sec:setting}. By the same method as in the proofs of Lemma \ref{lemma:gap-Dref} and Proposition \ref{prop:est-Qref-rhoref}, there exists $\alpha'=\alpha'(\mu,\nu,\epsilon)>0$, $\Lambda'=\Lambda'(\mu,\nu,\epsilon)\ge1$, $L'=L'(\mu,\nu,\epsilon)>0$ and $\xi'=\xi'(\mu,\nu,\epsilon)>0$ such that for all $0\le\alpha\le\alpha'$, for all $\Lambda\ge\Lambda'$ satisfying $\alpha\log\Lambda\le L'$, and for all $\kappa\in[1-\epsilon,1+\epsilon]$, we have $\kappa/(1+\alpha B_\Lambda)\in[1-\epsilon',1+\epsilon']$ and $d(\mu_\pm,\sigma_{\gH}(D^{\kappa\nur}_\Lambda))\ge\xi'$. The following result ensures that the operator $D^{\kappa\nur}_\Lambda$ also has an eigenvalue in $(\mu_-,\mu_+)$, for $\Lambda$ large enough. 

\begin{lemma}
 Let $\nu\in L^2\cap\cC$, $\mu_\pm\in(-1,1)$, and $\epsilon>0$ satisfying the assumptions given in Section \ref{sec:setting}. Then, there exists $\Lambda''=\Lambda''(\mu_\pm,\nu,\epsilon)\ge\Lambda'$ and $C=C(\mu_\pm,\nu,\epsilon)>0$ such that for any $\Lambda\ge\Lambda''$, for any $0\le\alpha\le\alpha'$ with $\alpha\log\Lambda\le L'$, and for any $\kappa\in[1-\epsilon,1+\epsilon]$, the operator $D^{\kappa\nur}_\Lambda$ on $\gH_\Lambda$ possesses a simple eigenvalue $\lambda(\kappa\nur,\Lambda)$ in $(\mu_-,\mu_+)$, which verifies the estimate
 \begin{equation}
  |\lambda(\kappa\nur)-\lambda(\kappa\nur,\Lambda)|\le C E(\Lambda)^{-1}.
 \end{equation}
 Furthermore, if we denote by $\phi(\kappa\nur,\Lambda)$ the unique normalized eigenvector of $D^{\kappa\nur}_\Lambda$ for the eigenvalue $\lambda(\kappa\nur,\Lambda)$ such that $\langle\phi(\kappa\nur),\phi(\kappa\nur,\Lambda)\rangle_{L^2}>0$, then we have
 \begin{equation}
  \|\phi(\kappa\nur)-\phi(\kappa\nur,\Lambda)\|_{H^1} \le C E(\Lambda)^{-1},
 \end{equation}
\end{lemma}

\begin{remark}
 Since $\kappa/(1+\alpha B_\Lambda)\in[1-\epsilon',1+\epsilon']$, and $(\mu_\pm,\nu,\epsilon')$ satisfy the assumptions of Section \ref{sec:setting}, $\lambda(\kappa\nur)=\lambda(\kappa/(1+\alpha B_\Lambda)\nu)$ is well-defined. 
\end{remark}

 \begin{proof}
  In order to compare $D^{\kappa\nur}$ with $D^{\kappa\nur}_\Lambda$, we identify $D^{\kappa\nur}_\Lambda$ with the operator on $\gH$ which is $D^{\kappa\nur}_\Lambda$ on $\gH_\Lambda$ and $0$ on $(1-\Pi_\Lambda)\gH_\Lambda$. Let $\CJ$ be a complex, compact contour around $(\mu_-,\mu_+)$, of index 1, crossing $\R$ only through $\mu_-$ and $\mu_+$. This way $\CJ$ remains uniformly far from $\sigma_\gH(D^{\kappa\nur})$ and $\sigma_\gH(D^{\kappa\nur}_\Lambda)$, for $\Lambda\ge\Lambda'$. We then have
  \begin{multline*}
\chi_{(\mu_-,\mu_+)}(D^{\kappa\nur})-\Pi_\Lambda\chi_{(\mu_-,\mu_+)}(D^{\kappa\nur}_\Lambda)\Pi_\Lambda=\\
-\frac{1}{2i\pi}\int_{\CJ}\Pi_\Lambda\left(\frac{1}{D^{\kappa\nur}-z}-\frac{1}{D^{\kappa\nur}_\Lambda-z}\right)\Pi_\Lambda\d{z}\\
-\frac{1}{2i\pi}\int_{\CJ}\left((1-\Pi_\Lambda)\frac{1}{D^{\kappa\nur}-z}\Pi_\Lambda+\Pi_\Lambda\frac{1}{D^{\kappa\nur}-z}(1-\Pi_\Lambda)\right.\\
\left.+(1-\Pi_\Lambda)\frac{1}{D^{\kappa\nur}-z}(1-\Pi_\Lambda)\right)\d{z}=I+II,   
  \end{multline*}
where $I$ is only the term in the second line. For the second term, we write
\begin{eqnarray*}
 \|II\|_{L^2\to L^2} & \le & 3\|(1-\Pi_\Lambda)|D^0|^{-1}\|_{L^2\to L^2}\sup_{z\in\CJ}\left\|\frac{|D^0|}{D^{\kappa\nur}-z}\right\|_{L^2\to L^2}\\
  & \le & C(\mu_\pm,\nu,\epsilon)E(\Lambda)^{-1}.
\end{eqnarray*}
The first term is computed by the resolvent formula 
$$I=\frac{1}{2i\pi}\int_{\CJ}\Pi_\Lambda\frac{1}{D^{\kappa\nur}-z}(1-\Pi_\Lambda)V_{\kappa\nur}\Pi_\Lambda\frac{1}{D^{\kappa\nur}_\Lambda-z}\d{z},$$
so that we have the estimation
\begin{eqnarray*}
\|I\|_{L^2\to L^2} & \le & \sup_{z\in\CJ}\left\|\frac{|D^0|}{D^{\kappa\nur}-z}\right\|_{L^2\to L^2}\||D^0|^{-1}(1-\Pi_\Lambda)\|_{L^2\to L^2}\sup_{z\in\CJ}\left\|\frac{1}{D^{\kappa\nur}_\Lambda-z}\right\|_{\gH_\Lambda\to\gH_\Lambda}\\
 & \le & C(\mu_\pm,\nu,\epsilon) E(\Lambda)^{-1}.
\end{eqnarray*}
As a consequence, for $\Lambda$ large enough such that for instance
$$\|\chi_{(\mu_-,\mu_+)}(D^{\kappa\nur})-\Pi_\Lambda\chi_{(\mu_-,\mu_+)}(D^{\kappa\nur}_\Lambda)\Pi_\Lambda\|_{L^2\to L^2}<1,$$
the operator $D^{\kappa\nur}_\Lambda$ on $\gH_\Lambda$ has a simple eigenvalue $\lambda(\kappa\nur,\Lambda)\in(\mu_-,\mu_+)$. To estimate the difference between these eigenvalues, notice that we have proved that for any $z\in[\mu_\pm-\delta_0,\mu_\pm+\delta_0]$, for some $\delta_0=\delta_0(\mu_\pm,\nu,\epsilon)>0$ small enough
$$\left\|\frac{1}{D^{\kappa\nur}-z}-\Pi_\Lambda\frac{1}{D^{\kappa\nur}_\Lambda-z}\Pi_\Lambda\right\|_{L^2\to L^2}\le C(\mu_\pm,\nu,\epsilon) E(\Lambda)^{-1}.$$
It implies that, for $\Lambda$ large enough,
$$\left|\frac{1}{\lambda(\kappa\nur)-z}-\frac{1}{\lambda(\kappa\nur,\Lambda)-z}\right|\le C(\mu_\pm,\nu,\epsilon) E(\Lambda)^{-1}$$
and hence
$$|\lambda(\kappa\nur)-\lambda(\kappa\nur,\Lambda)|\le C(\mu_\pm,\nu,\epsilon) E(\Lambda)^{-1}.$$
Let us now consider the unique eigenvector $\phi(\kappa\nur,\Lambda)$ as in the statement of the lemma. It is well-defined for $\Lambda$ large enough since the operator 
$$A:=|\phi(\kappa\nur)\rangle\langle\phi(\kappa\nur)|-|\phi(\kappa\nur,\Lambda)\rangle\langle\phi(\kappa\nur,\Lambda)|.$$
verifies
\begin{multline*}
0\le\langle A\phi(\kappa\nur),\phi(\kappa\nur)\rangle_{L^2}=1-\langle\phi(\kappa\nur,\Lambda),\phi(\kappa\nur)\rangle_{L^2}^2\\
\le\|A\|_{L^2\to L^2}\le C(\mu_\pm,\nu,\epsilon) E(\Lambda)^{-1}<1.
\end{multline*}
Hence, $\langle\phi(\kappa\nur,\Lambda),\phi(\kappa\nur)\rangle_{L^2}\neq0$ and
\begin{multline*}
\|\phi(\kappa\nur)-\phi(\kappa\nur,\Lambda)\|_{L^2}^2=2(1-\langle\phi(\kappa\nur),\phi(\kappa\nur,\Lambda)\rangle)\\
\le2(1-\langle\phi(\kappa\nur),\phi(\kappa\nur,\Lambda)\rangle^2)\le C(\mu_\pm,\nu,\epsilon) E(\Lambda)^{-1}. 
\end{multline*}
To obtain a control in $H^1$, we use the eigenvalue equation, which implies
\begin{multline*}
 D^0(\phi(\kappa\nur)-\phi(\kappa\nur,\Lambda))=\lambda(\kappa\nur)(\phi(\kappa\nur)-\phi(\kappa\nur,\Lambda))\\
 +(\lambda(\kappa\nur)-\lambda(\kappa\nur,\Lambda))\phi(\kappa\nur,\Lambda)\\
+ V_{\kappa\nur}(\phi(\kappa\nur)-\phi(\kappa\nur,\Lambda))+(1-\Pi_\Lambda)V_{\kappa\nur}\phi(\kappa\nur,\Lambda).
\end{multline*}
Using the fact that $V_{\kappa\nur}\phi(\kappa\nur,\Lambda)$ is bounded in $H^1$ uniformly in $\Lambda$, we get the estimate in $H^1$. 
 \end{proof}

 
 We now begin the proof of Theorem \ref{thm:thm2}. Since we work with $\kappa$ fixed we will omit the dependence on $\kappa$ in our notations, meaning that we write $Q_\pm(\alpha)$ instead of $Q_\pm(\kappa,\alpha)$. We assume $0\le\alpha\le\min(\alpha_0,\alpha')$, $\Lambda\ge\max(\Lambda_0,\Lambda'')$ and $\alpha\log\Lambda\le\min(L_0,L')$, where $(\alpha_0,\Lambda_0,L_0)$ is given by Theorem \ref{thm:thm1} and $(\alpha',\Lambda'',L')$ were given just above. In the sequel, we will use notations such as $\cO_{\alpha\to0}(\alpha)$ (resp. $\cO_{\alpha\to0}^X(\alpha)$) to denote a function of $\alpha$ which is bounded in absolute value (resp. in norm ${\|\cdot\|_X}$ for some space $X$) by $C\alpha$, where $C$ only depends on $\mu_\pm,\nu,\epsilon$. In our regime of parameters, notice that we have for instance $\cO_{\Lambda\to\ii}(\sqrt{\log\Lambda})=\cO_{\alpha\to0}(\sqrt{\alpha})$. Likewise, in the following we will use the notation $C$ for a constant only depending on $\mu_\pm$, $\nu$ and $\epsilon$, whose value may change from line to line. 
 
 We use the a priori bound given by Theorem \ref{thm:thm1}
\begin{equation}\label{eq:apriori}
\|(Q_\pm(\alpha)+G_{1,0}(\kappa\nur),\rho_\pm(\alpha)-B_\Lambda\kappa\nur)\|_\cX=\cO_{\alpha\to0}(1),
\end{equation}
 where we used the shortcut notation $\rho_\pm(\alpha):=\rho_{Q_\pm(\alpha)}$. The self consistent equation for $Q_\pm(\alpha)$ is 
 \begin{eqnarray*}
  Q_\pm(\alpha) & = & \projnegpm{\Pi_\Lambda\left(D^{\kappa\nu}_\Lambda+\alpha\left(V_{\rho_\pm(\alpha)}-R_{Q_\pm(\alpha)}\right)\right)\Pi_\Lambda}-P^0_{\Lambda,-}\\
 & = & \projnegpm{\Dref+\alpha\Pi_\Lambda\left(V_{\rho_\pm(\alpha)-B_\Lambda\kappa\nur}-R_{Q_\pm(\alpha)+G_{1,0}(\kappa\nur)}\right)\Pi_\Lambda}-P^0_{\Lambda,-}\\
 & = & \sum_{n\ge1}\alpha^n G_{n,\pm}\left(Q_\pm(\alpha)+G_{1,0}(\kappa\nur),\rho_\pm(\alpha)-B_\Lambda\kappa\nur\right)+\Qrefpm,
 \end{eqnarray*}
 where we used a notation similar to Section \ref{sec:proof-thm1} (with a slight modification of the definition of $\Qrefpm$):
 $$\Qrefpm:=\projnegpm{\Dref}-P^0_{\Lambda,-},$$
 $$\Dref:=\Pi_\Lambda(D^{\kappa\nur}+\alpha R_{1,0}(\kappa\nur))\Pi_\Lambda,$$
 $$G_{n,\pm}(Q,\rho):=\frac{(-1)^{n+1}}{2\pi}\int_\R\frac{1}{\Dref-\mu_\pm+i\eta}\left((V_\rho-R_Q)\frac{1}{\Dref-\mu_\pm+i\eta}\right)^n\d{\eta}.$$
 Notice that $G_{1,0}(\kappa\nur)$ is independent of the index $\pm$. By \eqref{eq:est-Gn}, we have for all $n\ge2$,
 $$\|G_{n,\pm}(Q,\rho)\|_\cQ\le C\|(Q,\rho)\|^n,$$
so that 
 $$\|G_{n,\pm}(Q_\pm(\alpha)+G_{1,0}(\kappa\nur),\rho_\pm(\alpha)-B_\Lambda\kappa\nur)\|_\cQ\le C^n.$$
 On the other hand, 
 $$G_{1,\pm}(Q,\rho)=G_{1,0,\pm}(\rho)-G_{0,1,\pm}(Q),$$
with 
 $$G_{1,0,\pm}(\rho):=\frac{1}{2\pi}\int_\R\frac{1}{\Dref-\mu_\pm+i\eta}\Pi_\Lambda V_\rho\Pi_\Lambda\frac{1}{\Dref-\mu_\pm+i\eta}\d{\eta},$$
$$G_{0,1,\pm}(Q):=\frac{1}{2\pi}\int_\R\frac{1}{\Dref-\mu_\pm+i\eta}\Pi_\Lambda R_Q\Pi_\Lambda\frac{1}{\Dref-\mu_\pm+i\eta}\d{\eta}.$$
Combining \eqref{eq:estGnx} and  \cite[Lemma 11]{HaiLewSer-05a}, we get
$$\|G_{0,1,\pm}(Q)\|_{\cQ}\le C\|Q\|_\cQ,\qquad \|G_{1,0,\pm}(\rho)\|_\cQ\le C\sqrt{\log\Lambda}\|\rho\|_{L^2\cap\cC}.$$
Hence, using \eqref{eq:apriori} we get  
$$\left\|Q_\pm(\alpha)-\Qrefpm \right\|_\cQ=\cO_{\alpha\to0}(\sqrt{\alpha}).$$
 Now using \eqref{eq:est-Qref-Qlin}, we deduce
\begin{equation}\label{eq:QOsqrtalpha}
\left\|Q_\pm(\alpha)-\Qlinpm(\kappa\nur)\right\|_\cQ=\cO_{\alpha\to0}(\sqrt{\alpha}),
\end{equation}
where
$$\Qlinpm(\omega)=\projnegpm{D^\omega_\Lambda}-P^0_{\Lambda,-},$$
which in particular implies Corollary \ref{coro}. The self-consistent equation satisfied by $\rho_\pm(\alpha)$ is
\begin{multline*}
\rho_\pm(\alpha)=\rhorefpm-\alpha B_\Lambda(\rho_\pm(\alpha)-B_\Lambda\kappa\nur)+\alpha \cU_\Lambda(\rho_\pm(\alpha)-B_\Lambda\kappa\nur)\\
+\alpha\rho_{1,0,\pm,r}(\rho_\pm(\alpha)-B_\Lambda\kappa\nur)-\alpha\rho_{0,1,\pm}(Q_\pm(\alpha)+G_{1,0}(\kappa\nur))\\
+\sum_{n\ge2}\alpha^n\rho_{n,\pm}(Q_\pm(\alpha)+G_{1,0}(\kappa\nur),\rho_\pm(\alpha)-B_\Lambda\kappa\nur), 
\end{multline*}
where $\rho_{n,\pm}(Q,\rho):=\rho_{G_{n,\pm}(Q,\rho)}$ for all $n\ge2$, $\rhorefpm:=\rho_\Qrefpm$ and for any $\rho$, $Q$,
\begin{multline*}
\rho_{1,0,\pm,r}(\rho):=\rho\left[\frac{1}{2\pi}\int_\R\frac{1}{\Dref-\mu_\pm+i\eta}\Pi_\Lambda V_\rho\Pi_\Lambda\frac{1}{\Dref-\mu_\pm+i\eta}\d{\eta}\right.\\
\left.-\frac{1}{2\pi}\int_\R\frac{1}{D^0+i\eta}\Pi_\Lambda V_\rho\Pi_\Lambda\frac{1}{D^0+i\eta}\d{\eta}\right], 
\end{multline*}
$$\rho_{0,1,\pm}(Q)=\rho\left[\frac{1}{2\pi}\int_\R\frac{1}{\Dref-\mu_\pm+i\eta}\Pi_\Lambda R_Q\Pi_\Lambda\frac{1}{\Dref-\mu_\pm+i\eta}\d{\eta}\right].$$
We rewrite the self-consistent equation for $\rho_\pm(\alpha)$ as
\begin{multline}\label{eq:dev-rho}
\rho_\pm(\alpha)-B_\Lambda\kappa\nur=\frac{1}{1+\alpha B_\Lambda}(\rhorefpm-B_\Lambda\kappa\nur)+\frac{\alpha}{1+\alpha B_\Lambda}\cU_\Lambda(\rho_\pm(\alpha)-B_\Lambda\kappa\nur)\\
+\frac{\alpha}{1+\alpha B_\Lambda}\rho_{1,0,\pm,r}(\rho_\pm(\alpha)-B_\Lambda\kappa\nur)-\frac{\alpha}{1+\alpha B_\Lambda}\rho_{0,1,\pm}(Q_\pm(\alpha)+G_{1,0}(\kappa\nur))\\
+\frac{1}{1+\alpha B_\Lambda}\sum_{n\ge2}\alpha^n\rho_{n,\pm}(Q_\pm(\alpha)+G_{1,0}(\kappa\nur),\rho_\pm(\alpha)-B_\Lambda\kappa\nur).
\end{multline}
We now have the necessary estimates to compute the limit of $F(\kappa,\alpha)$ as $\alpha\to0$.

\subsubsection{Kinetic energy terms}

By \eqref{eq:QOsqrtalpha}, we have 
\begin{equation}\label{eq:diffQpm}
Q_+(\alpha)=Q_-(\alpha)+|\phi(\kappa\nur,\Lambda)\rangle\langle\phi(\kappa\nur,\Lambda)|+\cO^\cQ_{\alpha\to0}(\sqrt{\alpha}).
\end{equation}
For any projector $P$ such that $Q=P-P^0\in\cK_\Lambda$ we have
\begin{equation}\label{eq:estkineticnormQ}
\tr_0(D^0Q)=\||D^0|^\frac{1}{2}Q\|_{\gS_2}^2\le\|Q\|_\cQ^2. 
\end{equation}
The subtlety of this estimate is that it only controls the kinetic energy by the $\cQ$-norm only for $Q$'s that are of the form $P-P^0$. Indeed, it would be wrong to deduce from \eqref{eq:diffQpm} that
$$\tr_0(D^0Q_+(\alpha))-\tr_0(D^0Q_-(\alpha))=\langle\phi(\kappa\nur,\Lambda),D^0\phi(\kappa\nur,\Lambda)\rangle+\cO_{\alpha\to0}(\sqrt{\alpha})$$
since the operator $Q_+(\alpha)-Q_-(\alpha)-|\phi(\kappa\nur,\Lambda)\rangle\langle\phi(\kappa\nur,\Lambda)|$ is not of the form $P-P^0$ for some projector $P$. Instead, one should see $\tr_0(D^0Q_\pm(\alpha))$ as a quadratic term by the first equality in \eqref{eq:estkineticnormQ}. Writing $Q_\pm(\alpha)=\Qlinpm(\kappa\nur)+Q_{\pm,r}(\alpha)$, we thus have
\begin{multline}\label{eq:kineticpart}
\tr(|D^0|Q_\pm(\alpha)^2)=\tr(|D^0|\Qlinpm(\kappa\nur)^2)+\tr(|D^0|Q_{\pm,r}(\alpha)^2)\\
+\tr(|D^0|\{\Qlinpm(\kappa\nur),Q_{\pm,r}(\alpha)\}),
\end{multline}
where $\{A,B\}:=AB+BA$. The first term has a good behaviour since $\Qlinpm(\kappa\nur)$ is of the form $P-P^0$, hence the first term is linear in $\Qlinpm(\kappa\nur)$ and we have
\begin{multline}
\tr(|D^0|\Qlinp(\kappa\nur)^2)-\tr(|D^0|\Qlinm(\kappa\nur)^2)\\
=\tr_0(D^0(\Qlinp(\kappa\nur)-\Qlinm(\kappa\nur)))\\
=\langle \phi(\kappa\nur,\Lambda),D^0\phi(\kappa\nur,\Lambda)\rangle. 
\end{multline}
The second term in \eqref{eq:kineticpart} is controlled using that $\|Q_{\pm,r}(\alpha)\|_\cQ=\cO_{\alpha\to0}(\sqrt{\alpha})$, hence $\tr(|D^0|Q_{\pm,r}(\alpha)^2)\le\|Q_{\pm,r}(\alpha)\|_\cQ^2=\cO_{\alpha\to0}(\alpha).$ Finally, the term with the Poisson bracket in \eqref{eq:kineticpart} cannot be shown to go to zero as $\alpha\to0$. Indeed, it is just a $\cO_{\alpha\to0}(1)$ since $\|\Qlinpm(\kappa\nur)\|_\cQ=\cO_{\Lambda\to\ii}(\sqrt{\log\Lambda})$ by \eqref{eq:est-projlibre} and $\|Q_{\pm,r}(\alpha)\|_\cQ=\cO_{\alpha\to0}(\sqrt{\alpha})$. We will compute explicitly the non-vanishing term in this expression. We have
 \begin{multline}\label{eq:dev-kinetic-energy}
\tr_0(D^0Q_+(\alpha))-\tr_0(D^0Q_-(\alpha))=\left\langle\phi(\kappa\nur,\Lambda),D^0\phi(\kappa\nur,\Lambda)\right\rangle\\
+\tr(|D^0|\{\Qlinp(\kappa\nur),(Q_{+,r}(\alpha)-Q_{-,r}(\alpha))\})\\
+\langle\phi(\kappa\nur,\Lambda),\{|D^0|,Q_{-,r}(\alpha)\}\phi(\kappa\nur,\Lambda)\rangle+\cO_{\alpha\to0}(\alpha).  
 \end{multline}
 The computations are lengthy and done in Appendix \ref{appendix:kinetic}. We obtain:
\begin{multline}\label{eq:part1}
 \tr_0(D^0Q_+(\alpha))-\tr_0(D^0Q_-(\alpha))=\left\langle\phi(\kappa\nur,\Lambda),D^0\phi(\kappa\nur,\Lambda)\right\rangle\\
-\frac{\alpha B_\Lambda}{1+\alpha B_\Lambda}\langle\phi(\kappa\nur,\Lambda),V_{\kappa\nur}\phi(\kappa\nur,\Lambda)\rangle+\cO_{\alpha\to0}(\alpha).  
\end{multline}
In particular, notice that 
$$\lim_{\alpha\to0}\tr_0(D^0Q_+(\alpha))-\tr_0(D^0Q_-(\alpha))\neq\left\langle\phi(\kappa\nur),D^0\phi(\kappa\nur)\right\rangle,$$
meaning that the difference of kinetic energy of the interacting vacuums does \emph{not} converge to the difference of kinetic energy of the non-interacting vacuum. This shows the subtlety of the limit $\alpha\to0$. Fortunately, the extra term in \eqref{eq:part1} will be compensated by another one coming from the direct term.

\subsubsection{Potential energy term}

The second term to compute in $F$ is the only ``true'' linear term: that is $-D(\kappa\nu,\cdot)$. By \eqref{eq:dev-rho+rho-Oalpha}, we have
\begin{equation}
\rho_+(\alpha)-\rho_-(\alpha)=\frac{1}{1+\alpha B_\Lambda}|\phi(\kappa\nur,\Lambda)|^2+\cO_{\alpha\to0}^{L^2\cap\cC}(\alpha), 
\end{equation}
which implies that 
\begin{eqnarray}
\nonumber -D(\kappa\nu,\rho_+(\alpha)-\rho_-(\alpha)) & = & -\frac{1}{1+\alpha B_\Lambda}D(\kappa\nu,|\phi(\kappa\nur,\Lambda)|^2)+\cO_{\alpha\to0}(\alpha)\\
 & = &  -\langle\phi(\kappa\nur,\Lambda),V_{\kappa\nur}\phi(\kappa\nur,\Lambda)\rangle+\cO_{\alpha\to0}(\alpha). \label{eq:part2}
\end{eqnarray}

\subsubsection{Direct term}

The next term we have to treat is 
$$\frac{\alpha}{2}\Big(D\left(\rho_+(\alpha)\right)-D(\rho_-(\alpha))\Big)=\frac{\alpha}{2}\Big(D(\rho_+(\alpha)-\rho_-(\alpha))+2D(\rho_+(\alpha)-\rho_-(\alpha),\rho_-(\alpha))\Big),$$
where $D(\rho)$ is a shortcut notation for $D(\rho,\rho)$. By \eqref{eq:dev-rho+rho-Oalpha} and \eqref{eq:apriori}, we have
$$\alpha D(\rho_+(\alpha)-\rho_-(\alpha))=\cO_{\alpha\to0}(\alpha),$$
$$\alpha D(\rho_+(\alpha)-\rho_-(\alpha),\rho_-(\alpha))=\frac{\alpha B_\Lambda}{1+\alpha B_\Lambda}\langle\phi(\kappa\nur,\Lambda),V_{\kappa\nur}\phi(\kappa\nur,\Lambda)\rangle+\cO_{\alpha\to0}(\alpha).$$
Hence,
\begin{equation}\label{eq:part3}
\frac{\alpha}{2}\Big(D\left(\rho_+(\alpha)\right)-D(\rho_-(\alpha))\Big)=\frac{\alpha B_\Lambda}{1+\alpha B_\Lambda}\langle\phi(\kappa\nur,\Lambda),V_{\kappa\nur}\phi(\kappa\nur,\Lambda)\rangle+\cO_{\alpha\to0}(\alpha),
\end{equation}
which exactly compensates the spurious term in \eqref{eq:part1}. 

\subsubsection{Exchange term}

Finally, we treat the last term of the energy, the exchange term. For any operator $Q$, we have
\begin{eqnarray*}
\int_{\R^3\times\R^3}\frac{|Q(x,y)|^2}{|x-y|}\d{x}\d{y} & = & \int_{\R^3\times\R^3}\tr_{\C^4}\left(E(p+q)^{\frac{1}{2}}\hat{Q}(p,q)^*\frac{\hat{R_Q}(p,q)}{E(p+q)^\frac{1}{2}}\right)\d{p}\d{q}\\
 & \le & \|Q\|_\cQ\left(\int_{\R^3\times\R^3}\frac{E(p-q)^2}{E(p+q)}|\hat{R_Q}(p,q)|^2\d{p}\d{q}\right)^\frac{1}{2}\\
 & \le  & C\|Q\|_\cQ^2, 
\end{eqnarray*}
where in the last inequality we used \cite[Lemma 8]{HaiLewSer-05a}. The term we have to estimate is 
\begin{multline*}
\frac{\alpha}{2}\left(\int_{\R^6}\frac{|[Q_+(\alpha)](x,y)|^2}{|x-y|}\d{x}\d{y}-\int_{\R^6}\frac{|[Q_-(\alpha)](x,y)|^2}{|x-y|}\d{x}\d{y}\right)\\
=\frac{\alpha}{2}\int_{\R^6}\frac{|[Q_+(\alpha)-Q_-(\alpha)](x,y)|^2}{|x-y|}\d{x}\d{y}\\
+\alpha\int_{\R^6}\frac{\tr_{\C^4}\left([Q_+(\alpha)-Q_-(\alpha)](x,y)[Q_-(\alpha)](x,y)^*\right)}{|x-y|}\d{x}\d{y}.
\end{multline*}
By \eqref{eq:diffQpm} together with $\|Q_-(\alpha)\|_\cQ=\cO(\sqrt{\log\Lambda})$, we have
$$\frac{\alpha}{2}\int_{\R^6}\frac{|[Q_+(\alpha)-Q_-(\alpha)](x,y)|^2}{|x-y|}\d{x}\d{y}=\cO_{\alpha\to0}(\alpha),$$
\begin{multline}\label{eq:part4}
\alpha\int_{\R^6}\frac{\tr_{\C^4}\left([Q_+(\alpha)-Q_-(\alpha)](x,y)[Q_-(\alpha)](x,y)^*\right)}{|x-y|}\d{x}\d{y}\\
=\alpha\langle\phi(\kappa\nur,\Lambda),R_{Q_-(\alpha)}\phi(\kappa\nur,\Lambda)\rangle+\cO_{\alpha\to0}(\alpha). 
\end{multline}
Using that
$$|\langle\phi(\kappa\nur,\Lambda),R_{Q_-(\alpha)}\phi(\kappa\nur,\Lambda)\rangle|\le\|R_{Q_-(\alpha)}\|_{H^1\to L^2}\|\phi(\kappa\nur,\Lambda)\|_{H^1}^2,$$
and the fact that $\|R_{Q_-(\alpha)}\|_{H^1\to L^2}=\|R_{Q_-(\alpha)}|D^0|^{-1}\|_{L^2\to L^2}\le C\|Q_-(\alpha)\|_{\gS_2}=\cO_{\alpha\to0}(1)$, we obtain
$$\frac{\alpha}{2}\left(\int_{\R^6}\frac{|[Q_+(\alpha)](x,y)|^2}{|x-y|}\d{x}\d{y}-\int_{\R^6}\frac{|[Q_-(\alpha)](x,y)|^2}{|x-y|}\d{x}\d{y}\right)=\cO_{\alpha\to0}(\alpha).$$
Regrouping \eqref{eq:part1}, \eqref{eq:part2}, \eqref{eq:part3}, and \eqref{eq:part4} together, we get
$$F(\kappa,\alpha)=\lambda(\kappa\nur,\Lambda)+\cO_{\alpha\to0}(\alpha).$$
Since $\lambda(\kappa\nur)=\lambda(\kappa\nur,\Lambda)+\cO(E(\Lambda)^{-1})=\lambda(\kappa\nur,\Lambda)+\cO(e^{-\frac{1}{\alpha}})=\lambda(\kappa\nur,\Lambda)+\cO(\alpha)$, we deduce the theorem. \qed

\appendix

\section{Kinetic energy estimates}\label{appendix:kinetic}

We expand $Q_{+,r}(\alpha)-Q_{-,r}(\alpha)$ up to $\cO_{\alpha\to0}(\alpha^{3/2})$. Recall that we have
\begin{multline*}
Q_\pm(\alpha)=\Qrefpm+\alpha G_{1,0,\pm}(\rho_\pm(\alpha)-B_\Lambda\kappa\nur)\\
-\alpha G_{0,1,\pm}(Q_\pm(\alpha)+G_{1,0}(\kappa\nur))+\cO^\cQ_{\alpha\to0}(\alpha^2). 
\end{multline*}
Then, using similar methods as in the proof of Propositions \ref{prop:main} and \ref{prop:est-Qref-rhoref}, we obtain 
$$G_{1,0,\pm}(\rho_\pm(\alpha)-B_\Lambda\kappa\nur)=G_{1,0,\pm}^{(\kappa)}(\rho_\pm(\alpha)-B_\Lambda\kappa\nur)+\cO^\cQ_{\alpha\to0}(\sqrt{\alpha}),$$
$$G_{0,1,\pm}(Q_\pm(\alpha)+G_{1,0}(\kappa\nur))=G_{0,1,\pm}^{(\kappa)}(Q_\pm(\alpha)+G_{1,0}(\kappa\nur))+\cO^\cQ_{\alpha\to0}(\sqrt{\alpha}),$$
with the notation
$$G_{1,0,\pm}^{(\kappa)}(\rho):=\frac{1}{2\pi}\int_\R\frac{1}{D^{\kappa\nur}_\Lambda-\mu_\pm+i\eta}\Pi_\Lambda V_\rho\Pi_\Lambda\frac{1}{D^{\kappa\nur}_\Lambda-\mu_\pm+i\eta}\d{\eta},$$
$$G_{0,1,\pm}^{(\kappa)}(Q):=\frac{1}{2\pi}\int_\R\frac{1}{D^{\kappa\nur}_\Lambda-\mu_\pm+i\eta}\Pi_\Lambda R_Q\Pi_\Lambda\frac{1}{D^{\kappa\nur}_\Lambda-\mu_\pm+i\eta}\d{\eta}.$$
Hence, 
\begin{multline*}
Q_\pm(\alpha)=\Qrefpm+\alpha G_{1,0,\pm}^{(\kappa)}(\rho_\pm(\alpha)-B_\Lambda\kappa\nur)\\
-\alpha G_{0,1,\pm}^{(\kappa)}(Q_\pm(\alpha)+G_{1,0}(\kappa\nur))+\cO^\cQ_{\alpha\to0}(\alpha^{3/2}). 
\end{multline*}
This leads to the decomposition
$$Q_{+,r}(\alpha)-Q_{-,r}(\alpha)=Q_1+\alpha Q_2-\alpha Q_3+\cO^\cQ_{\alpha\to0}(\alpha^{3/2}),$$
with
$$Q_1:=\Qrefp-\Qlinp(\kappa\nur)-(\Qrefm-\Qlinm(\kappa\nur)),$$
$$Q_2:=G_{1,0,+}^{(\kappa)}(\rho_+(\alpha)-B_\Lambda\kappa\nur)-G_{1,0,-}^{(\kappa)}(\rho_-(\alpha)-B_\Lambda\kappa\nur),$$
$$Q_3:=G_{0,1,+}^{(\kappa)}(Q_+(\alpha)+G_{1,0}(\kappa\nur))-G_{0,1,-}^{(\kappa)}(Q_-(\alpha)+G_{1,0}(\kappa\nur)).$$
We will treat each of these three terms separately. 

\subsection{Estimate on $Q_1$.}

To expand $Q_1$, we write
$$\Qrefpm-\Qlinpm(\kappa\nur)=\alpha G_{0,1,\pm}^{(\kappa)}(G_{1,0}(\kappa\nur))-\alpha^2 G_{0,2,\pm}^{(\kappa)}(G_{1,0}(\kappa\nur))+\cO^\cQ_{\alpha\to0}(\alpha^{3/2}),$$
where for any $Q$ we have
$$G_{0,2,\pm}^{(\kappa)}(Q):=\frac{1}{2\pi}\int_\R\frac{1}{D^{\kappa\nur}_\Lambda-\mu_\pm+i\eta}\left(\Pi_\Lambda R_Q\Pi_\Lambda\frac{1}{D^{\kappa\nur}_\Lambda-\mu_\pm+i\eta}\right)^2\d{\eta}.$$
Hence,
$$Q_1=\alpha(G_{0,1,+}^{(\kappa)}-G_{0,1,-}^{(\kappa)})(G_{1,0}(\kappa\nur))-\alpha^2(G_{0,2,+}^{(\kappa)}-G_{0,2,-}^{(\kappa)})(G_{1,0}(\kappa\nur))+\cO^\cQ_{\alpha\to0}(\alpha^{3/2}).$$
For $p=1,2$ and $Q\in\cQ$, we notice that
$$(G_{0,p,+}^{(\kappa)}-G_{0,p,-}^{(\kappa)})(Q)=\frac{1}{2i\pi}\oint_{\CJ}\frac{1}{D^{\kappa\nur}_\Lambda-z}\left(\Pi_\Lambda R_Q\Pi_\Lambda\frac{1}{D^{\kappa\nur}_\Lambda-z}\right)^p\d{z},$$
where $\CJ$ is a contour around 0 in the complex plane which intersects the real axis only at $\mu_-$ and $\mu_+$. By the residuum formula, one computes that
$$(G_{0,1,+}^{(\kappa)}-G_{0,1,-}^{(\kappa)})(G_{1,0}(\kappa\nur))=|\phi(\kappa\nur,\Lambda)\rangle\langle S(\kappa\nur)\phi(\kappa\nur,\Lambda)|+c.c.,$$
\begin{multline*}
(G_{0,2,+}^{(\kappa)}-G_{0,2,-}^{(\kappa)})(G_{1,0}(\kappa\nur))=|\phi(\kappa\nur,\Lambda)\rangle\langle S(\kappa\nur)^2\phi(\kappa\nur,\Lambda)|+c.c.\\
+|S(\kappa\nur)\phi(\kappa\nur,\Lambda)\rangle\langle S(\kappa\nur)\phi(\kappa\nur,\Lambda)|, 
\end{multline*}
where $c.c.$ denotes the adjoint of the preceding operator and 
$$S(\kappa\nur):=\frac{P_\phi^\perp}{D^{\kappa\nur}_\Lambda-\lambda(\kappa\nur,\Lambda)}R_{1,0}(\kappa\nur),$$
the operator $P_\phi^\perp$ being the projection on the orthogonal to $\phi(\kappa\nur,\Lambda)$ in $\gH_\Lambda$. By \eqref{eq:ineq-RQD0-1L2L2}, we have
$$\|S(\kappa\nur)\|_{H^1\to H^1}=\cO_{\alpha\to0}(1).$$
Since the we have to estimate the kinetic energy of these quantities, we introduce a new norm corresponding exactly to the kinetic energy, for which the rank one terms above are easier to estimate. For any $Q$ such that the following expression is well-defined, we set $\|Q\|_{\text{kin}}:=\||D^0|^{1/2}Q\|_{\gS_2}$. We thus see that for any $\psi,\zeta\in H^1$, $\||\psi\rangle\langle\zeta|\|_{\text{kin}}\le\|\psi\|_{H^1}\|\zeta\|_{H^1}$. For this norm, we can write 
\begin{multline*}
Q_1=\alpha(|\phi(\kappa\nur,\Lambda)\rangle\langle S(\kappa\nur)\phi(\kappa\nur,\Lambda)|+|S(\kappa\nur)\phi(\kappa\nur,\Lambda)\rangle\langle\phi(\kappa\nur,\Lambda)|)\\
+\cO^{\text{kin}}_{\alpha\to0}(\alpha^{3/2})=\cO_{\alpha\to0}^{\text{kin}}(\alpha). 
\end{multline*}
Using now that $\Qlinp(\kappa\nur)=-G_{1,0}(\kappa\nur)+\cO_{\alpha\to0}^\cQ(1)$, we deduce
\begin{multline*}
\tr(|D^0|\{\Qlinp(\kappa\nur),Q_1\})\\
=-2\alpha\text{Re}\langle\phi(\kappa\nur,\Lambda),\{|D^0|,G_{1,0}(\kappa\nur)\}S(\kappa\nur)\phi(\kappa\nur,\Lambda)\rangle+\cO_{\alpha\to0}(\alpha). 
\end{multline*}
Next, we have
\begin{multline*}
\left|\langle\phi(\kappa\nur,\Lambda),\{|D^0|,G_{1,0}(\kappa\nur)\}S(\kappa\nur)\phi(\kappa\nur,\Lambda)\rangle\right|\\
\le2\|G_{1,0}(\kappa\nur)\|_{\gS_2}\|S(\kappa\nur)\|_{H^1\to L^2}\|\phi(\kappa\nur,\Lambda)\|_{H^1}^2=\cO_{\alpha\to0}(1), 
\end{multline*}
so that 
$$\tr(|D^0|\{\Qlinp(\kappa\nur),Q_1\})=\cO_{\alpha\to0}(\alpha).$$

\subsection{Estimate on $Q_2$.}

We split $Q_2$ as
$$Q_2=(G_{1,0,+}^{(\kappa)}-G_{1,0,-}^{(\kappa)})(\rho_+(\alpha)-B_\Lambda\kappa\nur)+G_{1,0,-}^{(\kappa)}(\rho_+(\alpha)-\rho_-(\alpha)).$$
In the same fashion as before, the residuum formula leads to
\begin{multline*}
(G_{1,0,+}^{(\kappa)}-G_{1,0,-}^{(\kappa)})(\rho_+(\alpha)-B_\Lambda\kappa\nur)=\\
|\phi(\kappa\nur,\Lambda)\rangle\langle S(V_+)\phi(\kappa\nur,\Lambda)|+|S(V_+)\phi(\kappa\nur,\Lambda)\rangle\langle\phi(\kappa\nur,\Lambda)|, 
\end{multline*}
with $V_+:=V_{\rho_+(\alpha)-B_\Lambda\kappa\nur}$ and 
$$S(V_+):=\frac{P_\phi^\perp}{D^{\kappa\nur}_\Lambda-\lambda(\kappa\nur,\Lambda)}V_+.$$
Notice that we have $\|S(V_+)\|_{L^2\to H^1}=\cO_{\alpha\to0}(1)$. The second term in the expansion of $Q_2$ requires to develop $\rho_+(\alpha)-\rho_-(\alpha)$ up to $\cO(\alpha)$. Combining \eqref{eq:dev-rho} with the estimates valid  for any $\rho\in L^2\cap\cC$, any $Q\in\cQ$, and any $n\ge2$,
$$\|\rho_{n,\pm}(Q,\rho)\|_{L^2\cap\cC}\le C\|(Q,\rho)\|,\qquad\|\rho_{1,0,\pm,r}(\rho)\|_{L^2\cap\cC}\le C\|\rho\|_{L^2\cap\cC},$$
$$\rho_{0,1,\pm}(Q_\pm(\alpha)+G_{1,0}(\kappa\nur))=\rho_{0,1}(Q_\pm(\alpha)+G_{1,0}(\kappa\nur))+\cO^{L^2\cap\cC}_{\alpha\to0}(1),$$
where $\rho_{0,1}(Q):=\rho[G_{0,1}(Q)]$, we deduce that
\begin{multline*}
\rho_+(\alpha)-\rho_-(\alpha)=\frac{1}{1+\alpha B_\Lambda}(\rhorefp-\rhorefm)+\frac{\alpha}{1+\alpha B_\Lambda}\cU_\Lambda(\rho_+(\alpha)-\rho_-(\alpha))\\
-\frac{\alpha}{1+\alpha B_\Lambda}\rho_{0,1}(Q_+(\alpha)-Q_-(\alpha))+\cO^{L^2\cap\cC}_{\alpha\to0}(\alpha). 
\end{multline*}
Since $\rhorefp-\rhorefm=|\phi(\kappa\nur,\Lambda)|^2+\rho_{Q_1}$ and 
$$\rho_{Q_1}=2\alpha\text{Re}\left(\phi(\kappa\nur,\Lambda)\bar{S(\kappa\nur)\phi(\kappa\nur,\Lambda)}\right)+\cO^{L^2\cap\cC}_{\alpha\to0}(\alpha),$$
it remains to estimate $\|\text{Re}(\psi\bar{\zeta})\|_{L^2\cap\cC}$ for any $\psi,\zeta\in H^1$. First, by the Sobolev imbedding $H^1\hookrightarrow L^6$, we have $\|\text{Re}(\psi\bar{\zeta})\|_{L^2}\le\|\psi\|_{L^4}\|\zeta\|_{L^4}\lesssim\|\psi\|_{H^1}\|\zeta\|_{H^1}$. Secondly, by the Hardy-Littlewood-Sobolev inequality, we have
$$\|\text{Re}(\psi\bar{\zeta})\|_{\cC}^2\lesssim\||\psi|^2\|_{L^{6/5}}\||\zeta|^2\|_{L^{6/5}}\lesssim\|\psi\|_{H^1}^2\|\zeta\|_{H^1}^2.$$
Hence,
$$\left\|\text{Re}\left(\phi(\kappa\nur,\Lambda)\bar{S(\kappa\nur)\phi(\kappa\nur,\Lambda)}\right)\right\|_{L^2\cap\cC}=\cO_{\alpha\to0}(1),$$
and thus $\rhorefp-\rhorefm=|\phi(\kappa\nur,\Lambda)|^2+\cO^{L^2\cap\cC}_{\alpha\to0}(\alpha)$. For the third term in the expansion of $\rho_+(\alpha)-\rho_-(\alpha)$, we use $Q_+(\alpha)-Q_-(\alpha)=|\phi\rangle\langle\phi|+\cO^\cQ(\sqrt{\alpha})$ and the estimate $\|\rho_{0,1}(Q)\|_{L^2\cap\cC}\le C\sqrt{\log\Lambda}\|Q\|_\cQ$ to infer 
$$\rho_{0,1}(Q_+(\alpha)-Q_-(\alpha))=\rho_{0,1}(|\phi\rangle\langle\phi|)+\cO^{L^2\cap\cC}_{\alpha\to0}(1).$$

\begin{lemma}\label{lemma:est-rho01phiphi}
We have the estimate
 $$\|\rho_{0,1}(|\phi\rangle\langle\phi|)\|_{L^2\cap\cC}=\cO_{\alpha\to0}(1).$$
\end{lemma}

\begin{proof}[Proof of Lemma \ref{lemma:est-rho01phiphi}]
We estimate the Coulomb norm by duality: for $\zeta\in\cC'$ a smooth function, we have
$$|\tr(G_{0,1}(|\phi\rangle\langle\phi|)\zeta)|=|\tr(R_{|\phi\rangle\langle\phi|}G_{1,0}(\zeta))|\le\|R_{|\phi\rangle\langle\phi|}\|_{\gS_2}\|G_{1,0}(\zeta)\|_{\gS_2}\le C\|\phi\|_{H^1}^2\|\zeta\|_{\cC'}.$$
Hence, $\|\rho_{0,1}(|\phi\rangle\langle\phi|)\|_{\cC}=\cO_{\alpha\to0}(1)$. We cannot use the same method to estimate the $L^2$-norm, for which we have to be more precise. Let us use momentarily the shortcut notation $\rho_{0,1}:=\rho_{0,1}(|\phi\rangle\langle\phi|)$ and recall that for any $k\in\R^3$ we have
\begin{eqnarray*}
 |\hat{\rho_{0,1}}(k)|^2 & = & \left|\int_{\substack{|\ell-k/2|\le\Lambda \\ |\ell+k/2|\le\Lambda}}\tr_{\C^4}\left(\hat{R_{|\phi\rangle\langle\phi|}}(\ell-k/2,\ell+k/2)M(\ell-k/2,\ell+k/2)\right)\d{\ell}\right|^2\\
 & \le & \int_{\R^3}\left|\hat{R_{|\phi\rangle\langle\phi|}}(\ell-k/2,\ell+k/2)\right|^2\d{\ell}\int_{\R^3}|M(\ell-k/2,\ell+k/2)|^2\d{\ell}\\
& \le  &  C|k|^2\int_{\R^3}\left|\hat{R_{|\phi\rangle\langle\phi|}}(\ell-k/2,\ell+k/2)\right|^2\d{\ell},
\end{eqnarray*}
where we used the estimate $|M(p,q)|^2\le C|p-q|^2E(p+q)^{-4}$ (recall that $M(p,q)$ was defined in \eqref{eq:def-M}). As a consequence, we have
$$\|\rho_{0,1}\|_{L^2}^2\le C\int_{\R^3}|p-q|^2\left|\hat{R_{|\phi\rangle\langle\phi|}}(p,q)\right|^2\d{p}\d{q}=C\|[-i\nabla,R_{|\phi\rangle\langle\phi|}]\|_{(\gS_2)^3}^2.$$
But since $[-i\nabla,R_{|\phi\rangle\langle\phi|}](x,y)=\frac{\nabla\phi(x)\phi(y)^*+\phi(x)\nabla\phi(y)^*}{|x-y|}$, we infer from the Hardy inequality that $\|\rho_{0,1}\|_{L^2}\le C\|\phi\|_{H^1}^2$, so that $\|\rho_{0,1}\|_{L^2}=\cO_{\alpha\to0}(1)$. 
\end{proof}

The last term to estimate in $\rho_+(\alpha)-\rho_-(\alpha)$ is $\cU_\Lambda(\rho_+(\alpha)-\rho_-(\alpha))$. Using \eqref{eq:dev-rho}, we have for all $k\in\R^3$,
\begin{multline*}
U_\Lambda(|k|)\cF[\rho_\pm(\alpha)-B_\Lambda\kappa\nur](k)=\frac{U_\Lambda(|k|)}{1+\alpha B_\Lambda(k)}\cF[\rhorefpm-B_\Lambda\kappa\nur](k)\\
+\frac{\alpha U_\Lambda(|k|)}{1+\alpha B_\Lambda(k)}\cF[\rho_{1,0,\pm,r}(\rho_\pm(\alpha)-B_\Lambda\kappa\nur)](k)-\frac{\alpha U_\Lambda(|k|)}{1+\alpha B_\Lambda(k)}\cF[\rho_{0,1,\pm}(Q_\pm(\alpha)+G_{1,0}(\kappa\nur))](k)\\
+\frac{\alpha U_\Lambda(|k|)}{1+\alpha B_\Lambda(k)}\sum_{n\ge2}\alpha^{n-1}\cF[\rho_{n,\pm}(Q_\pm(\alpha)+G_{1,0}(\kappa\nur),\rho_\pm(\alpha)-B_\Lambda\kappa\nur)](k).
\end{multline*}
Since 
$$\left|\frac{\alpha U_\Lambda(|k|)}{1+\alpha B_\Lambda(k)}\right|\le \alpha B_\Lambda +\left|\frac{\alpha B_\Lambda(k)}{1+\alpha B_\Lambda(k)}\right|\le C,$$
we have
$$\cU_\Lambda(\rho_\pm(\alpha)-B_\Lambda\kappa\nur)=\cV_\Lambda(\rhorefpm-B_\Lambda\kappa\nur)-\alpha\cV_\Lambda\rho_{0,1}(Q_\pm(\alpha)+G_{1,0}(\kappa\nur))+\cO_{\alpha\to0}^{L^2\cap\cC}(1),$$
where $\cV_\Lambda$ is the Fourier multiplier by $k\mapsto\frac{U_\Lambda(|k|)}{1+\alpha B_\Lambda(k)}$. Hence, for the difference of densities we obtain 
$$\cU_\Lambda(\rho_+(\alpha)-\rho_-(\alpha))=\cV_\Lambda(\rhorefp-\rhorefm)-\alpha\cV_\Lambda\rho_{0,1}(Q_+(\alpha)-Q_-(\alpha))+\cO_{\alpha\to0}^{L^2\cap\cC}(1).$$
We have already seen that $\rho_{0,1}(Q_+(\alpha)-Q_-(\alpha))=\cO_{\alpha\to0}^{L^2\cap\cC}(1)$, and that $\rhorefp-\rhorefm=|\phi(\kappa\nur,\Lambda)|^2+\cO^{L^2\cap\cC}_{\alpha\to0}(\alpha)$, which implies
$$\cU_\Lambda(\rho_+(\alpha)-\rho_-(\alpha))=\cV_\Lambda|\phi(\kappa\nur,\Lambda)|^2+\cO_{\alpha\to0}^{L^2\cap\cC}(1).$$
By \eqref{eq:est-ULambda} and the estimate $(1+\alpha B_\Lambda(k))^{-1}\le1$ following from $B_\Lambda(k)\ge0$, it is enough to know that $|\phi|^2$ is bounded in $H^1$ to prove that $\cU_\Lambda(\rho_+(\alpha)-\rho_-(\alpha))=\cO_{\alpha\to0}^{L^2\cap\cC}(1)$. Since $\||\phi|^2\|_{H^1}\le C\|\phi\|_{H^2}^2$ and $\phi$ is an eigenvector of $D^{\kappa\nur}_\Lambda$, it is sufficient to control $\|\nabla(V_{\kappa\nur}\phi)\|_{L^2}$. By the Hardy-Littlewood-Sobolev together with $|\nabla V_\nu|\le C|\nu|\star|\cdot|^{-2}$, $\nu\in L^2$ implies that $\nabla V_\nu\in L^6$, and $\phi\in H^1\hookrightarrow L^3$ implies that $\phi\nabla V_{\kappa\nur}\in L^2$. Hence, $\phi=\cO^{H^2}_{\alpha\to0}(1)$ and $\cU_\Lambda(\rho_+(\alpha)-\rho_-(\alpha))=\cO^{L^2\cap\cC}_{\alpha\to0}(1)$. After these painful estimates, we obtain 
\begin{equation}\label{eq:dev-rho+rho-Oalpha}
\rho_+(\alpha)-\rho_-(\alpha)=\frac{1}{1+\alpha B_\Lambda}|\phi(\kappa\nur,\Lambda)|^2+\cO_{\alpha\to0}^{L^2\cap\cC}(\alpha), 
\end{equation}
which in turn implies
\begin{multline*}
Q_2=|\phi(\kappa\nur,\Lambda)\rangle\langle S(V_+)\phi(\kappa\nur,\Lambda)|+|S(V_+)\phi(\kappa\nur,\Lambda)\rangle\langle\phi(\kappa\nur,\Lambda)|\\
+\frac{1}{1+\alpha B_\Lambda}G_{1,0,-}^{(\kappa)}(|\phi(\kappa\nur,\Lambda)|^2)+\cO^\cQ_{\alpha\to0}(\sqrt{\alpha}). 
\end{multline*}
This leads to the expansion 
\begin{multline*}
\tr(|D^0|\{\Qlinp(\kappa\nur),\alpha Q_2\})=-2\alpha\text{Re}\langle\phi(\kappa\nur,\Lambda),\{|D^0|,G_{1,0}(\kappa\nur)\}S(V_+)\phi(\kappa\nur,\Lambda)\rangle\\
+\frac{\alpha}{1+\alpha B_\Lambda}\tr(|D^0|\{\Qlinp(\kappa\nur),G_{1,0,-}^{(\kappa)}(|\phi(\kappa\nur,\Lambda)|^2)\})+\cO_{\alpha\to0}(\alpha) 
\end{multline*}
As before we have 
$$\left|\langle\phi(\kappa\nur,\Lambda),\{|D^0|,G_{1,0}(\kappa\nur)\}S(V_+)\phi(\kappa\nur,\Lambda)\rangle\right|=\cO_{\alpha\to0}(1).$$
Furthermore, we write $\Qlinp(\kappa\nur)=-G_{1,0}(\kappa\nur)+\Qlinpr(\kappa\nur)$ and $G_{1,0,-}^{(\kappa)}(|\phi|^2)=G_{1,0}(|\phi|^2)+G_{1,0,-,r}^{(\kappa)}(|\phi|^2)$, so that 
\begin{multline*}
\tr(|D^0|\{\Qlinp(\kappa\nur),G_{1,0,-}^{(\kappa)}(|\phi|^2)\})=-\tr(|D^0|\{G_{1,0}(\kappa\nur),G_{1,0}(|\phi|^2)\})\\
-\tr(|D^0|\{G_{1,0}(\kappa\nur),G_{1,0,-,r}^{(\kappa)}(|\phi|^2)\})+\tr(|D^0|\{\Qlinpr(\kappa\nur),G_{1,0}(|\phi|^2)\})\\
+\tr(|D^0|\{\Qlinpr(\kappa\nur),G_{1,0,-,r}^{(\kappa)}(|\phi|^2)\}).
\end{multline*}
We know that $\|G_{1,0,-,r}^{(\kappa)}(|\phi|^2)\|_\cQ\le C\||\phi|^2\|_{L^2\cap\cC}\le C\|\phi\|_{H^1}^2=\cO_{\alpha\to0}(1)$, and $\|\Qlinpr(\kappa\nur)\|_\cQ=\cO_{\alpha\to0}(1)$. Hence,
$$\tr(|D^0|\{\Qlinpr(\kappa\nur),G_{1,0,-,r}^{(\kappa)}(|\phi|^2)\})=\cO_{\alpha\to0}(1).$$
By \cite[Lemma 6]{GraLewSer-09} and the estimates of the first part of the paper, we have for any $0\le\tau<1/2$ and for any $\rho\in L^2\cap\cC$,
$$\||D^0|^\tau G_{1,0}(\rho)\|_{\gS_2}\le C\|\rho\|_\cC,$$
$$\||D^0|^{\tau+1/2}\Qlinpr(\rho)\|_{\gS_2}\le C\|\rho\|_{L^2\cap\cC},$$
$$\||D^0|^{\tau+1/2}G_{1,0,-,r}^{(\kappa)}(\rho)\|_{\gS_2}\le C\|\rho\|_{L^2\cap\cC}.$$
Splitting $|D^0|=|D^0|^\tau|D^0|^{1-\tau}$ for some $0<\tau<1/2$, we thus get
$$-\tr(|D^0|\{G_{1,0}(\kappa\nur),G_{1,0,-,r}^{(\kappa)}(|\phi|^2)\})+\tr(|D^0|\{\Qlinpr(\kappa\nur),G_{1,0}(|\phi|^2)\})=\cO_{\alpha\to0}(1).$$
Finally,
$$\tr(|D^0|\{\Qlinp(\kappa\nur),\alpha Q_2\})=-\frac{\alpha}{1+\alpha B_\Lambda}\tr(|D^0|\{G_{1,0}(\kappa\nur),G_{1,0}(|\phi|^2)\})+\cO_{\alpha\to0}(\alpha).$$
We simplify the last expression in a general fashion. Let $\rho_1,\rho_2\in L^2\cap\cC$ and $T(\rho_1,\rho_2):=\tr(|D^0|\{G_{1,0}(\rho_1),G_{1,0}(\rho_2)\})$. Then,
\begin{eqnarray*}
 T(\rho_1,\rho_2) & = & \frac{1}{32\pi^3}\int_{B(0,\Lambda)^2}(E(p)+E(q))\hat{V_{\rho_1}}(p-q)\bar{\hat{V_{\rho_2}}(p-q)}\tr_{\C^4}(M(p,q)M(q,p))\d{p}\d{q} \\
 & = & \frac{1}{4\pi^3}\int_{B(0,\Lambda)^2}\hat{V_{\rho_1}}(p-q)\bar{\hat{V_{\rho_2}}(p-q)}\frac{1}{E(p)+E(q)}\left(1-\frac{1+p\cdot q}{E(p)E(q)}\right)\d{p}\d{q}\\
& = & \frac{1}{4\pi}\int_{B(0,2\Lambda)}\hat{V_{\rho_1}}(k)\bar{\hat{V_{\rho_2}}(k)}|k|^2B_ \Lambda(k)\d{k}\\
& = & 4\pi\int_{B(0,2\Lambda)}B_\Lambda(k)\frac{\hat{\rho_1}(k)\bar{\hat{\rho_2}(k)}}{|k|^2}\d{k}\\
& = & B_\Lambda D(\rho_1,\rho_2)-4\pi\int_{B(0,2\Lambda)}U_\Lambda(|k|)\frac{\hat{\rho_1}(k)\bar{\hat{\rho_2}(k)}}{|k|^2}\d{k}.
\end{eqnarray*}
where we recall that 
$$B_\Lambda(k)=\frac{1}{\pi^2|k|^2}\int_{\substack{|\ell+k/2|\le\Lambda \\ |\ell-k/2|\le\Lambda}}\frac{1}{E(\ell+k/2)+E(\ell-k/2)}\left(1-\frac{1+(\ell+k/2)\cdot(\ell-k/2)}{E(\ell+k/2)E(\ell-k/2)}\right)\d{\ell}.$$
In our case $\rho_1=\kappa\nur$ and $\rho_2=|\phi|^2$, since $\cU_\Lambda\nu\in\cC$ because $\nu\in L^2\cap\cC$, we deduce
\begin{multline*}
\tr(|D^0|\{G_{1,0}(\kappa\nur),G_{1,0}(|\phi|^2)\})=B_\Lambda D(\kappa\nur,|\phi|^2)+\cO_{\alpha\to0}(1)\\
=B_\Lambda\langle\phi(\kappa\nur,\Lambda),V_{\kappa\nur}\phi(\kappa\nur,\Lambda)\rangle+\cO_{\alpha\to0}(1), 
\end{multline*}
thus
$$\tr(|D^0|\{\Qlinp(\kappa\nur),\alpha Q_2\})=-\frac{\alpha B_\Lambda}{1+\alpha B_\Lambda}\langle\phi(\kappa\nur,\Lambda),V_{\kappa\nur}\phi(\kappa\nur,\Lambda)\rangle+\cO_{\alpha\to0}(\alpha).$$

\subsection{Estimate on $Q_3$.}

As for $Q_2$, we first write
$$Q_3=(G_{0,1,+}^{(\kappa)}-G_{0,1,-}^{(\kappa)})(Q_+(\alpha)+G_{1,0}(\kappa\nur))+G_{0,1,-}^{(\kappa)}(Q_+(\alpha)-Q_-(\alpha)).$$
By the residuum formula, 
\begin{multline*}
(G_{0,1,+}^{(\kappa)}-G_{0,1,-}^{(\kappa)})(Q_+(\alpha)+G_{1,0}(\kappa\nur))\\
=|\phi(\kappa\nur,\Lambda)\rangle\langle S(R_+)\phi(\kappa\nur,\Lambda)|+|S(R_+)\phi(\kappa\nur,\Lambda)\rangle\langle\phi(\kappa\nur,\Lambda)|, 
\end{multline*}
where
$$S(R_+):=\frac{P_{\phi}^\perp}{D^{\kappa\nur}_\Lambda-\lambda(\kappa\nur,\Lambda)}R(Q_+(\alpha)+G_{1,0}(\kappa\nur)),$$
so that $\|S(R_+)\|_{H^1\to H^1}\le C\|Q_+(\alpha)+G_{1,0}(\kappa\nur)\|_\cQ=\cO_{\alpha\to0}(1)$. Since $\|G_{0,1,-}^{(\kappa)}(Q)\|_\cQ\le C\|Q\|_\cQ$ and $Q_+(\alpha)-Q_-(\alpha)=|\phi\rangle\langle\phi|+\cO^\cQ_{\alpha\to0}(\sqrt{\alpha})$, we have
$$G_{0,1,-}^{(\kappa)}(Q_+(\alpha)-Q_-(\alpha))=G_{0,1,-}^{(\kappa)}(|\phi\rangle\langle\phi|)+\cO^\cQ_{\alpha\to0}(\sqrt{\alpha}).$$
This leads to the expression
\begin{multline*}
 \tr(|D^0|\{\Qlinp(\kappa\nur),\alpha Q_3\})=2\alpha\text{Re}\langle\phi(\kappa\nur,\Lambda),\{|D^0|,\Qlinp(\kappa\nur)\}S(R_+)\phi(\kappa\nur,\Lambda)\rangle\\
+\alpha\tr(|D^0|\{\Qlinp(\kappa\nur),G_{0,1,-}^{(\kappa)}(|\phi\rangle\langle\phi|)\})+\cO_{\alpha\to0}(\alpha).
\end{multline*}
Again we have
$$|\langle\phi(\kappa\nur,\Lambda),\{|D^0|,\Qlinp(\kappa\nur)\}S(R_+)\phi(\kappa\nur,\Lambda)\rangle|=\cO_{\alpha\to0}(1),$$
and 
\begin{multline*}
 \tr(|D^0|\{\Qlinp(\kappa\nur),G_{0,1,-}^{(\kappa)}(|\phi\rangle\langle\phi|)\})=-\tr(|D^0|\{G_{1,0}(\kappa\nur),G_{0,1}(|\phi\rangle\langle\phi|)\})\\
-\tr(|D^0|\{G_{1,0}(\kappa\nur),G_{0,1,-,r}^{(\kappa)}(|\phi\rangle\langle\phi|)\})+\tr(|D^0|\{\Qlinpr(\kappa\nur),G_{0,1,-}^{(\kappa)}(|\phi\rangle\langle\phi|)\})\\
+\tr(|D^0|\{\Qlinpr(\kappa\nur),G_{0,1,-,r}^{(\kappa)}(|\phi\rangle\langle\phi|)\}).
\end{multline*}
Since $\|\Qlinpr(\kappa\nur)\|_\cQ=\cO_{\alpha\to0}(1)$, $\|G_{0,1,-,r}^{(\kappa)}(|\phi\rangle\langle\phi|)\|_\cQ=\cO_{\alpha\to0}(1)$, $\|G_{0,1,-}^{(\kappa)}(|\phi\rangle\langle\phi|)\|_\cQ=\cO_{\alpha\to0}(1)$, $\|G_{1,0}(\kappa\nur)\|_{\gS_2}=\cO_{\alpha\to0}(1)$, and $\||D^0|G_{0,1,-,r}^{(\kappa)}(|\phi\rangle\langle\phi|)\|_{\gS_2}=\cO_{\alpha\to0}(1)$, we have
\begin{multline*}
 -\tr(|D^0|\{G_{1,0}(\kappa\nur),G_{0,1,-,r}^{(\kappa)}(|\phi\rangle\langle\phi|)\})+\tr(|D^0|\{\Qlinpr(\kappa\nur),G_{0,1,-}^{(\kappa)}(|\phi\rangle\langle\phi|)\})\\
+\tr(|D^0|\{\Qlinpr(\kappa\nur),G_{0,1,-,r}^{(\kappa)}(|\phi\rangle\langle\phi|)\})=\cO_{\alpha\to0}(1).
\end{multline*}
A short computation shows that
$$\tr(|D^0|\{G_{1,0}(\kappa\nur),G_{0,1}(|\phi\rangle\langle\phi|)\})=-\tr(G_{1,0}(\kappa\nur)R_{|\phi\rangle\langle\phi|}),$$
which implies
\begin{multline*}
|\tr(|D^0|\{G_{1,0}(\kappa\nur),G_{0,1}(|\phi\rangle\langle\phi|)\})|\le\|G_{1,0}(\kappa\nur)\|_{\gS_2}\|R_{|\phi\rangle\langle\phi|}\|_{\gS_2}\\
\le C\|\kappa\nur\|_\cC\|\phi\|_{H^1}^2=\cO_{\alpha\to0}(1). 
\end{multline*}
Finally, 
$$\tr(|D^0|\{\Qlinp(\kappa\nur),\alpha Q_3\})=\cO_{\alpha\to0}(\alpha).$$
The last term to estimate in \eqref{eq:dev-kinetic-energy} is thus $\langle\phi(\kappa\nur,\Lambda),\{|D^0|,Q_{-,r}(\alpha)\}\phi(\kappa\nur,\Lambda)\rangle$. We do so by writing 
$$|\langle\phi(\kappa\nur,\Lambda),\{|D^0|,Q_{-,r}(\alpha)\}\phi(\kappa\nur,\Lambda)\rangle|\le2\|Q_{-,r}(\alpha)\|_{H^1\to L^2}\|\phi(\kappa\nur,\Lambda)\|_{H^1}^2,$$
and it remains to estimate $\|Q_{-,r}(\alpha)|D^0|^{-1}\|_{L^2\to L^2}$. Recall that $Q_{-,r}(\alpha)=Q_-(\alpha)-\Qlinm(\kappa\nur)$, which we develop up to $\cO(\alpha)$. We have
$$Q_{-,r}(\alpha)=\alpha G_{0,1,-}^{(\kappa)}(G_{1,0}(\kappa\nur))+\alpha G_{1,0}(\rho_-(\alpha)-B_\Lambda\kappa\nur)+\cO^\cQ_{\alpha\to0}(\alpha),$$
and writing for any $Q$ that
\begin{multline*}
  G_{0,1,-}^{(\kappa)}(Q)|D^0|^{-1}=\frac{1}{2\pi}\int_\R\frac{1}{D^{\kappa\nur}_\Lambda-\mu_-+i\eta}\Pi_\Lambda R_Q|D^0|^{-1}\Pi_\Lambda\frac{1}{D^0_\Lambda-\mu_-+i\eta}\d{\eta}\\
+\frac{1}{2\pi}\int_\R\frac{1}{D^{\kappa\nur}_\Lambda-\mu_-+i\eta}\Pi_\Lambda R_Q\Pi_\Lambda\frac{1}{D^0_\Lambda-\mu_-+i\eta}\Pi_\Lambda V_{\kappa\nur}\Pi_\Lambda\frac{1}{D^{\kappa\nur}_\Lambda-\mu_-+i\eta}|D^0|^{-1}\d{\eta},
\end{multline*}
one obtains
$$\|G_{0,1,-}^{(\kappa)}(Q)|D^0|^{-1}\|_{L^2\to L^2}\le C\|R_Q|D^0|^{-1}\|_{L^2\to L^2}=\|R_Q\|_{H^1\to L^2}\le2\|Q\|_{\gS_2}$$
by the Hardy inequality. Together with $\|G_{1,0}(\kappa\nur)\|_{\gS_2}=\cO_{\alpha\to0}(1)$, this implies 
$$\|G_{0,1,-}^{(\kappa)}(G_{1,0}(\kappa\nur))|D^0|^{-1}\|_{L^2\to L^2}=\cO_{\alpha\to0}(1).$$
We also have $\|G_{1,0}(\rho)\|_{L^2\to L^2}\le\|G_{1,0}(\rho)\|_{\gS_2}\le C\|\rho\|_\cC$ for any $\rho$, and thus 
\begin{multline*}
\|G_{1,0}(\rho_-(\alpha)-B_\Lambda\kappa\nur)|D^0|^{-1}\|_{L^2\to L^2}\le\|G_{1,0}(\rho_-(\alpha)-B_\Lambda\kappa\nur)\|_{L^2\to L^2}\\
\le C\|\rho_-(\alpha)-B_\Lambda\kappa\nur\|_\cC=\cO_{\alpha\to0}(1). 
\end{multline*}
As a consequence,
$$|\langle\phi(\kappa\nur,\Lambda),\{|D^0|,Q_{-,r}(\alpha)\}\phi(\kappa\nur,\Lambda)\rangle|=\cO_{\alpha\to0}(\alpha),$$
and hence
\begin{multline}
 \tr_0(D^0Q_+(\alpha))-\tr_0(D^0Q_-(\alpha))=\left\langle\phi(\kappa\nur,\Lambda),D^0\phi(\kappa\nur,\Lambda)\right\rangle\\
-\frac{\alpha B_\Lambda}{1+\alpha B_\Lambda}\langle\phi(\kappa\nur,\Lambda),V_{\kappa\nur}\phi(\kappa\nur,\Lambda)\rangle+\cO_{\alpha\to0}(\alpha).  
\end{multline}

\bibliographystyle{siam}

\end{document}